\documentclass[a4paper, 11pt, onecolumn]{article}

\usepackage[utf8]{inputenc}
\usepackage[T1]{fontenc}
\usepackage{graphicx}
\usepackage[english]{babel}
\usepackage{textcomp}
\usepackage{layout}
\usepackage[top=2.2cm, bottom=2.2cm, left=2.8cm, right=2.8cm]{geometry}
\usepackage{setspace}
\usepackage{listings}
\usepackage{color}
\usepackage[dvipsnames]{xcolor}
\usepackage{lscape}
\usepackage{csquotes}
\usepackage{braket}
\usepackage{amsmath}
\usepackage{amssymb}
\usepackage{amsthm}
\usepackage{thmtools}
\usepackage{amsfonts}
\usepackage{mathrsfs}
\usepackage{xfrac}
\usepackage{faktor}
\usepackage{esint}
\usepackage{mathtools}
\usepackage{stmaryrd}
\usepackage{array}
\usepackage{bbm}
\usepackage{hyperref}
\hypersetup{
    colorlinks=true,
    linkcolor=red,
    filecolor=magenta,      
    urlcolor=blue,
    citecolor=BlueViolet,
    }
\usepackage{bbold}
\usepackage{algorithm}
\usepackage{algpseudocode}

\usepackage[backend = bibtex, style=alphabetic, giveninits=true]{biblatex}
\usepackage{csquotes}

\usepackage{tikz}
\usetikzlibrary{shadows}

\renewbibmacro{in:}{} %supprime le in dans les citations d'articles

\AtEveryBibitem{\clearfield{month}}
\AtEveryBibitem{\clearfield{day}}
\AtEveryBibitem{\clearfield{eprinttype}}
\AtEveryBibitem{\clearfield{eprint}}
\AtEveryBibitem{\clearfield{url}}
\AtEveryBibitem{\clearfield{urlyear}}
\AtEveryBibitem{\clearfield{review}}
\AtEveryBibitem{\clearfield{isbn}}
\AtEveryBibitem{\clearfield{issn}}
\AtEveryBibitem{\clearfield{pagetotal}}
\AtEveryBibitem{\clearfield{doi}}
\AtEveryBibitem{\clearfield{note}}
\AtEveryBibitem{\clearname{editor}}
\AtEveryBibitem{\clearlist{location}}
\AtEveryBibitem{\clearfield{number}}
\AtEveryBibitem{\clearfield{volume}}
\AtEveryBibitem{\clearfield{pages}}

\newtheorem{theorem}{Theorem}[section]
\newtheorem{corollary}{Corollary}[section]
\newtheorem{lemma}{Lemma}[section]
\newtheorem{proposition}{Proposition}[section]
\newtheorem{definition}{Definition}[section]

\theoremstyle{remark}
\newtheorem{remark}{Remark}

\newenvironment{nalign}{
    \begin{equation}
    \begin{aligned}
}{
    \end{aligned}
    \end{equation}
    \ignorespacesafterend
}

\newcommand{\subsecnonewline}[1]{\vspace{0.5cm}\noindent\textbf{\large{#1.}}}

\renewcommand{\le}{\leqslant}
\renewcommand{\ge}{\geqslant}
\renewcommand{\epsilon}{\varepsilon}

\newcommand{\qqtqq}[1]{\qquad \text{#1} \qquad}

\newcommand{\qqtq}[1]{\qquad \text{#1} \quad}
\newcommand{\qtq}[1]{\quad \text{#1} \quad}

\newcommand{\R}{\ensuremath{\mathbb{R}}}

\newcommand{\N}{\ensuremath{\mathbb{N}}}
\newcommand{\Q}{\ensuremath{\mathbb{Q}}}
\newcommand{\Z}{\ensuremath{\mathbb{Z}}}
\newcommand{\C}{\ensuremath{\mathbb{C}}}

\newcommand{\T}{\ensuremath{\mathbb{T}}}

\newcommand{\CA}{\mathcal{A}}
\newcommand{\CB}{\mathcal{B}}
\newcommand{\CC}{\mathcal{C}}

\newcommand{\CE}{\mathcal{E}}

\newcommand{\CH}{\mathcal{H}}
\newcommand{\CL}{\mathcal{L}}
\newcommand{\CO}{\mathcal{O}}
\newcommand{\CP}{\mathcal{P}}
\newcommand{\CQ}{\mathcal{Q}}
\newcommand{\CR}{\mathcal{R}}
\newcommand{\CS}{\mathcal{S}}
\newcommand{\CU}{\mathcal{U}}

\newcommand{\CM}{\mathcal{M}}

\newcommand{\eint}[2]{\left\llbracket #1 ,#2 \right\rrbracket}

\newcommand{\floor}[1]{\left\lfloor #1 \right\rfloor}
\newcommand{\ceil}[1]{\left\lceil #1 \right\rceil}
\newcommand{\sprod}[2]{\left\langle #1, #2 \right\rangle }
\renewcommand{\set}[1]{ \left \lbrace #1 \right \rbrace  }
\newcommand{\norm}[1]{\left \Vert #1 \right \Vert}

\newcommand{\pardif}[2]{\frac{\partial#1}{\partial#2}}
\newcommand{\grad}{\vec{\nabla}}

\newcommand{\jap}[1]{\left \langle #1 \right \rangle }

\newcommand{\mean}[1]{\langle #1 \rangle}
\newcommand*\diff{\mathop{}\!\mathrm{d}} %élément différentiel droit
\newcommand{\poissonbra}[2]{\left\{#1, #2 \right\}}

\DeclareMathOperator{\tr}{tr}
\DeclareMathOperator{\diag}{diag}

\DeclareMathOperator{\supp}{supp}

\DeclareMathOperator{\dist}{dist}

\DeclareMathOperator{\Op}{Op}
\DeclareMathOperator{\Sp}{Sp}
\DeclareMathOperator{\GL}{GL}

\DeclareMathOperator{\Leb}{Leb}

\DeclareMathOperator{\Vol}{Vol}

\title{Half-delocalization for semiclassical measures of split higher-dimensional quantum cat maps}
\author{Vincent Ferrari-Dominguez}

\date{\vspace{-5ex}}

\begin{document}
\maketitle

\begin{abstract}
    We study quantum cat maps on multidimensional tori and prove an optimal lower bound for the Kolmogorov-Sinai entropy of their semiclassical measures under the hypothesis that the dynamics splits over the rationals. This is the first instance of such a bound for a dynamical system with multiple positive Lyapunov exponents. The proof is based on a new notion of conditional quantum dynamical entropy and a careful study of the multiple Ehrenfest time adapted to the dynamics. 
\end{abstract}

\section{Introduction}

How does the dynamics of the geodesic flow on a compact Riemannian manifold $(M,g)$ influence the behavior of the eigenvectors of the Laplacian $-\Delta_g$ on $M$ in the high-eigenvalue limit? The quantum ergodicity theorem \cite{snirelcprimemanErgodicPropertiesEigenfunctions1974,colindeverdiereErgodiciteFonctionsPropres1985,zelditchUniformDistributionEigenfunctions1987b} answers this question by proving that generic sequences of eigenvectors equidistribute if the geodesic flow on the manifold is ergodic. This equidistribution statement can be refined by considering a microlocal lift to the cotangent bundle $T^*M$ of the measure $|\psi|^2 \diff \Vol_M$ defined on $M$, where $\psi \in L^2(M)$. Such a microlocal lift $Wig_h(\psi)$ is a distribution on $T^*M$ given by 
\begin{equation*}
    Wig_h(\psi)(a) = \sprod{\Op_h(a) \psi}{\psi} \qqtq{for} a \in \CC^\infty_c(T^*M),
\end{equation*}
where $0<h \le 1$ is a semiclassical parameter and $\Op_h : \CC_c^\infty(T^* M) \to \CB(L^2(M))$ is a quantization procedure that associates to each classical observable a bounded operator on $L^2(M)$. Intuitively, the distribution $Wig_h(\psi)$ synthesizes the behavior of $\psi$ in both position and frequency, with the frequency being rescaled by a factor $1/h$. The main object of interest of this article, semiclassical measures, are the accumulation points of sequences of distributions $(Wig_h(\psi_h))_h$. They inform us about the high-frequency behavior of the sequence $(\psi_h)_h$ at each point of $M$. Using this notion, the quantum ergodicity theorem can be rephrased as: if the geodesic flow on $M$ is ergodic, then the only semiclassical measure of a generic sequence of eigenvectors of $-\Delta$ whose eigenvalue diverges is the restriction of the Liouville measure, a natural measure on $T^*M$ that projects to the volume measure on $M$, to the cosphere bundle $S^*M$. We refer the reader to \cite[chap. 5]{zworskiSemiclassicalAnalysis2012} for an introduction to semiclassical measures and warn them that we will study in this article a similar notion where $T^*M$ is replaced by the torus $\T^{2d}$.

Trying to understand what could happen for non-generic sequences led Rudnick and Sarnak \cite{rudnickBehaviourEigenstatesArithmetic1994} to conjecture that on compact manifolds of negative curvature the only semiclassical measure is the Liouville measure, a phenomenon called unique quantum ergodicity. The conjecture was proved by Lindenstrauss \cite{lindenstraussInvariantMeasuresArithmetic2006} in an arithmetic setting, namely by only considering the joint eigenvectors of the Laplacian and Hecke operators on arithmetic surfaces. The general case remains however open and a strategy to attack it consists of proving stronger and stronger lower bounds for the Kolmogorov-Sinai entropy of the semiclassical measures. We highlight that the proof of the Arithmetic Unique Quantum Ergodicity by Lindenstrauss already uses crucially a lower bound on the entropy that he proved with Bourgain \cite{bourgainEntropyQuantumLimits2003}.

This line of attack was pursued by Anantharaman \cite{anantharamanEntropyLocalizationEigenfunctions2008} who proved that the Kolmogorov-Sinai entropy of semiclassical measures is always positive on compact manifolds of negative curvature. Motivated by their proof of an analog bound for the Walsh-quantized baker's map, a toy model for quantum chaos, Anantharaman and Nonnenmacher \cite{anantharamanEntropySemiclassicalMeasures2007} conjectured that for any semiclassical measure $\mu$ of the Laplacian,
\begin{equation}\label{conj_lower_bound_entropy_half_ruelle_bound}
    h_{KS}(\mu) \ge \frac{1}{2} \int_{S^*M} |\log(J^u(x))| \diff \mu(x),
\end{equation}
where $h_{KS}$ is the Kolmogorov-Sinai entropy of $\mu$, $S^*M$ is the cosphere bundle and $J^u$ is the unstable Jacobian of the geodesic flow on $S^*M$. Note that this lower bound is exactly half of the upper bound on the entropy given by the Ruelle inequality \cite[S.2.13]{katokIntroductionModernTheory1995}. Moreover, the unique quantum ergodicity conjecture is equivalent to proving the bound \eqref{conj_lower_bound_entropy_half_ruelle_bound} with $1/2$ replaced by $1$, by a result of Ledrappier and Young \cite{ledrappierMetricEntropyDiffeomorphisms1985} on the equality case in the Ruelle inequality. Adapting the techniques developed for the Walsh-quantized baker's map in \cite{anantharamanEntropySemiclassicalMeasures2007}, they later settled their conjecture in the case of compact hyperbolic manifolds in \cite{anantharamanHalfdelocalizationEigenfunctionsLaplacian2007a} while also proving a lower bound on compact negatively curved manifolds, which they improved with Koch \cite{anantharamanEntropyEigenfunctions2009}. All these results are valid as long as the geodesic flow is Anosov, which holds on manifolds with negative sectional curvature \cite[thm. 17.6.2]{katokIntroductionModernTheory1995}. Rivière extended their method to prove the bound \eqref{conj_lower_bound_entropy_half_ruelle_bound} for compact surfaces with negative curvature \cite{riviereEntropySemiclassicalMeasures2010c} and later only assuming nonpositive sectional curvature \cite{riviereEntropySemiclassicalMeasuresNonPositivelyCurvedSurface2010d}. Anantharaman and Silberman \cite{anantharamanHaarComponentQuantum2013a} also obtained a positive lower bound on the Kolmogorov-Sinai entropy on locally symmetric spaces of non-positive curvature. 

Entropy lower bound were also studied for another standard example of chaotic maps on the torus: hyperbolic symplectic integer matrices acting on the torus, the so-called cat maps. To build a toy model for quantum chaos, Hannay and Berry \cite{hannayQuantizationLinearMaps1980a} constructed a quantum equivalent for them, now called quantum cat maps. As they have been extensively studied, we refer the reader to the later Section \ref{subsec_review_literature_quantum_cat_maps} for a short review of the literature on them and concentrate here on the entropy bounds for their semiclassical measures. Let $A \in \Sp_{2d}(\Z)$ be a symplectic integer matrix. The analog of the bound \eqref{conj_lower_bound_entropy_half_ruelle_bound} for a semiclassical measure $\mu$ of $A$ is 
\begin{equation}\label{conj_lower_bound_entropy_half_ruelle_bound_torus}
    h_{KS}(\mu) \ge \frac{1}{2} \sum_{\lambda_i \ge 0} \lambda_i = \frac{1}{2}h_{KS}(\Leb),
\end{equation}
where $\Leb$ is the Lebesgue measure on the torus and $\lambda_i$ are the logarithm of the modulus of the eigenvalues of $A$ counted with multiplicity. The $(\lambda_i)_i$ are the Lyapunov exponents of the matrix $A$ and measure the hyperbolicity of the associated dynamical system. This lower bound is known to be optimal in dimension $2$ thanks to an example of De Bièvre, Faure and Nonnenmacher \cite{faureScarredEigenstatesQuantum2003} who proved that $(\Leb + \delta_0)/2$ is a semiclassical measure for any hyperbolic matrix in $\Sp_{2}(\Z)$. It is also optimal in higher dimensions since $\Leb_{x} \otimes \delta_\xi$ is a semiclassical measure of any block-diagonal matrix $\tilde{A} \oplus \tilde{A}^{-T}$ where $\tilde{A} \in \GL_{d}(\Z)$, see \cite[discussion after thm. 1]{kelmerArithmeticQuantumUnique2010a}.

Adapting the method for the Laplacian on manifolds, Rivière \cite{riviereEntropySemiclassicalMeasuresSymplecticMapTorus2011} proved the bound 
\begin{equation}\label{lower_bound_entropy_gabriel}
    h_{KS}(\mu) \ge \Lambda_0(A) := \sum_{\lambda_i \ge 0} \max(\lambda_i - \frac{\lambda_{max}}{2},0),
\end{equation}
where $\lambda_{max}$ is the largest Lyapunov exponent of $A$. The presence of $\lambda_{max}$ comes from its link with the Ehrenfest time and renders the bound weaker than the conjectured one \eqref{conj_lower_bound_entropy_half_ruelle_bound_torus}. The Ehrenfest time can be thought of as the largest time for which the correspondence between the classical and quantum dynamics holds, or more precisely for which a uniform Egorov theorem holds. The Ehrenfest time thus imposes a limit on the timescale at which semiclassical techniques are effective. As the Ehrenfest time for a generic classical observable is $|\log h|/(2\lambda_{max})$, if there are multiple positive Lyapunov exponents, one cannot use large enough times to capture all the entropy coming from all the Lyapunov exponents in the proofs of \cite{anantharamanHalfdelocalizationEigenfunctionsLaplacian2007a} or \cite{riviereEntropySemiclassicalMeasuresSymplecticMapTorus2011}. The main conceptual contribution of this paper is to modify the previous techniques to circumvent this problem and capture the contribution of all the Lyapunov exponents. This is done by constructing a conditional quantum dynamical entropy, see Section \ref{sec_entropy}, and by building observables that have an effective Ehrenfest time tailored to each Lyapunov exponent, see Section \ref{sec_effective_ehrenfest_time}.

\subsection{Main results}
The first result of this paper is a proof of the conjectured bound \eqref{conj_lower_bound_entropy_half_ruelle_bound_torus} under the hypothesis that the dynamics of the symplectic matrix $A$ splits over $\Q$. 

\begin{theorem}\label{thm_entropy_bound_split_torus_one_lyap_exp}    
    Let $A \in \Sp_{2d}(\Z)$ be a symplectic integer matrix. Suppose that $\Q^{2d}$ splits into rational symplectic $A$-invariant subspaces $(V_i)_{i \in \eint{1}{r}}$ of respective dimension $2d_i$
    \begin{equation*}
        \Q^{2d} = V_1 \oplus^{\sigma} \dots \oplus^{\sigma} V_r \qqtqq{and} A V_i = V_i.
    \end{equation*}
    Suppose that the restriction of $A$ to $V_i$ has only one positive Lyapunov exponent. Denote by $(\lambda_i)_{i \in I}$ the Lyapunov exponents of $A$ counted with multiplicity. Then any semiclassical measure $\mu$ of $A$ satisfies 
    \begin{equation*}
        h_{KS}(\mu, A) \ge \frac{h_{KS}(\Leb,A)}{2} = \frac{1}{2} \sum_{\lambda_i \ge 0} \lambda_i.
    \end{equation*}
\end{theorem}

This theorem is, to the best of the author's knowledge, the first proof of an analog of the conjecture \eqref{conj_lower_bound_entropy_half_ruelle_bound} in a case where there are multiple positive Lyapunov exponents. 

\begin{remark}\label{rem_thm_entropy_bound_split_torus_one_lyap_exp_powers_of_A}
    Note that any semiclassical measure $\mu$ of $A$ is also a semiclassical measure of $A^m$ for any $m \ge 1$ and that $h_{KS}(\mu,A^m) = m h_{KS}(\mu,A)$. Since the Lyapunov exponents of $A^m$ are the ones of $A$ multiplied by $m$, it is straightforward to see that the conclusion of Theorem \ref{thm_entropy_bound_split_torus_one_lyap_exp} holds if $\Q^{2d}$ splits into rational symplectic $A^m$-invariant subspaces such that the restriction of $A^m$ to each of them has only one positive Lyapunov exponent. 
\end{remark}

\begin{remark}
    Given two symplectic matrices $A_1$ and $A_2$ acting on $\T^{2d_1}$, resp. $\T^{2d_2}$, we build a matrix $A_1 \diamond A_2$ that acts diagonally on $\T^{2(d_1+d_2)} \cong \T^{2d_1} \times \T^{2d_2}$, see \eqref{eq_definition_diamond_symplectic_matrices} for a precise formula. One could hope that the semiclassical measures associated to the product $A_1 \diamond A_2$ would be of the form $\mu_1 \otimes \mu_2$, with $\mu_i$ a semiclassical measure associated to $A_i$. As the entropy of a product measure satisfies 
    \begin{equation*}
        h_{KS}(\mu_1 \otimes \mu_2, A_1 \diamond A_2) = h_{KS}(\mu_1, A_1) + h_{KS}(\mu_2, A_2).
    \end{equation*}
    Theorem \ref{thm_entropy_bound_split_torus_one_lyap_exp} in the simpler case of products of symplectic matrices would just be a direct application of the lower bound \eqref{lower_bound_entropy_gabriel} obtained by Rivière, which is optimal if the matrix has only one positive Lyapunov exponent. 
    
    It is however not the case that semiclassical measures of product of matrices always have a product structure as shown by an example of Kim \cite[prop. 6.1]{kimCharacterizingSupportSemiclassical2025a}. She proves that 
    \begin{equation*}
        \frac{1}{2} \left(\delta(x_1, \xi_1) \otimes \Leb_{x_2, \xi_2} + \Leb_{x_1, \xi_1}\otimes \delta(x_2, \xi_2) \right),
    \end{equation*}
    is a semiclassical measure for $A \diamond A$, where $A = \begin{pmatrix}
    2 & 1 \\ 1 & 1 \end{pmatrix}$ is the Arnold's cat map.
\end{remark}

Theorem \ref{thm_entropy_bound_split_torus_one_lyap_exp} is deduced from the more general theorem below. 
\begin{theorem}\label{thm_lower_bound_full_entropy}
    Let $A \in \Sp_{2d}(\Z)$. Suppose that $\Q^{2d}$ splits into rational symplectic $A$-invariant subspaces $(V_i)_{i \in \eint{1}{r}}$, meaning that
\begin{equation}\label{eq_thm_lower_bound_full_entropy_splitting}
    \Q^{2d} = V_1 \oplus^{\sigma} \dots \oplus^{\sigma} V_r \qqtqq{and} A V_i = V_i.
\end{equation}
Denote by $\lambda_i$ the maximal Lyapunov exponent of $A$ restricted to $V_i$ and suppose that 
\begin{equation*}
    0 < \lambda_1 < \lambda_2 < \cdots < \lambda_{r}.
\end{equation*}
Denote $\hat{V}_i = V_1 \oplus \dots \oplus V_i$ and $\hat{\pi}_i : \T^{2d} \to \T^{2d}/\hat{V}_i^\sigma$ the natural projection.
Introduce the numbers $\nu_{i,j}$ and $B_j$ for $1 \le i \le j \le r$ defined by 
\begin{equation*}
    \nu_{i,j} = \min(\frac{\lambda_j}{\lambda_i}-1, 1) \qtq{and} B_j = \Lambda_0(A_{|\hat{V}_j}) + \sum_{i=1}^{j-1} (\nu_{i,j}-\nu_{i+1,j}) B_{i} > 0.
\end{equation*}
Then for any semiclassical measure $\mu$ of $A$ and any $j \in \eint{1}{r}$, 
\begin{equation*}
    h_{KS}(\hat{\pi}_{j\#} \mu, A) \ge B_j .
\end{equation*}
\end{theorem}

To interpret the theorem, suppose that there are only two subspaces $V_1$ and $V_2$ in the decomposition \eqref{eq_thm_lower_bound_full_entropy_splitting}. Theorem \ref{thm_lower_bound_full_entropy} then gives the lower bound
\begin{equation*}
    h_{KS}(\mu) \ge \Lambda_0(A) + \nu_{1,2}\Lambda_0(A_{|V_1}),
\end{equation*}
as $\nu_{i,i} = 0$.
In a sense made precise in Section \ref{sec_effective_ehrenfest_time}, a function that depends only on the direction $V_i \otimes \R$ has an effective Ehrenfest time of order $T_i = |\log h|/(2\lambda_i)$. Therefore, the quotient $\lambda_2/\lambda_1$ is the quotient of the Ehrenfest times $T_1/T_2$. It measures, in proportion to $T_2$, for how much longer the classical-quantum correspondence holds for a function on $\hat{V}_2 \otimes \R$ which only depends on the directions $\hat{V}_1 \otimes \R$. With this in mind, Theorem \ref{thm_lower_bound_full_entropy} tells us that we can improve on the bound \eqref{lower_bound_entropy_gabriel}, whose arguments go only to the Ehrenfest time $T_2$, by recovering the entropy we lost by not pushing the arguments to the longer time Ehrenfest time $T_1$ in the direction $V_1$. The central tool to merge the contribution of each direction is the conditional quantum dynamical entropy defined in Section \ref{sec_entropy}.

Note that by following this reasoning, one would expect that we choose $\nu_{1,2} = \lambda_2/\lambda_1 -1 = T_1/T_2 -1 $ instead of the expression in Theorem \ref{thm_lower_bound_full_entropy}. Indeed, the observable in the direction $V_1$ was already evolved for a time $T_2$ so we only need to evolve them for a fraction $T_1/T_2 -1$ of the time $T_2$. The minimum with $1$ in the definition of $\nu_{i,j}$ is there to ensure that we do not evolve the observable for too long so that the commutators remain small, a property needed for the almost subadditivity of our conditional quantum entropy. Since this choice leads to the optimal bound \eqref{conj_lower_bound_entropy_half_ruelle_bound_torus} in Theorem \ref{thm_entropy_bound_split_torus_one_lyap_exp}, we suspect that this choice for $\nu_{i,j}$ is optimal.

\begin{remark}
    The bound of Theorem \ref{thm_lower_bound_full_entropy} recovers the one obtained by Rivière \eqref{lower_bound_entropy_gabriel} as $B_1 = \Lambda_0(A)$ if we take $V_1 = \Q^{2d}$ in the decomposition \eqref{eq_thm_lower_bound_full_entropy_splitting}. It is an improvement as long as there are at least two $A$-invariant rational symplectic subspaces such that the restriction of $A$ to any of them has a positive Lyapunov exponent.
\end{remark}

\begin{remark}
    Unsurprisingly, the existence of a non-trivial decomposition \eqref{eq_thm_lower_bound_full_entropy_splitting} is an extremely restrictive condition to impose on a symplectic matrix $A$. Indeed, the existence of a non-trivial $A$ rational invariant symplectic subspace implies that the characteristic polynomial of $A$ is reducible over $\Q$, which almost never happens by a result of Anderson and Lemmke Oliver \cite[thm. B.1]{kimCharacterizingSupportSemiclassical2025a}. They proved that there exists a subset $S \subset \Sp_{2d}(\Z)$ of density one in the integer symplectic matrices such that if $A \in S$, then $A^m$ has an irreducible characteristic polynomial over $\Q$ for any $m \ge 1$.
\end{remark}

\begin{remark}
    As in remark \ref{rem_thm_entropy_bound_split_torus_one_lyap_exp_powers_of_A}, if there exists a splitting of $\Q^{2d}$ by symplectic $A^m$-invariant subspaces with $m \ge 1$, the lower bound for the entropy of semiclassical measures of $A^m$ obtained from Theorem \ref{thm_lower_bound_full_entropy} applied to this splitting leads to a lower bound for the entropy of semiclassical measures of $A$.
\end{remark}

The last theorem of this paper is a version of the bound \eqref{lower_bound_entropy_gabriel} obtained by Rivière for the projection of semiclassical measures on an $A$-invariant symplectic subtorus.  
\begin{theorem}\label{thm_lower_bound_entropy_projection}
    Let $A \in \Sp_{2d}(\Z)$. Suppose that there exists a rational symplectic $A$-invariant subspace $V$. Then any semiclassical measure $\mu$ of $A$ satisfies
    \begin{equation*}
        h_{KS}(\pi_\# \mu, A) \ge \Lambda_0(A_{|V}),
    \end{equation*}
    where $\pi : \T^{2d} \to \T^{2d}/V^\sigma$ is the natural projection. 
\end{theorem}

The proof of this theorem is almost identical to the standard one \cite{anantharamanEntropySemiclassicalMeasures2007,anantharamanHalfdelocalizationEigenfunctionsLaplacian2007a,riviereEntropySemiclassicalMeasuresSymplecticMapTorus2011}, except that we consider observables that depend only on the direction $V \otimes \R$ and thus need to adapt the hyperbolic dispersive estimate \ref{cor_applied_dispersive_estimate} to allow such observables. It serves as a reminder, or quick presentation, of the standard reasoning to prove lower bounds on the entropy of semiclassical measures, before delving into the proof of Theorem \ref{thm_lower_bound_full_entropy}. 

We also deduce as a corollary an extension of a result of Kim \cite[thm. 1.3]{kimCharacterizingSupportSemiclassical2025a} that was obtained by fractal uncertainty principle techniques. 

\begin{corollary}\label{cor_no_scarring_translate_symplectic_tori}
    Let $A \in \Sp_{2d}(\Z)$ and $V$ a rational symplectic $A$-invariant subspace. Let $\Sigma$ be the subtorus of $\T^{2d}$ with tangent space $V$ and $K$ a subset of the subtorus of $\T^{2d}$ with tangent space $V^\sigma$. Suppose that there exists a semiclassical measure $\mu$ for $A$ supported in $\Sigma \times K$, then 
    \begin{equation*}
        h_{top}(K, A) \ge h_{KS}(\pi_{\#}\mu,A) \ge \Lambda_0(A_{|V^\sigma}),
    \end{equation*}
    where $\pi : \T^{2d} \to \T^{2d} /V$ is the natural projection. In particular, if $A_{|V^\sigma}$ has a positive Lyapunov exponent, then $K$ cannot be a finite set.
\end{corollary}

This corollary is deduced from Theorem \ref{thm_lower_bound_entropy_projection}, the variational principle \cite[thm. 4.5.3]{katokIntroductionModernTheory1995} and the fact that the topological entropy of a finite set is always zero. 

\begin{remark}
    Compared to the result of Kim, this corollary removes the diagonalizability condition for the matrix $A$ at the cost of only considering symplectic subspaces. We also suspect, as Kim does, that the corollary should hold as long as there exists a rational $A_{|V^\sigma}$-invariant complement to $V \cap V^\sigma$ inside $V^\sigma$, where $V^\sigma$ is the symplectic orthogonal of $V$. To prove it with entropy techniques would require a version of the hyperbolic dispersive estimate \ref{cor_applied_dispersive_estimate} that does not impose that $V$ be symplectic. Note that Corollary \ref{cor_no_scarring_translate_symplectic_tori} works if there exists a symplectic rational $A_{|V^\sigma}$-invariant complement to $V \cap V^\sigma$ inside $V^\sigma$.
\end{remark}

\subsection{A panorama on quantum cat maps}\label{subsec_review_literature_quantum_cat_maps}

We do not claim exhaustiveness in this review as the goal is to give a panoramic view of the results and directions of research on quantum cat maps that were first introduced by Hannay and Berry \cite{hannayQuantizationLinearMaps1980a} in dimension $2$. Keating studied the distribution of their periods \cite{keatingAsymptoticPropertiesPeriodic1991}, which he found to be quite erratic and related to number theoretic considerations, and the distribution of their eigenvalues \cite{keatingCatMapsQuantum1991} and proved in particular an exact Gutzwiller trace formula for them. The analog of the quantum ergodicity theorem for quantum cat maps was obtained by Bouzouina and De Bièvre \cite{bouzouinaEquipartitionEigenfunctionsQuantized1996a} and a rate for quantum ergodicity was obtained by Schubert \cite{schubertRateQuantumErgodicity2008} adapting similar techniques \cite{schubertUpperBoundsRate2006a} he used to study the same question for the Laplacian. 

\subsubsection{Scarring phenomenon}

To investigate the Ehrenfest time of quantum cat maps, Bonechi and De Bièvre studied the evolution of coherent states. Denoting by $\lambda$ the Lyapunov exponent of the classical map, they found that the Husimi distribution of a coherent state evolved for a time smaller than $|\log h|/2\lambda$ converges to a Dirac mass, while it equidistributes when evolved for times between $|\log h|/2\lambda$ and $|\log h|/\lambda$. Building on this insight and using short periods of the quantum cat maps, De Bièvre, Faure and Nonnenmacher \cite{faureScarredEigenstatesQuantum2003} produced a counterexample to the unique quantum ergodicity and proved that any measure of the form $(\Leb + \delta_0)/2$ is a semiclassical measure for $A$ if $A$ is hyperbolic. The singular behavior of this semiclassical measure at $0$ is called a scar. Bonechi and De Bièvre \cite{bonechiControllingStrongScarring2003a} proved a bound on the mass of the atomic part of semiclassical measures thus limiting the scarring phenomenon. This bound was later made optimal by Faure and Nonnenmacher \cite{faureMaximalScarringQuantum2004}. 

Another example of scarring called super-scars is known in higher dimensions \cite{gurevichWeilRepresentationDeligne2006,kelmerArithmeticQuantumUnique2010a}: if there exists an isotropic rational $A$ invariant subspace of dimension $d_1$, then there exists a subtorus $X$ of co-dimension $d_1$ such that the Lebesgue measure restricted to $X$ is a semiclassical measure for $A$. Kelmer proved that such a phenomenon endures when considering perturbations of quantum cat maps \cite{kelmerScarringInvariantManifolds2007a} or Hecke eigenvectors \cite{kelmerArithmeticQuantumUnique2010a}.

\subsubsection{Constraints on semiclassical measures: full support and entropy lower bound}

While scarring is possible in dimension $2$, Schwartz \cite{schwartzFullDelocalizationEigenstates2024a} showed that any semiclassical measure of a hyperbolic matrix $A \in \Sp_2(\Z)$ has full support. His work was based on that of Dyatlov and Jin \cite{dyatlovSemiclassicalMeasuresHyperbolic2018}, who proved that all semiclassical measures for the Laplacian on hyperbolic surfaces have full support using a fractal uncertainty principle of Bourgain and Dyatlov \cite{bourgainSpectralGapsPressure2018}, a result they extended with Nonnenmacher \cite{dyatlovControlEigenfunctionsSurfaces2022} to variable negative curvature. 

The full support of semiclassical measures for quantum cat maps was extended by Dyatlov and Jezequel \cite{dyatlovSemiclassicalMeasuresHigherdimensional2024} to higher dimensions under the hypothesis that $A$ has a simple leading eigenvalue and an irreducible characteristic polynomial over $\Q$. Kim \cite{kimCharacterizingSupportSemiclassical2025a}, using a higher-dimensional version of the fractal uncertainty principle developed by Cohen \cite{cohenFractalUncertaintyHigher2024}, later removed the simple leading eigenvalue hypothesis, while requiring at least one non-unit eigenvalue for $A$ and the irreducibility of the characteristic polynomial of any power of $A$. These last two results have a counterpart for the Laplacian on negatively curved manifolds: Athreya, Dyatlov and Miller \cite{athreyaSemiclassicalMeasuresComplex2024} obtained some constraints on the support of semiclassical measures on complex hyperbolic quotients, while Kim and Miller \cite{kimSemiclassicalMeasuresHyperbolic2025} proved similar results on hyperbolic manifolds of any dimension. 

Another constraint on semiclassical measures is the topic of this article, namely lower bounds on their Kolmogorov-Sinai entropy. As already mentioned, the first bound for quantum maps is due to Anantharaman and Nonnenmacher who proved the analog of \eqref{conj_lower_bound_entropy_half_ruelle_bound} for the Walsh-quantized baker's map. For quantum cat maps, Rivière \cite{riviereEntropySemiclassicalMeasuresSymplecticMapTorus2011} proved the general bound \eqref{lower_bound_entropy_gabriel}, which is positive as long as $A$ admits a non-unit eigenvalue. A stronger result due to Brooks \cite{brooksEntropyQuantumLimits2010a} holds in dimension $2$ for hyperbolic matrices, namely that the ergodic components of a semiclassical measure of low entropy carry less mass than those of high entropy. This gives both the bound \eqref{conj_lower_bound_entropy_half_ruelle_bound_torus} and the one in \cite{faureMaximalScarringQuantum2004}. We mention that Gutkin \cite{gutkinEntropicBoundsSemiclassical2010} studied a class of piecewise affine maps of the interval, proved the analog of the conjecture \eqref{conj_lower_bound_entropy_half_ruelle_bound} for them and the optimality of the bound. 

\subsubsection{Arithmetic quantum unique ergodicity and Hecke eigenfunctions}

Similarly to the Hecke operators on arithmetic surfaces, Kurlberg and Rudnick \cite{kurlbergHeckeTheoryEquidistribution2000} introduced, given a symplectic matrix $A \in \Sp_{2}(\Z)$ satisfying some arithmetic conditions, Hecke operators that commute with the quantum cat map $M_N(A)$ associated to $A$ and proved arithmetic quantum unique ergodicity. More precisely, they showed that given any sequence of Hecke eigenfunctions, i.e. common eigenvectors of the Hecke operators and $M_N(A)$, the associated Wigner distributions converge in the semiclassical limit to the Lebesgue measure, with an explicit rate of convergence of order $N^{-1/4}$. They conjectured that the optimal rate was $N^{-1/2}$, which was proven by Gurevich and Hadani \cite{gurevichProofKurlbergRudnickRate2011} when $N$ is restricted to the prime numbers. Arithmetic quantum unique ergodicity for quantum cat maps was extended to higher dimensions by Kelmer \cite{kelmerArithmeticQuantumUnique2010a} who proved that the only obstruction is the super-scar phenomenon described before. 

The relation between the supremum norm of the Hecke eigenfunctions and the arithmetic properties of $A$ and $N$ was studied by Kurlberg and Rudnick \cite{kurlbergDistributionMatrixElements2005a}, Kurlberg \cite{kurlbergBoundsSupremumNorms2007} and Olofsson \cite{olofssonHeckeEigenfunctionsQuantized2009,olofssonLargeSupremumNorms2009,olofssonLargeNewformsQuantized2010a}. A lower and upper bound for the supremum norm of all eigenvectors, i.e. not only Hecke eigenfunctions, was obtained by Kim and Koirala \cite{kimBoundsEigenfunctionsQuantum2023} in dimension $2$.

Kurlberg and Rudnick \cite{kurlbergDistributionMatrixElements2005a} conjectured a limit distribution for the fluctuations of the Wigner distribution of Hecke eigenfunctions, the so-called matrix elements, and they justified their conjecture by a computation of the second and fourth moments. This question was also investigated by Kelmer \cite{kelmerMatrixElementsQuantized2008a} and Rosenzweig \cite{rosenzweigFluctuationsMatrixElements2011}. 

Rivière and Wolf \cite{riviereLebesgueComponentSemiclassical2025} considered another subclass of eigenvectors of quantum cat maps, namely joint eigenvectors of two quantum cat maps $M_N(A)$ and $M_N(B)$. They found that under some conditions on the Galois group of the characteristic polynomial of $A$, a good choice of $B$ ensured that the semiclassical measures of the joint eigenvectors always have a large Lebesgue component.

\subsubsection{Quantum unique ergodicity along a subsequence of semiclassical parameters}

The scarring example of De Bièvre, Faure and Nonnenmacher \cite{faureScarredEigenstatesQuantum2003} shows that quantum unique ergodicity does not hold for quantum cat maps. Since their example was built by using extremely particular values of $N$ for which the quantum cat map has short periods of order $\log(N)$, can we expect QUE to hold by restricting the values of the semiclassical parameter $N$? 

This question was first investigated by Degli Esposti, Graffi and Isola \cite{espostiClassicalLimitQuantized1995} who showed that when $N$ is restricted to prime numbers such that the period of the quantum cat map is small, then QUE holds. This result was extended by Kurlberg and Rudnick \cite{kurlbergQuantumErgodicityLinear2001} who proved that QUE holds along a subset of density one of $N$ for a hyperbolic matrix $A \in \Sp_2(\Z)$. The values of $N$ are chosen such that the period of $A$ modulo $N$ is larger than $\sqrt{N}$ and they also studied the fluctuations of matrix elements along the same sequence of $N$ \cite{kurlbergMatrixElementsQuantum2007a}. The QUE results were improved by Bourgain \cite{bourgainRemarkQuantumErgodicity2007} who provided a rate of convergence of order $N^{-\delta}$ and only required that the period of $A$ modulo $N$ be larger than $N^\epsilon$. Ostafe, Shparlinski and Voloch \cite{ostafeEquationsCharacterSums2023} showed that one can take $\delta = 1/60$ in Bourgain's rate of convergence. 

Quantum unique ergodicity along a sequence of $N$ of density one in all dimensions was recently proven by Kurlberg, Ostafe, Rudnick and Shparlinski \cite{kurlbergQuantumErgodicityHigher2025a} under some natural hypotheses on the characteristic polynomial of $A$, preventing in particular super-scars. A precise rate of convergence was proved by Bhakta and Shparlinski \cite{bhaktaQuantumErgodicityHigherdimensional2026} when $N$ is restricted to the powers of a well-chosen prime number, the set of which has positive density in the set of all primes. 

\subsecnonewline{Structure of the paper} The new notion of conditional quantum dynamical entropy is defined in Section \ref{sec_entropy}, its subadditive properties are investigated and an entropic uncertainty principle is proved for it. The quantization of symplectic linear maps on the torus is recalled in Section \ref{sec_quantization_torus}. We study the Ehrenfest time in Section \ref{sec_effective_ehrenfest_time} and introduce a notion of effective Ehrenfest time for a given observable. We show that this notion interacts well with the conditional quantum entropy. Section \ref{sec_hyperbolic_disp_estimate} contains the proof of the needed hyperbolic dispersive estimate. All the tools are combined in Section \ref{sec_full_proofs} to prove Theorems \ref{thm_lower_bound_entropy_projection},\ref{thm_lower_bound_full_entropy} and \ref{thm_entropy_bound_split_torus_one_lyap_exp}.

\subsecnonewline{Conventions and notations} We fix our conventions and give some notations for the symplectic geometry on $\R^{2d}$ and $\T^{2d}$. The spaces $\R^{2d}$ and $\T^{2d}$ are endowed with a symplectic structure given by the natural symplectic form $\sigma$ defined as
\begin{equation*}
    \sigma(z,z') = \sprod{z}{Jz'} \qtq{for} z, z' \in \R^{2d} \text{ or } \T^{2d} \qqtqq{where} J = \begin{pmatrix*}
        O & - I_d \\ I_d & 0         
    \end{pmatrix*}.
\end{equation*}
The symplectic orthogonal of a vector subspace $V$ is denoted $V^\sigma$. 

Let $f : ]0, \epsilon[ \to E$ be a function defined near $0^+$ with values in a space $E$ endowed with a countable family of seminorms $(\norm{\cdot}_n)_n$. We say that $f(h) = \CO(h^\alpha)$ for an $\alpha \ge 0$ if 
\begin{equation*}
    \forall n \ge 0, \; \exists C_n \ge 0 , \quad \norm{f(h)}_n \le C_n h^\alpha
\end{equation*}
We denote $f(h) = \CO(h^\infty)$ when $f(h) = \CO(h^k)$ for all $k \ge 0$.

\subsecnonewline{Acknowledgements} I warmly thank my advisor Nalini Anantharaman for introducing me to this subject and for the numerous discussions about it, as well as her support during the writing of this article. I used a Large Language Model, Coder Gwen 3.8 27B NVFP4, only to correct typos and English mistakes, find references in the literature and review the results once they were written.

\section{Classical, quantum and conditional entropies}\label{sec_entropy}

In their proof of the bound \eqref{conj_lower_bound_entropy_half_ruelle_bound} on hyperbolic manifolds, Anantharaman and Nonnenmacher \cite{anantharamanHalfdelocalizationEigenfunctionsLaplacian2007a} introduce the entropy-like quantity 
\begin{equation}\label{eq_quantum_entropy_anantharaman_nonnenmacher}
    h(\psi, (P_i)_{i \in I}) = \sum_{i \in I} \eta(\norm{P_i \psi}^2),
\end{equation}
where $\psi$ is a normed vector in a Hilbert space $\CH$ and $(P_i)_{i \in I}$ is a family of operators on $\CH$ that splits the identity, meaning
\begin{equation}\label{eq_quantum_partition_hilbert_space}
    \sum_{i \in I} P_i^* P_i = Id_{\CH} .
\end{equation}
Such a family of operators is called a quantum partition and the relation \eqref{eq_quantum_partition_hilbert_space} implies that $(\norm{P_i \psi}^2)_{i \in I}$ is a probability measure on the set $I$. 

The goal of this section is to present a conditional version of this quantum entropy and extend it to $C^*$ algebras since it does not require much effort. We refer the reader to \cite[chap. 5]{alickiQuantumDynamicalSystems2001a} for a quick introduction to $C^*$ algebras centered on quantum dynamics and to \cite{davidsonC^algebrasExample1996} for a more standard reference. 

Before delving into the quantum world, we first skim through the classical one and recall the definition of the entropy of a partition, its properties and the definition of the Kolmogorov-Sinai entropy. Let $(X, \CA, \mu)$ be a measured space with $\mu$ a probability measure and $\CP = (p_i)_{i \in I}$ and $\CQ = (q_j)_{j \in J}$ two finite measurable partitions of unity of $X$, meaning that 
\begin{equation*}
    \sum_{i \in I} p_i = \sum_{j \in J} q_j = 1 \qqtqq{and} p_i \ge 0, \quad q_j \ge 0.
\end{equation*}
The join of the two partitions of unity $\CP$ and $\CQ$ is $\CP \vee \CQ = (p_i q_j)_{i,j}$. 
The entropy of $\CP$ for $\mu$ is defined by
\begin{equation*}
    h^{c}(\mu,\CP) = \sum_{i \in I} \eta(\mu(p_i)) \qqtq{where} \eta(x) = -x \log(x).
\end{equation*}
If $f:X \to \R^+$ is a non-negative measurable function on $X$ such that $\mu(f) >0$, the restriction of $\mu$ to $f$ is a probability measure given by $\mu_{|f}(g) = \mu(fg)/\mu(f)$ on any measurable function $g$ on $X$. The conditional entropy of $\CQ$ with respect to $\CP$ is defined by 
\begin{equation*}
    h^{c}(\mu, \CQ | \CP) = \sum_{i\in I} \mu(p_i) h^c(\mu_{|p_i}, \CQ) \qqtq{where} h^c(\mu_{|p_i}, \CQ) = 0 \qtq{if} \mu(p_i) = 0,
\end{equation*}
The entropy satisfies the following properties for any partitions of unity $\CP$ and $\CQ$,
\begin{equation}\label{relation_cond_classical_entropy}
    h^c(\mu, \CQ \vee \CP) = h^c(\mu,\CP) + h^c(\mu, \CQ|\CP ) \qtq{and} h^c(\mu, \CP | \CQ ) \le h^c(\mu, \CP).
\end{equation}
Those two properties imply that the entropy is subadditive, meaning that 
\begin{equation*}
    h^c(\mu, \CP \vee \CQ) \le h^c(\mu, \CP) + h^c(\mu, \CQ) .
\end{equation*}

To recover the usual definition of the entropy of a finite partition of $X$ in measurable sets $\CP = (A_i)_{i \in I}$, we can associate to it the partition of unity $\tilde{\CP} = (\mathbb{1}_{A_i})_{i \in I}$ given by the indicator functions of the sets $A_i$. It is straightforward to check that 
\begin{equation*}
    h^c(\mu, \tilde{\CP}) = \sum_{i \in I} \eta(\mu(A_i)),
\end{equation*}
and the right-hand side is the standard definition of the entropy of a partition. 

The Kolmogorov-Sinai entropy is an invariant of a measured dynamical system that, intuitively, informs us on the growth rate of the number of different trajectories of length $n$. Let $T : X \to X$ be a measurable map that preserves the measure $\mu$. The Kolmogorov-Sinai entropy of $(T, \mu)$ with respect to a finite partition of unity $\CP = (p_i)_{i\in I}$ of $X$ is given by 
\begin{equation*}
    h_{KS}(\mu,T, \CP) = \lim_{m \to \infty} \frac{1}{m} h^c(\mu, \CP \vee T^{-1} \CP \vee \dots \vee T^{-m+1} \CP),
\end{equation*}
where $T^{-k} \CP = (p_i \circ T^k)_{i \in I}$ for $k \in \N$.
The Kolmogorov-Sinai entropy of $(T, \mu)$ is then given by the supremum over all finite partitions:
\begin{equation*}
    h_{KS}(\mu,T) = \sup_{\CP \text{ finite partition}} h_{KS}(\mu, T, \CP).
\end{equation*}
We refer the reader to \cite[chap. 4]{katokIntroductionModernTheory1995} for an introduction to ergodic theory and the Kolmogorov-Sinai entropy.

\subsection{A conditional quantum dynamical entropy} 

Let $\CA$ be a unital $C^*$-algebra. The equivalent of a probability measure in this context is a positive continuous linear functional of norm $1$ on $\CA$. Denote by $\mathcal{E}$ the set of such functionals and call them states. Any state $\rho \in \CE$ satisfies a Cauchy-Schwarz inequality 
\begin{equation}
    |\rho(a^*b)|^2 \le \rho(a^*a) \rho(b^*b) \qqtq{for} a,b \in \CA.
\end{equation}
Moreover, a continuous linear functional $\rho$ on $\CA$ is a state if and only if $\rho(\mathbb{1}) = 1$ \cite[][lem. I.9.9]{davidsonC^algebrasExample1996}.

To build a conditional quantum entropy, we define a quantum analog of the restriction of a probability measure to a measurable set. Let $\rho \in \CE$ and $p \in \CA$ such that $\rho(p^*p) \neq 0$. The \textit{restriction of a state $\rho$ to $p$} is defined by 
\begin{equation*}
    \rho_{|p} : x \in \CA \mapsto \frac{\rho(p^* x p)}{\rho(p^* p)}.
\end{equation*}
It is clear that $\rho_{|p} \in \mathcal{E}$, which implies that $|\rho_{|p}(a)| \le \norm{a}$ for any $a \in \CA$. This last inequality is equivalent to 
\begin{equation}\label{ineq_restriction_state_C*algebra}
    |\rho(p^*ap)| \le \norm{a} \rho(p^* p) \qqtq{for} a \in \CA.
\end{equation}

A \textit{quantum partition} $\CP$ is a finite family of elements $(p_i)_{i \in I}$ splitting the unit of $\CA$
\begin{equation*}
    \sum_{i \in I} p_i^* p_i = \mathbb{1}.
\end{equation*}
The \textit{joint} of two quantum partitions $\CP =  (p_i)_{i \in I}$ and $\CQ = (q_j)_{j \in J}$ is defined as the quantum partition 
\begin{equation*}
    \CP \vee \CQ = (p_i q_j)_{(i,j) \in I \times J}.
\end{equation*}
It is straightforward to check that $\CP \vee \CQ$ is indeed a quantum partition. The main difference with the classical case is that the join operation is not commutative, meaning that in general $\CP \vee \CQ \neq \CQ \vee \CP$. The join operation remains nonetheless associative.

We now turn to quantum dynamical entropies and define them analogously to the classical entropy of a partition for a probability measure. 
\begin{definition}
    Let $\rho \in \CE$ be a state and $\CP, \CQ$ be quantum partitions. The \textit{quantum dynamical entropy} of $\rho$ with respect to $\CP$ is defined as 
    \begin{equation*}
        h^q(\rho, \CP) = \sum_{p \in \CP} \eta(\rho(p^* p)).
    \end{equation*}
    The quantum dynamical entropy of $\rho$ with respect to $\CQ$ conditioned by $\CP$ is defined as 
    \begin{equation*}
        h^q(\rho, \CQ|\CP) = \sum_{p \in \CP} \rho(p^* p) h^q(\rho_{|p}, \CQ),
    \end{equation*}
    where $h^q(\rho_{|p}, \CQ) = 0$ if $\rho(p^* p) = 0$.
\end{definition}

If $\CA = \CB(\CH)$ is the $C^*$ algebra of bounded operators on a Hilbert space $\CH$ and $\psi \in \CH$ is a normalized vector, we denote 
\begin{equation*}
    h^q(\psi, \cdot) := h^q(\rho_{\psi}, \cdot) \qqtq{where} \rho_\psi(x) = \sprod{x \psi}{\psi}.
\end{equation*}
It is straightforward to check that $\rho_\psi$ is a state on $\CA$ as $\psi$ is normalized.

\begin{remark}
    Other definitions of quantum entropy for a state on a Hilbert space exist, the most prominent one being the Von Neumann entropy. Let $\rho$ be a density matrix on a Hilbert space $\CH$, meaning a trace-one positive operator on $\CH$. The Von Neumann entropy of $\rho$, generally denoted $S(\rho)$, is defined as
    \begin{equation*}
        S(\rho) := \tr(\eta(\rho)) = -\tr(\rho \log( \rho)) = \sum_{n \ge 0} \eta(\lambda_n),
    \end{equation*}
    if $(\lambda_n)_n$ is the spectrum of $\rho$. The entropy defined by \eqref{eq_quantum_entropy_anantharaman_nonnenmacher} can be seen as a Von Neumann entropy. It is the Von Neumann entropy of the diagonal matrix whose diagonal coefficients are the $(\norm{P_i \psi}^2)_{i \in I}$. 
    
    Alicki and Fannes \cite[chap. 10]{alickiQuantumDynamicalSystems2001a} defined a quantum dynamical entropy associated to a quantum partition by also considering the non-diagonal coefficients. More precisely, given a quantum partition $\CP = (p_i)_{1 \le i\le n}$ and a state $\rho \in \CE$ on $\CA$, they defined the entropy of $\rho$ with respect to $\CP$ as 
    \begin{equation*}
        S(\Gamma_\CP(\rho)) \qqtq{where} \Gamma_\CP(\rho) = (\rho(p_i^*p_j))_{1 \le i,j \le n} \in M_n(\C).
    \end{equation*}
    The matrix $\Gamma_\CP(\rho)$ is a positive trace-one matrix in $M_n(\C)$ because $\rho$ is positive of norm $1$ and $\CP$ is a quantum partition. Our quantum dynamical entropy is similar but we only keep the diagonal coefficients of $\Gamma_\CP$. This does not change much as we will only consider quantum partitions that satisfy $\norm{p_i^* p_j} \ll 1$ if $i \neq j$. Under such a condition, the matrix $\Gamma_\CP$ is a small perturbation of the matrix obtained by only keeping the diagonal coefficients. 
\end{remark}

The unitary elements of $\CA$ naturally act on quantum partitions and states. More precisely, a unitary element $u \in \CA$ acts on a state $\rho \in \CE$ and on a quantum partition $\CP = (p_i)_{i \in I}$ by
\begin{equation*}
    (u \cdot \rho)(\cdot) = \rho(u^{-1} \cdot u) \qqtqq{and} u \cdot \CP = (u^{-1} p_i u)_{i \in I}.
\end{equation*}
With this convention, we obtain the following relation for the conditional entropy of $u \cdot \CQ$ with respect to $\CR$
\begin{equation}\label{eq_action_unitary_conditional_quantum_entropy}
    h^q(\rho, u \cdot \CQ | \CR) = h^q(u \cdot \rho, \CQ | u^{-1} \cdot \CR) .
\end{equation}

For classical entropies, getting rid of the conditioning always increases the entropy. In the quantum setting, this may no longer be the case. Nevertheless, the error can be controlled by the lack of commutativity of the quantum partitions involved. 

\begin{theorem}\label{thm_prop_cond_quantum_entrop_quantum_part}
    Let $\CP$, $\CQ$ and $\CR$ be three quantum partitions of $\CA$ and $\rho \in \CE$ be a state. Then, the following holds
    \begin{itemize}
        \item $h^q(\rho, \CQ \vee \CP |\CR) = h^q(\rho, \CP| \CR) + h^q(\rho, \CQ | \CP \vee \CR)$,
        \item $h^q(\rho, \CQ | \CP \vee \CR) \le h^q(\rho, \CQ|\CR) + \CO(\norm{[\CP,\CQ]})$,
        \item $h^q(\rho, \CP| \CR) + h^q(\rho, \CQ | \CR) \ge -2 \log(\max_{p \in \CP, q \in \CQ} \norm{p^*q})$.
    \end{itemize}
    The implicit constant in the error term depends on the cardinal $|\CQ|$ of the partition $\CQ$ and 
    \begin{equation*}
        \norm{[\CP,\CQ]} = \sum_{q \in \CQ, p \in \CP} \norm{[p^*, q^*q]p}.
    \end{equation*}
\end{theorem}
The first two points of the theorem imply directly that the conditional quantum entropy is subadditive up to an error involving the commutators. The third point of the theorem is a conditional entropic uncertainty principle that generalizes that of Maassen and Uffink \cite{maassenGeneralizedEntropicUncertainty1988}. It is a central tool in the proof of lower bound on the entropy of semiclassical measures by Anantharaman and Nonnenmacher \cite{anantharamanHalfdelocalizationEigenfunctionsLaplacian2007a}. The proof of Theorem \ref{thm_prop_cond_quantum_entrop_quantum_part} is given in the more general case of approximate quantum partitions in Section \ref{sec_prop_cond_quantum_entrop}. This extension of the notion of quantum partition is needed to handle the objects that will appear in the proof of Theorems \ref{thm_lower_bound_full_entropy} and \ref{thm_lower_bound_entropy_projection} and is presented in the next section.

\subsection{From classical partitions to approximate quantum partitions}\label{sec_classical_to_quantum_partitions}

Since we want to bound the Kolmogorov-Sinai entropy of semiclassical measures on $\T^{2d}$, we will build quantum partitions out of classical partitions of unity of $\T^{2d}$. Given a smooth classical partition of unity $\CQ = (q_i^2)_{i \in I}$ of $\T^{2d}$, we denote $Q_i = \Op^w_h(q_i)$ the quantization of $q_i$ and $\CQ^q := (Q_i)_{i \in I}$ the approximate quantum partition associated to $\CQ$. The product rule for the Weyl quantization \eqref{Moyal_product_R2d} and the Calderon-Vaillancourt theorem \eqref{Calderon-Vaillancourt_R2d} implies that 
\begin{equation*}
    \sum_{i \in I} Q_i^* Q_i = \sum_{i \in I} \Op^w_h(\bar{q_i}) \Op^w_h(q_i) = \sum_{i \in I} \Op^w_h(|q_i|^2) + \CO(h) = Id_{L^2} + \CO(h).
\end{equation*}
Therefore, $\CQ^q$ is not exactly a quantum partition due to the error term $\CO(h)$. It nonetheless behaves like one in the semiclassical limit, up to some errors of order $h$. 

The notion of approximate quantum partition is introduced to deal with this issue. An \textit{approximate quantum partition} is a family $\CP = \set{p_i}_{i\in I}$ which satisfies
\begin{equation}\label{eq_def_approximate_quantum_partition}
    \sum_{p_i \in \CP} p_i^* p_i = \mathbb{1} + r_{\CP} \qqtq{with} r_{\CP} \in \CA.
\end{equation}
The element $r_{\CP}$ will be generally thought of as small and $\CP$ is called a \textit{$(1+ \epsilon)$ quantum partition} if $\norm{r_{\CP}} \le \epsilon$. Approximate quantum partitions behave similarly to exact quantum partitions, up to some errors in $\epsilon$. Note that the relation \eqref{eq_def_approximate_quantum_partition} implies that $r_{\CP}$ is self-adjoint, bigger than $-\mathbb{1}$ and that any element $p_i$ of $\CP$ satisfies $\norm{p_i} \le \sqrt{1 + \norm{r_{\CP}}}$. The join of two approximate quantum partitions is again an approximate quantum partition.

\begin{lemma}\label{lem_joint_approx_quantum_part_is_approx_quantum_part}
    Let $\CP$, respectively $\CQ$, be a $(1+\epsilon_1)$ quantum partition, respectively a $(1+\epsilon_2)$ quantum partition. Then there exists $r \in \CA$ such that 
    \begin{equation*}
        \sum_{p \in \CP, q \in \CQ} q^* p^* p q = \mathbb{1} + r \qqtqq{where} \norm{r} \le (1+\epsilon_1) (1+ \epsilon_2) - 1.
    \end{equation*}
    In particular, $\CP \vee \CQ$ is a $(1+\epsilon_1) (1+ \epsilon_2)$ quantum partition. 
\end{lemma}

\begin{proof}
    We calculate 
    \begin{align*}
        \sum_{p \in \CP, q \in \CQ} q^* p^* p q = \sum_{q \in \CQ} q^* (I + r_{\CP}) q = I + r_{\CQ} + \sum_{q \in \CQ} q^* r_{\CP} q.
    \end{align*}
    Denote $A = \sum_{q \in \CQ} q^* r_{\CP} q$, it is a self-adjoint element of $\CA$. Let $\rho \in \CE$. The inequality \eqref{ineq_restriction_state_C*algebra} implies that
    \begin{equation*}
        |\rho(A)| \le \sum_{q \in \CQ} |\rho(q^* r_{\CP} q)| \le \sum_{q \in \CQ} \norm{r_{\CP}} \rho(q^*q) = \epsilon_1 \rho(I + r_{\CQ}) \le \epsilon_1(1+ \epsilon_2).
    \end{equation*}
    Since $A$ is a self-adjoint element of $\CA$, then $\norm{A} = \sup_{\rho \in \CE} |\rho(A)|$ by \cite[lem. I.9.10]{davidsonC^algebrasExample1996} and thus
    \begin{equation*}
        \norm{r_{\CQ} + A} \le \epsilon_2 + \epsilon_1(1+ \epsilon_2) = (1+\epsilon_1) (1+ \epsilon_2) - 1.
    \end{equation*}
\end{proof}

We record a useful bound on the entropy of a $(1+\epsilon)$ quantum partition $\CQ$ valid for any state $\rho \in \CE$.
\begin{equation}\label{ineq_bound_quantum_entropy_approximate_quant_part}
    |h^q(\rho, \CQ)| \le (1+\epsilon)\log(|\CQ|) + \epsilon(1+\epsilon)
\end{equation}
Indeed, denoting $m := \sum_{q \in \CQ} \rho(q^* q) \le 1+\epsilon$, we get
\begin{align*}
    |h^q(\rho, \CQ)| &= m |\sum_{q \in \CQ} \frac{\rho(q^* q)}{m} \log(\frac{m}{\rho(q^* q)}) - \log(m)| \le  |\eta(m)| + m | \log \left(\frac{m}{m} \sum_{q \in \CQ} \frac{\rho(q^* q)}{\rho(q^* q)}\right)| \\
    &\le m\log(|\CQ|) + |\eta(m)| \le (1+\epsilon)\log(|\CQ|) + \epsilon(1+\epsilon).
\end{align*}

We will need better errors than $\CO(h)$ for the approximate quantum partitions that we will use. The next lemma tells us that up to a small modification of the classical observables, we can always ensure that the quantization of a classical smooth partition of unity is a $(1+\CO(h^\infty))$ quantum partition.  

\begin{lemma}\label{lem_modifyin_partition_to_have_1_plus_h_infinity}
    Let $\CP = (p_j^2)_{j \in J}$ be a partition of unity of $\R^{2d}$ by smooth functions that are invariant under a symplectic subspace $V$. Then there exists a family of functions $\tilde{\CP} = (\tilde{p}_j)_{j \in J}$ composed of smooth real functions invariant under $V$ such that each $\tilde{p_j}$ has the same support as $p_j$ and 
    \begin{equation*}
        \norm{\sum_{j \in J} \Op^w_h(\tilde{p}_j)^* \Op^w_h(\tilde{p}_j) - Id}_{L^2(\R^d)} = \CO(h^\infty).
    \end{equation*}
    In other words, $\tilde{\CP}^q$ is a $(1+ \CO(h^\infty))$ quantum partition of $L^2(\R^d)$. Moreover, $\tilde{p}_j = p_j + \CO(h)$.
\end{lemma}

\begin{proof}
    This proof relies on the tools of semiclassical analysis recalled at the beginning of Section \ref{sec_quantization_torus} and the reader unfamiliar with this theory is invited to skim through the part on the Weyl quantization before reading this proof.  
    The asymptotic expansion for the Moyal product \cite[thm. 4.17-4.18]{zworskiSemiclassicalAnalysis2012} for symbols in $S(1)$ gives an asymptotic expansion in $S(1)$ of the form
    \begin{equation*}
        a \# b(z) \sim \sum_{k=0}^\infty \frac{i^k h^k}{k!} A(D)^k(a(x,\xi)b(y, \eta))_{|(x, \xi) = (y,\eta) = z} ,
    \end{equation*}
    where $A(D) := \frac{1}{2} \sigma(D_x,D_\xi, D_y, D_\eta)$. The notation $f \sim g$ means in this context that $f-g = \CO(h^\infty)$ as elements of $S(1)$. Note that if $a = b$ then all the terms of odd order are null by antisymmetry of the operator $A(D)$. 
    
    We search for functions $\tilde{p}_j$ of the following form
    \begin{equation}\label{eq_proof_lem_modifyin_partition_to_have_1_plus_h_infinity_ansatz_developpment}
        \tilde{p}_j \sim \sum_{k=0}^\infty h^{2k} r_{j,k} \qqtq{with} r_{j,0} = p_j,
    \end{equation}
    and the $r_{j,k}$ are smooth real valued functions. This definition ensure that $\tilde{p}_j = p_j + \CO(h)$. The formula for the Moyal product gives
    \begin{equation*}
        \sum_{j \in J} \Op^w_h(\tilde{p}_j)^* \Op^w_h(\tilde{p}_j) -Id = \sum_{j \in J} \Op^w_h(\tilde{p}_j) \Op^w_h(\tilde{p}_j) - Id = \Op^w_h(q-1) ,
    \end{equation*}
    where
    \begin{align*}
        q - 1&\sim -1 + \sum_{j \in J} \sum_{n=0}^\infty \frac{(-1)^n h^{2n}}{(2n)!} A(D)^{2n}(\tilde{p}_j(x,\xi) \tilde{p}_j(y, \eta))_{| z} \\ 
        &\sim \sum_{j \in J} \sum_{\substack{n,k,l=0 \\ (n,k,l) \neq (0,0,0)}}^\infty \frac{(-1)^n h^{2(n+k+l)}}{(2n)!} A(D)^{2n}(r_{j,k}(x,\xi) r_{j,l}(y, \eta))_{|z}.
    \end{align*}
    We used that $\sum_{j \in J} p_j^2 = 1$ as $\CP = (p_j^2)_{j \in J}$ is a partition of unity and $r_{j,0} = p_j$. Let $n,k,l,m \in \N$ such that $n+k+l = m$, then either $k=m$ and $n=l=0$, or $l=m$ and $n=k=0$, or $n$, $k$ and $l$ are strictly smaller than $m$. Therefore, regrouping the powers of $h$ together gives 
    \begin{align*}
        q-1 & \sim \sum_{m \ge 1} c_m h^{2m} \sum_{j \in J} 2p_j r_{j,m} - F_m(r_{j,0}, \dots, r_{j,m-1}),
    \end{align*}
    where $F_m$ is a differential operator of finite order. In particular, if we define $r_{j,m}$ inductively by 
    \begin{equation}\label{eq_proof_lem_modifyin_partition_to_have_1_plus_h_infinity_ansatz_r_j_k}
        r_{j,m} = - \frac{1}{2} p_j \sum_{i \in J} F_m(r_{i,0}, \dots, r_{i,m-1}),
    \end{equation}
    they are invariant under $V$ and have the same support as $p_j$. 
    The equality $\sum_{j \in J} p_j^2 = 1$ implies that 
    \begin{align*}
        q-1 & \sim \sum_{m \ge 1} c_m h^{2m} \sum_{j \in J} p_j^2 \sum_{i \in J} F_m(r_{i,0}, \dots, r_{i,m-1})  - F_m(r_{j,0}, \dots, r_{j,m-1}) \\
        &\sim 0.
    \end{align*}
    In other words, $q = 1 + \CO_{S(1)}(h^\infty)$. The Calderón-Vaillancourt \eqref{Calderon-Vaillancourt_R2d} thus leads to
    \begin{equation*}
        \norm{\sum_{j \in J} \Op^w_h(\tilde{p}_j)^* \Op^w_h(\tilde{p}_j) -Id} = \norm{\Op^w_h(q-1)} = \CO(h^\infty) .
    \end{equation*} 
    
    The reasoning presented here only explained how one should choose the functions $r_{j,k}$ to obtain the desired result. To prove the result, we define functions $(r_{j,k})_{j \in J, k \ge 0}$ inductively by \eqref{eq_proof_lem_modifyin_partition_to_have_1_plus_h_infinity_ansatz_r_j_k} and define the $\tilde{p}_j$ by \eqref{eq_proof_lem_modifyin_partition_to_have_1_plus_h_infinity_ansatz_developpment} via Borel's theorem \cite[thm. 4.15]{zworskiSemiclassicalAnalysis2012}. The computation above then guarantees us that the functions $(\tilde{p}_j)_{j \in J}$ fit the bill. 
\end{proof}

\subsection{Properties of conditional quantum dynamical entropy}\label{sec_prop_cond_quantum_entrop}

We give the proof of the statements in Theorem \ref{thm_prop_cond_quantum_entrop_quantum_part} for approximate quantum partitions. To recover the results of the theorem, it suffices to take $\epsilon$ in the statement of this section. We also advise the reader to suppose that $\epsilon =0$ when reading the proofs of this section for the first time.

Up to an error for approximate quantum partitions, we recover the classical relation between the entropy of the join and the conditional entropy.  
\begin{lemma}\label{lem_key_relation_conditional_quantum_entropy}
    Let $\rho \in \mathcal{E}$ be a state and $\CP$, $\CQ$ and $\CR$ be three $(1 + \epsilon)$ quantum partitions. Then,
    \begin{equation*}
        h^q(\rho, \CQ \vee \CP | \CR) = h^q(\rho, \CP | \CR) + h^q(\rho, \CQ | \CP \vee \CR) + \CO(\epsilon(1+ \log(|\CP|))).
    \end{equation*}
\end{lemma}

\begin{proof} A direct computation gives
    \begin{align*}
        h^q(\rho, \CQ \vee \CP | \CR) &= - \sum_{r \in \CR} \sum_{p \in \CP} \sum_{q \in \CQ} \rho(r^* r) \rho_{|r}(p^* q^*q p) \log(\rho_{|r}(p^* q^*q p)) \\
        &= - \sum_{r \in \CR} \sum_{p \in \CP} \sum_{q \in \CQ} \rho(r^* p^* p r) \rho_{|pr}(q^*q) \log(\rho_{|pr}(q^*q) \rho_{|r}(p^* p)) \\
        &= - \sum_{pr \in \CP \vee \CR} \sum_{q \in \CQ} \rho((pr)^* pr) \rho_{|pr}(q^*q) \log(\rho_{|pr}(q^*q)) \\
        &\qquad - \sum_{r \in \CR} \sum_{p \in \CP} \rho(r^* p^* p r) \log(\rho_{|r}(p^* p)) \sum_{q \in \CQ} \rho_{|pr}(q^*q) \\
        &= h^q(\rho, \CQ | \CP \vee \CR) - \sum_{r \in \CR} \rho(r^* r) \sum_{p \in \CP} \rho_{|r}(p^* p) \log(\rho_{|r}(p^* p)) \rho_{|pr}(\mathbb{1}) \\
        &\qquad - \sum_{r \in \CR} \rho(r^* r) \sum_{p \in \CP} \rho_{|r}(p^* p) \log(\rho_{|r}(p^* p)) \rho_{|pr}(r_\CQ) \\
        &= h^q(\rho, \CQ | \CP \vee \CR) + h^q(\rho,\CP | \CR) + \CO(\epsilon(1 + \log(|\CP|))).
    \end{align*}
    Indeed, for $r\in \CR$ and denoting $m_r = \sum_{p \in \CP} \rho_{|r}(p^*p)$, we have 
    \begin{align*}
        |\sum_{p \in \CP} \rho_{|r}(p^* p) \log(\rho_{|r}(p^* p)) \rho_{|pr}(r_\CQ)| &\le m_r | \sum_{p \in \CP} \rho_{|pr}(r_\CQ)(\eta (\frac{\rho_{|r}(p^* p)}{m_r}) +\frac{\rho_{|r}(p^* p)}{m_r} \log(m_r))| \\
        &\le \epsilon m_r |\log(m_r)| + \epsilon m_r \sum_{p \in \CP} \eta (\frac{\rho_{|r}(p^* p)}{m_r}) \\
        &\le \epsilon \max(e^{-1}, \epsilon(1+\epsilon)) + \epsilon(1+\epsilon) \log(|\CP|) \\
        &\le \CO(\epsilon(1+\log(|\CP|))).
    \end{align*}
    Therefore 
    \begin{align*}
        |\sum_{r \in \CR} \rho(r^* r) \sum_{p \in \CP} \rho_{|r}(p^* p) \log(\rho_{|r}(p^* p)) \rho_{|pr}(r_\CQ)| &\le \CO(\epsilon(1+\log(|\CP|))) \sum_{r \in \CR} \rho(r^* r) \\
        &\le \CO(\epsilon(1+\log(|\CP|))) (1+ \epsilon) .
    \end{align*}
\end{proof}

This proof is exactly the transposition of the proof of the first equation of \eqref{relation_cond_classical_entropy} to the quantum setting. We highlight that in the quantum version it is not possible to exchange either $\CQ$ and $\CP$ on the left-hand side or $\CP$ and $\CR$ on the right-hand side due to the non-commutativity of the join of two quantum partitions.

We will need a modulus of continuity for the entropy of a finite probability measure and we adapt the one given in \cite[eq. 4]{zhangEstimatingMutualInformation2007}. Let $(p_i)_{i \in I}$ be a probability measure over the finite set $I$. Its entropy is defined as 
\begin{equation}
    H(p_i) = \sum_{i \in I} \eta(p_i).
\end{equation}
This entropy functional has a modulus of continuity depending on the total variation distance and the size of the set $I$.

\begin{lemma}\label{lem_modulus_continuity_classical_entropy}
    Let $I$ be a finite set, $(p_i)_{i \in I}$ and $(q_i)_{i\in I}$ two probability measures on $I$. Then 
    \begin{equation*}
        |H(p_i) - H(q_i)| \le 2 (\log(|I|)+1) \sqrt{\theta} \qtq{where} \theta := \sum_{i \in I} |p_i - q_i|.
    \end{equation*}
\end{lemma}

\begin{proof}
    An optimal bound is given in \cite[eq. 4]{zhangEstimatingMutualInformation2007} and has the form
    \begin{align*}
        |H(p_i) - H(q_i)| \le \frac{\theta}{2} \log(|I|-1) + H_b(\frac{\theta}{2}),
    \end{align*}
    where $H_b(p) = \eta(p) + \eta(1-p)$ is the binary entropy for a number $p \in [0,1]$. As $\log(x) \le \sqrt{x-1}$ if $x\ge 1$, we get for $p \in ]0,1[$ that\footnote{This nice proof of the bound $H_b(p) \le 2 \sqrt{p(1-p)}$ is the one given by Chun SO on https://math.stackexchange.com/questions/1432043/an-upper-bound-of-binary-entropy.}
    \begin{align*}
        H_b(p) &= p \log \left(\frac{1}{p}\right) + (1-p) \log \left(\frac{1}{1-p}\right) \le p \sqrt{\frac{1}{p}-1} + (1-p) \sqrt{\frac{1}{1-p}-1} \\
        &\le 2 \sqrt{p(1-p)} \le 2 \sqrt{p}.
    \end{align*}
    Since $\theta/2 \le 1$, we thus have $\theta/2 \le \sqrt{\theta/2}$ and $H_b(\theta/2) \le 2 \sqrt{\theta/2}$. This implies the lemma.
\end{proof}

We derive a modulus of continuity for the quantum entropy thanks to the one for the classical entropy of Lemma \ref{lem_modulus_continuity_classical_entropy}.

\begin{lemma}\label{lem_modulus_continuity_quantum_entropies}
    Let $\rho_1 \in \CA^*$ be a linear positive functional on $\CA$ and $\rho_2 \in \CE$ be a state, $\CQ$ and $\CR$ be two $(1+\epsilon)$ quantum partitions with $\epsilon <1/2$. Denote $N = \norm{\rho_1} = \rho_1(\mathbb{1})$ and 
    \begin{equation*}
        \theta := \sum_{q \in \CQ} |\rho_1(q^*q) - \rho_2(q^*q)| \qqtqq{and} \delta := |N- 1| ,
    \end{equation*}
    then we have 
    \begin{equation*}
        |N h^q(\frac{\rho_1}{N}, \CQ) - h^q(\rho_2, \CQ)| \le \CO(\log(|\CQ|))(1+\delta)\sqrt{\epsilon(1+\delta) + \theta}.
    \end{equation*}
\end{lemma}

\begin{proof}
    Denote $\tilde{\rho}_1 = \frac{\rho_1}{N}$. Define $m_1 = \sum_{q \in \CQ} \tilde{\rho}_1(q^* q)$ and $m_2 = \sum_{q \in \CQ} \rho_2(q^* q)$. We have 
    \begin{equation*}
        |m_1 - 1| = |\sum_{q \in \CQ} \tilde{\rho}_1(q^*q) - 1| = |\tilde{\rho}_1(\mathbb{1}) + \tilde{\rho}_1(r_\CQ) - 1| \le \norm{r_\CQ} \le \epsilon,
    \end{equation*}
    as $\CQ$ is a $(1+ \epsilon)$ quantum partition and $\tilde{\rho}_1(\mathbb{1}) = 1$. Similarly, $|m_2 - 1| \le \epsilon$.
    Note that for any $(\rho, \alpha) \in \CE \times \R^{+*} $
    \begin{equation*}
        h^q(\rho, \CQ) = \alpha h^q(\frac{\rho}{\alpha}, \CQ) - \sum_{q \in \CQ}\rho(q^*q) \log(\alpha).
    \end{equation*}
    Applying this equality to $(\tilde{\rho}_1,m_1)$ and $(\rho_2, m_2)$ gives 
    \begin{align*}
        |Nh^q(\frac{\rho_1}{N}, \CQ) - h^q(\rho_2, \CQ)| &= |m_1 N h^q(\frac{\rho_1}{m_1 N}, \CQ) + N \eta(m_1) - m_2 h^q(\frac{\rho_2}{m_2}, \CQ) - \eta(m_2)| \\
        &\le N |\eta(m_1)| + |\eta(m_2)| + |m_1 N-m_2| h^q(\frac{\rho_2}{m_2}, \CQ) \\
        &\qquad + m_1 N |h^q(\frac{\rho_1}{m_1 N}, \CQ) - h^q(\frac{\rho_2}{m_2}, \CQ)| .
    \end{align*}
    As $|m_1 N -m_2| \le \epsilon(2 + \delta)$, $|h^q(\frac{\rho_2}{m_2}, \CQ) \le \log(|\CQ|)$ and $|\eta(x)| = \CO(|x-1|)$ if $x \le 2$, we have 
    \begin{equation*}
        |Nh^q(\frac{\rho_1}{N}, \CQ) - h^q(\rho_2, \CQ)| \le \epsilon(2+ \delta)(\log(|\CQ|) + 1) + N m_1 |h^q(\frac{\rho_1}{m_1 N}, \CQ) - h^q(\frac{\rho_2}{m_2}, \CQ)|.
    \end{equation*}
    We apply the modulus of continuity for classical entropy \ref{lem_modulus_continuity_classical_entropy} to the probability distributions 
    \begin{equation*}
        p^1(q) = \frac{1}{N m_1}\rho_1(q^*q) \qqtqq{and} p^2(q) = \frac{1}{m_2} \rho_2(q^* q),
    \end{equation*}
    to obtain
    \begin{align*}
        Nm_1|h^q(\frac{\rho_1}{Nm_1}, \CQ) - h^q(\frac{\rho_2}{m_2}, \CQ)| &\le 2(\log(|\CQ|)+1)N m_1 \sqrt{\sum_{q \in \CQ}|\frac{\rho_1}{N m_1}(q^*q) - \frac{\rho_2}{m_2}(q^*q)|} \\
        &\le 2(\log(|\CQ|)+1) \sqrt{Nm_1} \sqrt{|\frac{Nm_1}{m_2}-1| \sum_{q \in \CQ} \rho_2(q^*q) + \theta} \\
        &\le 2(\log(|\CQ)+1) \sqrt{(1+\delta)(1+\epsilon)}\sqrt{|Nm_1-m_2| + \theta} \\
        &\le 2(\log(|\CQ)+1)(1+\epsilon)(1+\delta) \sqrt{\epsilon (2+\delta) + \theta} \\
        &\le 4(\log(|\CQ)+1)(1+\delta) \sqrt{\epsilon (2+\delta) + \theta} ,
    \end{align*}
    which concludes.
\end{proof}

Let $\rho \in \CE$ be a state and $\CP$ be a $(1+ \epsilon)$ quantum partition. The \textit{mean of the state $\rho$ with respect to $\CP$}, denoted $\mean{\rho}_{\CP}$, is defined as 
\begin{equation*}
    \mean{\rho}_{\CP}(a) = \sum_{p \in \CP} \rho(p^* a p) \qqtq{for} a \in \CA.
\end{equation*}
For $a \in \CA$, we will sometimes denote $\mean{a}_{\CP}$ instead of $\mean{\rho}_{\CP}(a)$ if $\rho$ is clear from the context to lighten the notations. 
The linear functional $\mean{\rho}_{\CP}$ is again a state if $\epsilon =0$, otherwise it is a positive linear functional on $\CA$ and satisfies 
\begin{equation*}
    \norm{\mean{\rho}} = \mean{\rho}_\CP(\mathbb{1}) \le \sum_{p \in \CP} \rho(p^*p) \le 1 + \epsilon.
\end{equation*}

The difference between $\mean{\rho}_{\CP}$ and $\rho$ is controlled by commutators with the elements of the quantum partition $\CP$. The following two quantities appear naturally in the error terms, 
\begin{equation}\label{eq_def_error_terms_commutators}
    [\CP,a] = \sum_{p \in \CP} [p^*,a] p \qqtqq{and} \norm{[\CP,\CQ]} = \sum_{q \in \CQ} \norm{[\CP, q^*q]},
\end{equation}
for $a \in \CA$, $\CP$ and $\CQ$ be two quantum partitions.

\begin{lemma}\label{lem_almost_eq_mean_quantum_part}
    Let $\rho \in \mathcal{E}$ be a state, $a \in \CA$ be an observable and $\CP_1, \dots, \CP_n$ be $(1+ \epsilon) $quantum partitions. Then 
    \begin{equation*}
        |\mean{\rho}_{\CP_1 \vee \dots \vee \CP_n}(a) - \rho(a)| \le (1 + \epsilon)^n (n \epsilon \norm{a} + \sum_{i=1}^n\norm{[\CP_i, a]}).
    \end{equation*}
\end{lemma}

\begin{proof}
    We have 
    \begin{align*}
        \mean{a}_{\CP_1 \vee \dots \vee \CP_n} = \mean{a}_{\CP_2 \vee \dots \vee \CP_n} + \mean{ar_{\CP_1}}_{\CP_2 \vee \dots \vee \CP_n} + \sum_{p_1 \in \CP_1} \mean{[p_1^*, a] p_1}_{\CP_2 \vee \dots \vee \CP_n}.
    \end{align*}
    Iterating this equality leads to
    \begin{align*}
        \mean{a}_{\CP_1 \vee \dots \vee \CP_n} = \rho(a) + \sum_{i=1}^n \mean{ar_{\CP_i}}_{\CP_{i+1} \vee \dots \vee \CP_n} + \mean{\sum_{p_i \in \CP_i} [p_i^*, a] p_i}_{\CP_{i+1} \vee \dots \vee \CP_n}.
    \end{align*}
    Note that, for any $x \in \CA$ and any $(1+\epsilon')$ quantum partition $\CP'$,
    \begin{equation*}
        |\mean{x}_{\CP'}| \le \norm{x} \sum_{p \in \CP'} \rho(p^*p) = \norm{x} \rho(\mathbb{1} + r_{\CP'}) \le \norm{x} (1+\epsilon').
    \end{equation*}
    Thus, since $\CP_{i+1} \vee \dots \vee \CP_n$ is a $(1+ \epsilon)^{n-i}$ quantum partition, 
    \begin{align*}
        |\mean{a}_{\CP_1 \vee \dots \vee \CP_n} - \rho(a)| &\le \sum_{i=1}^n \norm{a} \epsilon (1+\epsilon)^{n-i} + \norm{\sum_{p_i \in \CP_i} [p_i^*, a] p_i} (1+\epsilon)^{n-i} \\
        &\le (1 + \epsilon)^n (n \epsilon \norm{a} + \sum_{i=1}^n\norm{[\CP_i,a]}).
    \end{align*}
\end{proof}

Using Lemmas \ref{lem_modulus_continuity_quantum_entropies} and \ref{lem_almost_eq_mean_quantum_part}, we show that the conditional quantum dynamical entropy increases when the conditioning is removed, up to an error controlled by the commutators of the partitions.   

\begin{lemma}\label{lem_almost_monotony_by_deconditioning_quantum_entropy}
    Let $\rho \in \CE$ be a state and $\CQ, \CR, \CP_1, \dots, \CP_n$ be $(1+\epsilon)$ quantum partitions. Suppose that $n\epsilon <1/2$, then
    \begin{equation*}
        h^q(\rho,\CQ | \CP_1 \vee \dots \vee \CP_n \vee \CR) \le h^q(\rho, \CQ|\CR) + \CO(\sqrt{n\epsilon + \sum_{i=1}^n \norm{[\CP_i, \CQ]}}),
    \end{equation*}
    where the implicit constant depends on $|\CQ|$. Moreover,
    \begin{equation*}
        h^q(\rho,\CQ | \CR \vee \CP_1 \vee \dots \vee \CP_n ) \le h^q(\rho, \CQ| \CR) + \CO(\sqrt{n \epsilon + \sum_{i=1}^n \norm{[\CP_i, \CQ\vee \CR]}}),
    \end{equation*}
    where the implicit constants depend on $|\CQ|$, $|\CR|$ and $\max_{1 \le i \le n}|\CP_i|$.
\end{lemma}
\begin{proof}
    Denote $\tilde{\CP} = \CP_1 \vee \dots \vee \CP_n$ which is a $(1+\epsilon)^n$ quantum partition and for $r \in \CR$
    \begin{equation*}
        L_r := \sum_{p \in \tilde{\CP}}  \rho_{|r}(p^* p) = \rho_{|r}(\mathbb{1} + r_{\tilde{\CP}}) \le (1+\epsilon)^n < e^{1/2} < 2,
    \end{equation*}
    and $L_r>0$ as $\norm{r_{\tilde{\CP}}} \le (1+\epsilon)^n -1 < 1$. 
    Fix $r \in \CR$, the concavity of $\eta$ implies
    \begin{align*}
        h^q(\rho_{|r}, \CQ | \CP_1 \vee \cdots \vee \CP_n) &= \sum_{q \in \CQ} \sum_{p \in \tilde{\CP}} \rho_{|r}(p^* p) \eta(\rho_{|p r}(q^*q)) =  L_r \sum_{q \in \CQ} \frac{\sum_{p \in \tilde{\CP}}  \rho_{|r}(p^* p)}{L_r} \eta(\rho_{|p r}(q^*q)) \\
        &\le  L_r \sum_{q \in \CQ}  \eta( \frac{1}{L_r}\sum_{p \in \tilde{\CP}} \rho_{|r}(p^* p) \rho_{|p r}(q^*q)) =  L_r \sum_{q \in \CQ} \eta( \frac{1}{L_r} \sum_{p \in \tilde{\CP}} \rho_{|r}(p^* q^*q p)) \\
        &=  L_r \sum_{q \in \CQ}  \eta( \frac{1}{L_r} \mean{\rho_{|r}}_{\tilde{\CP}}(q^* q)) =  L_r h^q(\frac{1}{L_r} \mean{\rho_{|r}}_{\tilde{\CP}}, \CQ).
    \end{align*}
    As $|L_r -1| \le 1$, Lemma \ref{lem_modulus_continuity_quantum_entropies} implies
    \begin{align*}
        |L_r h^q(\frac{1}{L_r} \mean{\rho_{|r}}_{\tilde{\CP}}, \CQ) - h^q(\rho_{|r}, \CQ)| \le \CO(\log(|\CQ|))\sqrt{\epsilon + \sum_{q \in \CQ}|\mean{\rho_{|r}}_{\tilde{\CP}}(q^* q) - \rho_{|r}(q^* q )|}.
    \end{align*}
    We deduce from Lemma \ref{lem_almost_eq_mean_quantum_part} and the fact that $\norm{q^*q} \le 1+ \epsilon \le 2$ and $(1+\epsilon)^n \le 2$ that
    \begin{align*}
        \sum_{q \in \CQ}|\mean{\rho_{|r}}_{\tilde{\CP}}(q^* q) - \rho_{|r}(q^* q )| \le \CO(|\CQ|n\epsilon) + 2\sum_{i=1}^n \norm{[\CP_i, \CQ]}.
    \end{align*}
    This concludes the proof of the first part of the lemma.

    We turn to the second inequality. The concavity of $\eta$ leads to 
    \begin{align*}
        h^q(\rho, \CQ | \CR \vee \CP_1 \vee \cdots \vee \CP_n ) &= \sum_{q \in \CQ} \sum_{r \in \CR} \sum_{p \in \tilde{\CP}} \rho(p^* r^* r p ) \eta(\rho_{|r p}(q^*q)) \\
        &= \sum_{q \in \CQ} \sum_{r \in \CR}  \sum_{p \in \tilde{\CP}} \frac{\rho(p^* r^* r p )}{\mean{r^* r}_{\tilde{\CP}}} \eta(\rho_{|r p}(q^*q)) \\
        &\le \sum_{q \in \CQ} \sum_{r \in \CR} \mean{r^* r}_{\tilde{\CP}} \eta( \sum_{p \in \tilde{\CP}} \frac{\rho(p^* r^* r p )}{\mean{r^* r}_{\tilde{\CP}}} \rho_{|rp}(q^*q)) \\
        &= \sum_{q \in \CQ} \sum_{r \in \CR} \mean{r^* r}_{\tilde{\CP}} \eta( \frac{\mean{r^*q^*qr}_{\tilde{\CP}}}{\mean{r^* r}_{\tilde{\CP}}}) \\
        &= h^q(\mean{\rho}_{\tilde{\CP}}, \CQ | \CR) \\
        &= M h^q( \frac{1}{M}\mean{\rho}_{\tilde{\CP}}, \CQ | \CR),
    \end{align*}
    where $M = \mean{\rho}_{\tilde{\CP}}(\mathbb{1})$. Since $\tilde{\CP}$ is a $(1+\epsilon)^n$ quantum partition, $M$ satisfies $|M-1| = o(n \epsilon)$ and $M \le 2$. Using the relation of Lemma \ref{lem_key_relation_conditional_quantum_entropy}, we split the error in two parts and a small remainder 
    \begin{align*}
        |M h^q( \frac{1}{M}\mean{\rho}_{\tilde{\CP}}, \CQ | \CR) - h^q(\rho, \CQ| \CR)| &= |M h^q( \frac{1}{M}\mean{\rho}_{\tilde{\CP}}, \CQ \vee  \CR) - h^q(\rho, \CQ \vee \CR)| + \CO(\epsilon(1 + \log(|\CR|)))\\
        & \qquad + |M h^q( \frac{1}{M}\mean{\rho}_{\tilde{\CP}}, \CR) - h^q(\rho,\CR)|.
    \end{align*}
    Similarly to the proof of the first inequality of the lemma, we get thanks to the lemmas \ref{lem_almost_eq_mean_quantum_part} and \ref{lem_modulus_continuity_quantum_entropies} that 
    \begin{align*}
        |M h^q( \frac{1}{M}\mean{\rho}_{\tilde{\CP}}, \CQ \vee  \CR) - h^q(\rho, \CQ \vee \CR)| &\le \CO(|\CQ\vee \CR|\log(|\CQ\vee \CR|)) \sqrt{n\epsilon + \sum_{i=1}^n \norm{[\CP_i, \CQ\vee \CR]}} \\
        |M h^q( \frac{1}{M}\mean{\rho}_{\tilde{\CP}}, \CR) - h^q(\rho,\CR)| &\le \CO(|\CR|\log(|\CR|)) \sqrt{n\epsilon + \sum_{i=1}^n \norm{[\CP_i, \CR]}}.
    \end{align*}
    For $i \in \eint{1}{n}$, we have 
    \begin{align*}
        \norm{[\CP_i, \CR]} &= \sum_{r \in \CR} \norm{\sum_{p \in \CP_i} [p_i^*, r^*r]p_i} = \sum_{r \in \CR} \norm{\sum_{p \in \CP_i} [p_i^*, r^* (\sum_{q \in \CQ} q^* q - r_\CQ)r]p_i} \\
        &\le \sum_{r \in \CR} \sum_{q \in \CQ} \norm{\sum_{p \in \CP_i} [p_i^*, r^* q^* q r]p_i} + \sum_{r \in \CR} \sum_{q \in \CQ} \sum_{p \in \CP_i} \norm{[p_i^*, r^* r_\CQ r]p_i} \\
        &\le \norm{[\CP_i, \CQ \vee \CR]} + \CO(\epsilon |\CP_i||\CQ||\CR|).
    \end{align*}
    We used that all the partitions involved are $(1+\epsilon)$ quantum partitions and thus the norm of their elements is bounded by $\sqrt{1+\epsilon}$ and hence by $\sqrt{2}$ since $\epsilon <1/2$. Therefore,
    \begin{equation*}
        |M h^q( \frac{1}{M}\mean{\rho}_{\tilde{\CP}}, \CR) - h^q(\rho,\CR)| \le \CO(|\CR|\log(|\CR|)) \sqrt{ \CO(n\epsilon |\CQ||\CR| \max_{1 \le i \le n}|\CP_i|) + \sum_{i=1}^n \norm{[\CP_i, \CQ \vee \CR]}},
    \end{equation*}
    and combining everything gives
    \begin{align*}
        M h^q( \frac{1}{M}\mean{\rho}_{\tilde{\CP}}, \CQ | \CR) \le h^q(\rho, \CQ| \CR) + \CO(|\CR|^2 |\CQ|^2 \max_{1 \le i \le n}|\CP_i|) \sqrt{n \epsilon + \sum_{i=1}^n \norm{[\CP_i, \CQ\vee \CR]}}.
    \end{align*}
\end{proof}

We conclude this section with a conditional version of the entropic uncertainty principle adapted to our setting. The proof consists in a reduction via the GNS construction to a version of the entropic uncertainty principle on Hilbert space for approximate partitions by Nonnenmacher \cite[sec. 3.7 prop. 29]{nonnenmacherEntropyChaoticEigenstates2010}. 

\begin{theorem}[Conditional entropic uncertainty principle for approximate quantum partitions]\label{thm_conditional_entropic_uncertainty_principle_approximate quantum_partitions}
    Let $\rho \in \CE$ be a state, $\CP$, $\CQ$ be two $(1 + \epsilon)$ quantum partitions, $\CR$ a $(1+\epsilon_1)$ quantum partition and $u_1,u_2 \in \CA$ two unitary elements. Suppose that $\epsilon \le \min(|\CP|^{-2}, |\CQ|^{-2})$, then,
    \begin{equation*}
        h^q(\rho, u_1 \cdot \CP | \CR) + h^q(\rho, u_2 \cdot \CQ | \CR) \ge (- 2 \log \left( \max_{p \in \CP, q \in \CQ} \norm{p u_1 u_2^* q^*} + 3|\CP|\epsilon \right) + O(\epsilon^{1/5}))(1-\epsilon_1).
    \end{equation*}
\end{theorem}

\begin{proof}
    We first prove the inequality in the case where there is no conditioning $\CR$ and no unitary elements $u_1$ and $u_2$. 
    The first step consists of applying the GNS construction \cite[thm. I.9.6]{davidsonC^algebrasExample1996} to describe the state $\rho$ using a $C^*$ representation $\pi$ of $\CA$ on a Hilbert space $\CH$ generated by a cyclic normed vector $\psi$. This gives
    \begin{equation*}
        \forall x \in \CA \quad \rho(x) = \sprod{\psi}{\pi(x) \psi}.
    \end{equation*}
    Under this description, the quantum entropy becomes
    \begin{align*}
        h^q(\rho, u_1 \cdot \CP) &= \sum_{p \in \CP} \eta(\rho(p^* p)) = \sum_{p \in \CP} \eta(\sprod{\psi}{\pi(p^* p) \psi}) \\
        &= \sum_{p \in \CP} \eta(\sprod{\pi(p) \psi}{\pi(p) \psi}) = \sum_{p \in \CP} \eta(\norm{\pi(p) \psi}^2).
    \end{align*}
    Since $\pi$ is a $*$-morphism between $C^*$ algebras, it satisfies $\norm{\pi(x)}_{\CB(\CH)} \le \norm{x}$ for any element $x \in \CA$ \cite[thm. I.5.5]{davidsonC^algebrasExample1996}. Thus, the inequality to prove becomes 
    \begin{equation}\label{ineq_thm_conditional_entropic_uncertainty_principle_approximate quantum_partitions}
        \sum_{p \in \CP} \eta(\norm{\pi(p)  \psi}^2) + \sum_{q \in \CQ} \eta(\norm{\pi(q) \psi}^2) \ge - 2 \log \left( \max_{p \in \CP, q \in \CQ} \norm{p q^*} + 3|\CP|\epsilon \right) + O(\epsilon^{1/5}).
    \end{equation}
    Remark that 
    \begin{equation*}
        \sum_{p \in \CP} \pi(p)^* \pi(p) = \pi(\sum_{p \in \CP} p^* p) = \pi(\mathbb{1} + r_\CP) \le (1+\epsilon) Id_\CH,
    \end{equation*}
    and similarly for $\CQ$. Therefore, the inequality \eqref{ineq_thm_conditional_entropic_uncertainty_principle_approximate quantum_partitions} is a direct application of the entropic uncertainty principle for approximate quantum partitions in the context of Hilbert space given in \cite[sec. 3.7 prop. 29]{nonnenmacherEntropyChaoticEigenstates2010}, with all the weights taken equal to $1$.

    We now add the conditioning. Let $r \in \CR$. An application of the entropic uncertainty principle to $\rho_{|r}$, $\CP$ and $\CQ$ gives
    \begin{equation*}
        h^q(\rho_{|r}, \CP) + h^q(\rho_{|r}, \CQ) \ge - 2 \log \left( \max_{p \in \CP, q \in \CQ} \norm{p q^*} + 3|\CP|\epsilon \right) + O(\epsilon^{1/5}).
    \end{equation*}
    Multiplying these inequalities by $\rho(r^* r)$ and summing them leads to 
    \begin{align*}
        \sum_{r \in \CR} \rho(r^* r) \left(h^q(\rho_{|r}, \CP) + h^q(\rho_{|r}, \CQ) \right) &\ge (- 2 \log \left( \max_{p \in \CP, q \in \CQ} \norm{p q^*} + 3|\CP|\epsilon \right) + O(\epsilon^{1/5})) \sum_{r \in \CR} \rho(r^* r).
    \end{align*}
    Since $1- \epsilon_1 \le \sum_{r \in \CR} \rho(r^* r)$ as $\CR$ is a $(1+ \epsilon_1)$ quantum partition, we get
    \begin{align*}
        h^q(\rho,  \CP| \CR) + h^q(\rho, \CQ |CR) &\ge (- 2 \log \left( \max_{p \in \CP, q \in \CQ} \norm{p q^*} + 3|\CP|\epsilon \right) + O(\epsilon^{1/5})) (1-\epsilon_1).
    \end{align*}

    It is straightforward to check that $u_1 \cdot \CP$ and $u_2 \cdot \CQ$ are $(1+\epsilon)$ quantum partitions. Applying the inequality just proven to them gives the theorem.  
\end{proof}

\section{Quantization on the torus}\label{sec_quantization_torus}

We recall in this section the usual quantization procedure on the torus as introduced in \cite{hannayQuantizationLinearMaps1980a} and follow mostly the presentation of \cite{riviereEntropySemiclassicalMeasuresSymplecticMapTorus2011} and \cite{dyatlovSemiclassicalMeasuresHigherdimensional2024} that are based on that of \cite{bouzouinaEquipartitionEigenfunctionsQuantized1996a}.

\subsection{Weyl quantization and quantum translations}

Let $h >0$ be a semiclassical parameter. The Weyl quantization of a classical observable $a \in \mathscr{S}(\R^{2d})$ is defined as 

\begin{equation*}
    \forall u \in L^2(\R^d), \; \forall x \in \R^d  \quad \left(\Op^w_h(a)u \right) (x) = \frac{1}{(2\pi h)^d} \int_{\R^{2d}} e^{\frac{i}{h}\sprod{x-y}{\xi}} a\left(\frac{x+y}{2}, \xi \right) u(y) \diff y \diff \xi.
\end{equation*}

\noindent This quantization procedure can be extended to a large class of symbols $S_\nu (1)$, defined for $\nu \in [0, 1/2]$ by
\begin{equation*}
    S_\nu (1) = \set{(a_h)_{h >0} \in \CC^\infty(\R^{2d}) \; | \; \forall \alpha \in \N^{2d}, \; \norm{\partial^\alpha a}_{\infty} = \CO(h^{-\nu|\alpha|}) },
\end{equation*}
on which it satisfies 
\begin{itemize}
    \item a composition rule: if $a,b \in S_\nu(1)$, then there exists $c \in S_\nu(1)$ such that $c = a\#b$ is the Moyal product of $a$ and $b$ and it satisfies
    \begin{equation}\label{Moyal_product_R2d}
        \Op^w_h(a) \Op^w_h(b) = Op^w_h(c) \qqtqq{where} c = ab + \frac{i}{2h}\poissonbra{a}{b} + \CO_{S_\nu}(h^{1-2\delta}),
    \end{equation}
    \item the Calderon-Vaillancourt theorem: there exists $C,M >0$ such that for all symbols $a \in S_\nu(1)$
    \begin{equation}\label{Calderon-Vaillancourt_R2d}
        \norm{\Op^w_h(a)}_{L^2(\R^d \to \R^d)} \le C \sum_{|\alpha| \le M} h^{|\alpha|/2} \norm{\partial^\alpha a}_\infty,
    \end{equation}
    \item a G\r{a}rding inequality: for any $a \in S(1)$ with $a \ge 0$, there exists $C>0$ and $h_0 >0$ such that 
    \begin{equation}\label{ineq_Garding_R2d}
        \forall u \in L^2(\R^{d}), \; \forall 0 < h < h_0 \qquad \sprod{\Op^w_h(a)u}{u} \ge -Ch \norm{u}^2_{L^2(\R^d)},
    \end{equation}
\end{itemize}
The Poisson bracket between two classical observable $a,b \in \CC^\infty(\R^{2d})$ is defined by 
\begin{equation*}
    \poissonbra{a}{b} = \sum_{i=1}^{d} \pardif{a}{\xi_i} \pardif{b}{x_i} - \pardif{a}{x_i}\pardif{b}{\xi_i}.
\end{equation*}
We refer the reader to \cite[chap. 4]{zworskiSemiclassicalAnalysis2012} for the proofs of these statements and more information about the Weyl quantization. The product of the Weyl quantization of two symbols with disjoint support is controlled by the following lemma.

\begin{lemma}\label{lem_disjoint_support_weyl_quantization_R2d}
    Let $a,b \in S_{\nu}$ with $\nu \le 1/2$ and suppose that 
    \begin{equation*}
        d(\supp(a), \supp(b)) = \gamma > 0.
    \end{equation*}
    Then, for $h \le 1$ and denoting $L = \ceil{\frac{2d}{1-\nu}} $, we have
    \begin{equation*}
        \Op^w_h(a) \Op^w_h(b) =  \CO_{S_\nu}(\gamma^{-L+4d} \left(\frac{h^{1-\nu}}{\gamma} \right)^{\infty}),
    \end{equation*}
    where the implied constants depend on the $S_\nu$ seminorms of $a$ and $b$.
\end{lemma}

\begin{proof}
    We follow the proof of \cite[thm. 4.25]{zworskiSemiclassicalAnalysis2012} and keep track of the constant $\gamma$. The Moyal product of $a$ and $b$ can be expressed \cite[thm. 4.11]{zworskiSemiclassicalAnalysis2012} as 
    \begin{equation*}
        a \# b(z) = \frac{1}{(\pi h)^{2d}} \int_{\R^{4d}} e^{\frac{i}{h}\phi(w_1,w_2)} a(z-w_1)(b(z-w_2)) \diff w_1 \diff w_2 ,
    \end{equation*}
    for $z \in \R^{2d}$ and $\phi(w_1,w_2) = -2 \sigma(w_1,w_2)$. Define the operator 
    \begin{equation*}
        L := \frac{\sprod{\nabla \phi}{hD}}{\norm{\nabla \phi}^2} = -\frac{h}{2i} \frac{\sprod{J^T w_2}{\partial_{w_1}}+ \sprod{J w_1}{\partial_{w_2}}}{\norm{w_1}^2 + \norm{w_2}^2} ,
    \end{equation*}
    as $\nabla \phi(w_1,w_2) = -2 (J^T w_2, J w_1)$. The main point of this operator is that 
    \begin{equation*}
        L e^{\frac{i}{h}\phi(w_1,w_2)} = e^{\frac{i}{h}\phi(w_1,w_2)}.
    \end{equation*}
    In particular, for any $M \ge 0$ and $z \in \R^{2d}$, we have 
    \begin{equation}\label{eq_lem_disjoint_support_weyl_quantization_R2d_moyal_prod_with_L}
        a \# b(z) = \frac{1}{(\pi h)^{2d}} \int_{\R^{4d}} L^M(e^{\frac{i}{h}\phi(w_1,w_2)}) a(z-w_1)(b(z-w_2)) \diff w_1 \diff w_2 .
    \end{equation}
    Define the adjoint $L^*$ of $L$ by  
    \begin{equation*}
        \forall f,g \in \CC_c(\R^{4d} \setminus \set{0}), \quad \int_{\R^{4d}} L(f) g \diff x = \int_{\R^{4d}} f L^*(g) \diff x.
    \end{equation*}
    One readily checks that 
    \begin{equation}\label{eq_lem_disjoint_support_weyl_quantization_R2d_expr_adjoint}
        L^*(f) = h \frac{P_1(w)}{\norm{w}^4} f + h \frac{\sprod{P_2(w)}{\nabla f}}{\norm{w}^2} \qqtq{for} f \in \CC_c(\R^{4d} \setminus \set{0}),
    \end{equation}
    where $P_1$ is a polynomial of degree $3$ and $P_2$ a polynomial of degree $1$ in $w$. This holds for any smooth function that vanishes in a neighborhood of $0$. The condition on the support of $a$ and $b$ implies that 
    \begin{equation*}
        \forall (w_1, w_2) \in \R^{4d}, \quad a(z-w_1)b(z-w_2) \neq 0 \implies \norm{w_1 - w_2} > \gamma  \implies \norm{w} \ge \frac{\gamma}{\sqrt{2}} .
    \end{equation*}
    Hence, the function $w \mapsto a(z-w_1)b(z-w_2)$ vanishes in a neighborhood of $0$. Multiple applications of the relation \eqref{eq_lem_disjoint_support_weyl_quantization_R2d_expr_adjoint} and the fact that $a,b \in S_{\nu}$ lead to
    \begin{equation*}
        (L^*)^M (a(z-w_1)b(z-w_2)) = \CO_{S_\nu}(h^{M(1 - \nu)} \norm{w}^{-M} \mathbb{1}_{\norm{w} \ge \frac{\gamma}{\sqrt{2}}}),
    \end{equation*}
    where the implied constant depends on $M$ and the $S_\nu$ seminorms of $a$ and $b$. 
    Combined with the equality \eqref{eq_lem_disjoint_support_weyl_quantization_R2d_moyal_prod_with_L} and the definition of $L^*$, this gives 
    \begin{align*}
        |a \# b(z)| &\le \CO_{S_\nu}(h^{M(1-\nu) - 2d}) \int_{\R^{4d}, \norm{w} \ge \frac{\gamma}{\sqrt{2}}} \frac{1}{\norm{w}^{M}} \diff w \\
        &\le \CO_{S_\nu}(\frac{h^{M(1-\nu) - 2d}}{\gamma^{M-4d}}) \int_{\R^{4d}, \norm{w} \ge \frac{1}{\sqrt{2}}} \frac{1}{\norm{w}^{M}} \diff w.
    \end{align*}
    Thus, denoting $L = \ceil{\frac{2d}{1-\nu}} $ and as $h \le 1$, we get for any $n \ge 1$ that 
    \begin{equation*}
        \norm{a \# b}_\infty \le \frac{1}{\gamma^{L-4d}}\CO_{S_\nu}\left(\left(\frac{h^{1-\nu}}{\gamma}\right)^{n} \right).
    \end{equation*}
    A similar estimate holds for the derivatives of $a\#b$. Therefore, $a \# b = \CO_S(\gamma^{-L+4d}(\frac{h^{1-\nu}}{\gamma})^{\infty})$ and we deduce from the Calderón-Vaillancourt theorem \eqref{Calderon-Vaillancourt_R2d} that 
    \begin{equation*}
        \Op^w_h(a) \Op^w_h(b) = \Op^w_h(a\#b) =  \CO_{S_\nu}\left(\gamma^{-L+4d}\left(\frac{h^{1-\nu}}{\gamma}\right)^{\infty}\right),
    \end{equation*}
    where the implied constants depend on the $S_\nu$ seminorms of $a$ and $b$. 
\end{proof}

We now turn to the Heisenberg group and the quantum translations. The Heisenberg group $Heis$ of dimension $2d+1$ is $\R^{2d+1}$ endowed with the group law 
\begin{equation*}
    (x,\xi,t) \cdot (y, \eta,s) = (x+y, \xi +\eta, t + s + \frac{1}{2} \sigma((x,\xi), (y,\eta))) \qtq{for} (x,\xi, t), (y, \eta,s) \in \R^{2d+1}
\end{equation*}
The quantum translations $(U_h(z))_{z \in \R^{2d}}$ are unitary operators on $L^2(\R^d)$ that act on functions by 
\begin{equation}\label{eq_action_quantum_translation_on_functions}
    \forall \psi \in \R^\R, \, \forall y \in \R^d \quad (U_h(x,\xi) \psi)(y) = e^{\frac{i}{h}\sprod{y-x/2}{\xi}} \psi(y-x).
\end{equation}
They do not commute but they satisfy for $z, z' \in \R^{2d}$ the following relations
\begin{equation}\label{eq_commutation_quantum_translation}
    U_h(z) U_h(z') = e^{\frac{i}{2h}\sigma(z, z')} U_h(z + z') = e^{\frac{i}{h}\sigma(z, z')} U_h(z') U_h(z).
\end{equation}
Hence, they generate a representation of the Heisenberg group $Heis$ given by 
\begin{equation*}
    \rho_h : (x, \xi, t) \in Heis \mapsto e^{\frac{i}{h}t} U_h(x,\xi) .
\end{equation*}

The family of representations $(\rho_h)_{h >0}$ is the set of irreducible representations of the Heisenberg group which are non-trivial on the center. This is the content of the Stone-Von Neumann theorem and we refer the reader to \cite[chap. 1.5]{follandHarmonicAnalysisPhase1989a} for more details\footnote{We warn the reader that Folland chooses a different convention for the Schrödinger representation, see \cite[chap. 1, sec. 2, eq (1.27)]{follandHarmonicAnalysisPhase1989a} for a remark on this choice.}.

Finally, we remark that
\begin{equation}\label{eq_U_h(z)=quantif_e(isigma(z,.))}
    U_h(z) = \Op^w_h(e^{\frac{i}{h}\sigma(z, \cdot)}) \qqtq{for} z \in \R^{2d},
\end{equation}
and that the quantum translations satisfy an exact Egorov theorem
    \begin{equation}\label{eq_exact_Egorov_quantum_transla}
        U_h(z) \Op^w_h(a) U_h(z)^{-1} = \Op^w_h(a(\cdot - z)) \qqtq{for} z \in \R^{2d}, \; a \in S(1).
    \end{equation}

\subsection{Weyl quantization on the torus}

The Weyl quantization on the torus $\T^{2d}$ is built from the Weyl quantization on $\R^{2d}$. The first step is to define the Hilbert space of states. Intuitively, our quantum states should be invariant under quantum phase-space translations of length $1$ since the classical phase space is invariant under such classical translations. Since two quantum states that differ by a global phase are physically identical, we search for tempered distributions $\psi \in \mathscr{S}'(\R^d)$ such that 
\begin{equation*}
    \forall w \in \Z^{2d} ,\; \exists \theta \in \R \quad U_h(w) \psi = e^{i \theta} \psi.
\end{equation*}

A necessary condition for such a tempered distribution $\psi$ to exist is that all the operators $(U_h(w))_{w \in \Z^{2d}}$ must commute. According to the commutation relations \eqref{eq_commutation_quantum_translation} of the $U_h$, this implies that there exists an integer $N \in \N^*$ such that
\begin{equation*}
    \frac{1}{h} = 2 \pi N.
\end{equation*}
We suppose in the following that such a choice of $h >0$ has been made and we define 
\begin{equation*}
    \CH_N(\theta) = \set{\psi \in \mathscr{S}'(\R^d) \; | \; \forall w \in \Z^{2d}\; U_h(w) \psi = e^{2 \pi i \sigma(\theta,w) + \pi i N Q(w) } \psi},
\end{equation*}
where $Q : \R^{2d} \to \R$ is the polarization defined by
\begin{equation*}
    Q(z) = \sprod{y}{\eta} \qqtqq{if} z = (y, \eta) \in \R^{2d}.
\end{equation*}

The important properties of the spaces $\CH_N(\theta)$ are summarized in the following proposition \cite{dyatlovSemiclassicalMeasuresHigherdimensional2024}.

\begin{proposition}\label{pro_proprietes_CH_N(theta)}
    Let $\theta = (\theta_x, \theta_\xi) \in \T^{2d}$. The dimension of $\CH_N(\theta)$ is $N^d$ and a basis is given by $(e^\theta_j)_{j \in \eint{0}{N-1}^d}$ where
    \begin{equation*}
        e^\theta_j(x) = \frac{1}{N^{d/2}} \sum_{k \in \Z^d} e^{2 \pi i \sprod{k}{\theta_\xi}} \delta\left(x - \frac{j+kN - \theta_x}{N}\right).
    \end{equation*}
    If the scalar product on $\CH_N(\theta)$ is fixed such that $(e^\theta_j)_{j \in \eint{0}{N-1}^d}$ is an orthonormal basis, there exists a unitary isomorphism $\Phi_N : L^2(\R^d) \simeq L^2(\T^{2d}, \CH_N(\theta))$ whose expression on $\mathscr{S}(\R^d)$ is 
    \begin{equation*}
        \Phi_N(f)(\theta) = \sum_{j \in \eint{0}{N-1}^d} \sprod{f}{e^\theta_j}_{L^2} e^\theta_j \qquad f \in \mathscr{S}(\R^d), \; \theta \in \T^{2d} .
    \end{equation*}
    Furthermore,
    \begin{equation*}
        U_h(z) \CH_N(\theta) = \CH_N(\theta - Nz).
    \end{equation*}
\end{proposition}
The space $\CH_N(\theta)$ will therefore be endowed with the scalar product making $(e^\theta_j)_{j\in \eint{0}{N-1}^d}$ an orthonormal basis. The existence of the unitary isomorphism $\Phi_N$ is equivalent to the following integral decomposition of $L^2(\R^d)$ 
\begin{equation}\label{integral_decomposition_L2_Rn}
    L^2(\R^d) = \int^{\oplus}_{\theta \in \T^{2d}} \CH_N(\theta) \diff \theta.
\end{equation}

The last part of the proposition implies that $U_h(n/N) \CH_N(\theta) = \CH_N(\theta)$ for any $n \in \Z^{2d}$. We denote by $U_{N, \theta}(n/N)$, or $U_N(n/N)$ if the value of $\theta$ is clear from the context, the restriction of $U_h(n/N)$ to $\CH_N(\theta)$. 

We now turn to the quantization of observables on the torus. Any function $a \in \CC^\infty(\T^{2d})$ can be seen as a symbol in $S^0_0(1)$ on which we can use the Weyl quantization. Using a Fourier decomposition of $a$ and the formula \eqref{eq_U_h(z)=quantif_e(isigma(z,.))}, we get 
\begin{equation}\label{eq_quantization_obs_torus}
    \Op^w_h(a) = \sum_{n \in \Z^{2d}} a_n U_h(\frac{n}{N}) \qqtqq{where} a(x) = \sum_{n \in \Z^{2d}} a_n e^{2 i \pi \sigma(n,x)}.
\end{equation}
The series of operators converges since $a \in \CC^{\infty}(\T^{2d})$ and thus its Fourier coefficients satisfy $|a_n| = \CO(n^{-\infty})$. Since the operators $U_h(n/N)$ stabilize $\CH_N(\theta)$, the Weyl quantization of $a$ can be restricted to those spaces. This restriction is denoted $\Op^w_{N, \theta}(a)$ or $\Op^w_N(a)$ if the value of $\theta$ is clear. 

The composition rule \eqref{Moyal_product_R2d} for the Weyl quantization on $\R^{2d}$ descends to a similar one for the quantization on $\T^{2d}$ and the direct integral decomposition \eqref{integral_decomposition_L2_Rn} of $L^2(\R^d)$ implies \cite[thm. XIII.83]{reedMethodsModernMathematical4Operators1978}
\begin{equation}\label{ineq_bound_norms_pseudo_on H_N_via_L2_norms}
    \sup_{\theta \in \T^{2d}} \norm{\Op^w_{N, \theta}(a)}_{\CL(\CH_N(\theta))} = \norm{\Op^w_h(a)}_{\CL(L^2(\R^d))}.
\end{equation}
This leads to a bound on the operator norm of $\Op^w_{N, \theta}(a)$ for $a \in \CC^\infty(\T^{2d}) \cap S_\nu(1)$ via the Calderón-Vaillancourt theorem \eqref{Calderon-Vaillancourt_R2d}
\begin{equation}\label{Calderon-Vaillancourt_torus}
    \norm{\Op^w_{N, \theta}(a)}_{\CL(\CH_N(\theta))} \le C \sum_{|\alpha| \le M} N^{-|\alpha|/2} \norm{\partial^\alpha a}_{\infty}.
\end{equation}
One deduces from \eqref{ineq_Garding_R2d} a G\r{a}rding inequality on $\CH_N(\theta)$: for any $a \in \CC^\infty(\T^{2d})$ with $a \ge 0$, there exists $C>0$ and $h_0 >0$ such that
\begin{equation}\label{ineq_Garding_torus}
    \forall u \in \CH_N(\theta), \; \forall 0 < h < h_0 \qquad \sprod{\Op^w_N(a)u}{u} \ge -\frac{C}{N} \norm{u}^2_{\CH_N(\theta)}.
\end{equation}

\subsection{Symplectic automorphism of the torus and metaplectic representation}

Let $R$ be a ring. A matrix $A \in M_{2d}(R)$ is symplectic if
\begin{equation*}
    A^T J A = J.
\end{equation*}
The set of $2d \times 2d$ symplectic matrices with coefficients in $R$ is denoted $\Sp_{2d}(R)$. If $A \in \Sp_{2d_1}(R)$ and $B \in \Sp_{2d_2}(R)$ are two symplectic matrices, the block diagonal matrix obtained by putting $A$ and $B$ on the diagonal is not symplectic in general\footnote{It would be symplectic with respect to the symplectic form $\sigma_{2d_1} \oplus \sigma_{2d_2}$ which is different from $\sigma_{2(d_1+d_2)}$.}. To solve this minor issue, we define $A \diamond B \in \Sp_{2(d_1 + d_2)}(R)$ to be
\begin{equation}\label{eq_definition_diamond_symplectic_matrices}
    A \diamond B = \begin{pmatrix}
        A_1 & 0 & B_1 & 0 \\
        0 & A_2 & 0 & B_2 \\
        C_1 & 0 & D_1 & 0 \\
        0 & C_2 & 0 & D_2 
    \end{pmatrix}
    \qqtq{where}
        A = \begin{pmatrix}
        A_1 & B_1 \\
        C_1 & D_1 \\
    \end{pmatrix} \qtq{and} B = \begin{pmatrix}
        A_2 & B_2 \\
        C_2 & D_2 \\
    \end{pmatrix} .
\end{equation}
One can check that $A \diamond B$ is indeed a symplectic matrix. We need a result on the reduction of real symplectic matrices. It is a consequence of the results in either \cite{laubCanonicalFormsSymplectic1974} or \cite[sec. 1.4 to 1.7]{longIndexTheorySymplectic2002a}.

\begin{theorem}\label{thm_reduction_symplectic_matrices_real}
    Let $A \in \Sp_{2d}(\R)$ and denote by $\lambda_1 < \cdots < \lambda_n$ the positive Lyapunov exponents of $A$. Then there exists $B \in \Sp_{2d}(\R)$, $A_i \in \Sp_{2d_i}(\R)$, $D_i \in \GL_{d_i}(\R)$ and $C \in \Sp_{2d'}(\R)$ such that 
    \begin{equation*}
        A = B^{-1}(A_1 \diamond \dots \diamond A_n \diamond C) B \qqtq{and} A_ i = \begin{pmatrix}
            D_i & 0 \\
            0 & (D_i^T)^{-1}
        \end{pmatrix},
    \end{equation*}
    such that the spectrum of $C$ is included in the unit circle and the spectrum of $D_i$ is included in the circle of radius $e^{\lambda_i}$.
\end{theorem}

We now turn to the definition of the metaplectic representation acting on $L^2(\R^d)$. Let $A \in \Sp_{2d}(\R)$ be a real symplectic matrix. The map defined by 
\begin{equation*}
    (z, t) \in Heis \mapsto (Az,t) \in Heis,
\end{equation*}
is an automorphism of the Heisenberg group that leaves the center $\set{(0,0,t) \in Heis \; | \: t \in \R}$ unchanged. We still denote it $A$ and use it to twist the irreducible representation $\rho_h$ to get the representation $\rho_h \circ A$. This new representation is still irreducible and acts on the center in the same manner as $\rho_h$. As a consequence of the Stone-von Neumann theorem \cite[thm. 1.50]{follandHarmonicAnalysisPhase1989a}, $\rho_h$ and $\rho_h \circ A$ are conjugate, meaning that there exists a unitary operator $M_h(A)$ on $L^2(\R^d)$ such that 
\begin{equation*}
    \forall z \in \R^{2d}, \; \forall t \in \R \qquad \rho_h (Az, t) = M_h(A) \rho_h(z,t) M_h(A)^{-1}.
\end{equation*} 
This implies 
\begin{equation}\label{eq_exact_Egorov_R2d}
    \forall z \in \R^{2d}, \qquad U_h (Az) = M_h(A) U_h(z) M_h(A)^{-1},
\end{equation} 
which is an exact Egorov theorem. In other words, the operator $M_h(A)$ is the quantum analog of the diffeomorphism $A$ acting on $\R^{2d}$. 
According to Schur's lemma, the operator $M_h(A)$ is unique up to a phase and the morphism 
\begin{equation*}
    A \in \Sp_{2d}(\R) \mapsto M_h(A) \in \CU(L^2(\R^d))/\set{cId \, | \,  |c|=1 },
\end{equation*}
is thus a projective representation of the symplectic group. This representation can be lifted into a representation of the metaplectic group, which is the double cover of the symplectic group, and is called the metaplectic representation. 

The metaplectic operators are compatible with the concatenation of symplectic matrices, meaning that 
\begin{equation}\label{eq_compatibility_metaplectic_op_concatenation_symplectic}
    M_h(A \diamond B) = M_h(A) \otimes M_h(B) \qqtq{for} A \in \Sp_{2d_1}(\R), \; B \in \Sp_{2d_2}(\R),
\end{equation}
where the operator $M_h(A) \otimes M_h(B)$ acts component by component on $L^2(\R^{d_1}) \otimes L^2(\R^{d_2})$ which is a dense subset of $ L^2(\R^{d_1+d_2})$. This equality can be checked by verifying it on generators of the symplectic group, see \cite[chap. 4, eq 4.24-4.26]{follandHarmonicAnalysisPhase1989a} for their expressions.

Since any matrix with integer coefficients stabilizes the torus $\T^{2d}$, any symplectic integer matrix $A \in \Sp_{2d}(\Z)$ induces an automorphism of $\T^{2d}$, still denoted $A$. It is quantized by restricting the quantum propagator $M_h(A)$ to $\CH_N(\theta)$ for a good choice of $\theta$. Such a choice can be made as long as $1$ is not an eigenvalue of $A$ according to the discussion in \cite[sec. 2.2.4]{dyatlovSemiclassicalMeasuresHigherdimensional2024}. 

\begin{lemma}\label{lem_condition_quantizability_symp_matrix}
    Let $A \in \Sp_{2d}(\Z)$ such that $1$ does not belong to the spectrum of $A$. Then, for any $N \in \N^*$, there exists an angle $\theta_N(A) \in \T^{2d}$ such that 
    \begin{equation*}
        M_h(A) \CH_N(\theta_N(A)) = \CH_N(\theta_N(A)).
    \end{equation*}
\end{lemma}

\begin{remark}\label{rem_CH_N_theta_invariant_if_N_even}
    If $N$ is an even integer, there is no restriction on $A$ and the quantum space $\CH_N(0)$ is always invariant under $M_N(A)$.
\end{remark}

The restriction of $M_h(A)$ to $\CH_N(\theta_N(A))$ is denoted $M_N(A)$. It satisfies an exact Egorov theorem with the Weyl quantization as a direct consequence of \eqref{eq_exact_Egorov_R2d}
\begin{equation}\label{eq_exact_Egorov_torus}
    \forall n \in \Z^{2d}, \qquad U_N (\frac{An}{N}) = M_N(A) U_N(\frac{n }{N}) M_N(A)^{-1},
\end{equation}
which, together with \eqref{eq_quantization_obs_torus}, implies
\begin{equation}\label{eq_exact_Egorov_observable_torus}
    \forall a \in \CC^\infty(\T^{2d}) \qquad M_N(A)^{-1} \Op^w_N(a) M_N(A) = \Op^w_N(a \circ A) .
\end{equation}

\subsection{Semiclassical measures}

The main objects of interest of this article are semiclassical measures. They are a set of measures on $\T^{2d}$ that inform us on the asymptotic behavior of a family of quantum states $(\psi_N)_{N} \in \CH_N(\theta_N)$ in the semiclassical limit $N \to \infty$. Given such a sequence of quantum states, we associate to it a sequence of distributions $Wig(\psi_N)_{N}$ called the Wigner distribution and given by
\begin{equation*}
    \forall a \in \CC^\infty(\T^{2d}) \qquad Wig(\psi_N)(a) = \sprod{\Op^w_N(a)\psi_N}{\psi_N}.
\end{equation*}
The Wigner distribution is a finite order distribution thanks to the Calderon-Vaillancourt theorem \eqref{Calderon-Vaillancourt_torus}
\begin{equation}\label{ineq_bound_Wigner_by_Calderon_Vaillancourt}
    |Wig(\psi_N)(a)| \le C \norm{\psi_N}^2  \sum_{|\alpha| \le M} N^{-|\alpha|/2} \norm{\partial^\alpha a}_{\infty}.
\end{equation} 
It is positive up to an error of order $1/N$ according to the G\r{a}rding  inequality \eqref{ineq_Garding_torus}
\begin{equation}
    \forall a \in \CC^{\infty}(\T^{2d}), \quad a\ge 0 \implies Wig(\psi_N)(a) \ge \CO_a(\frac{1}{N}).
\end{equation}

We suppose in the following that the quantum states $(\psi_N)$ are of norm $1$. The set of semiclassical measures associated to $(\psi_N)_{N}$ is the set of accumulation points of the sequence $(Wig(\psi_N))_{N}$, noted $\CS \CM((\psi_N)_{N})$. It is a non-empty compact subset of the set of probability measures on $\T^{2d}$. This is a consequence of the Banach-Alaoglu theorem which implies that the unit ball of the dual of $\CC^M(\T^{2d})$ is compact for the weak-star topology. 

Let $A \in \Sp_{2d}(\Z)$. A probability measure $\mu \in \CP(\T^{2d})$ is a semiclassical measure for $A$ if there exist sequences $(\theta_N)_{N}$ and $(\psi_N)_N$ such that $\CH_N(\theta_N)$ is invariant under $M_{(2 \pi N)^{-1}}(A)$, $\psi_N$ is a normed eigenvector of $M_N(A)$ and $\mu$ is a semiclassical measure for $(\psi_N)_N$, meaning that up to a subsequence
\begin{equation*}
    Wig(\psi_N) \xrightarrow[N \to + \infty]{} \mu.
\end{equation*}
The set of semiclassical measures associated to $A$ is denoted $\CS \CM(A)$. It satisfies properties similar to the set of semiclassical measures of the Laplacian. 
\begin{lemma}[Basic properties of $\CS \CM(A)$]\label{lem_basic_pro_semicla_meas_symp_matrix}
    Let $A \in \Sp_{2d}(\Z)$, then
    \begin{equation*}
        \emptyset \neq \CS \CM(A) \subset \CP_A(\T^{2d}),
    \end{equation*}
    where $\CP_A(\T^{2d})$ is the set of $A$-invariant probability measures on $\T^{2d}$.
    Moreover, $\CS \CM(A)$ is weak-star compact in the set of probability measures on $\T^{2d}$. Finally,  
    \begin{equation*}
        \CS \CM(A) = \CS \CM(A^{-1}) \qqtqq{and} \CS \CM(A^s) \subset \CS \CM(A^t) \qtq{if} 0 \le s \le t \in \N.
    \end{equation*}
\end{lemma}
\begin{proof}
    Denote by $E_N(A)$ the set of normed eigenvectors of $M_N(A)$ in $\CH_N(\theta)$ for any $\theta$ for which $\CH_N(\theta)$ is stable by $M_N(A)$. This set is non-empty if $N$ is even as $\CH_N(0)$ is then stable by $M_N(A)$ according to the remark \ref{rem_CH_N_theta_invariant_if_N_even} and the diagonalizability of unitary operators on finite dimensional Hilbert spaces. The definition of the set of semiclassical measures associated to $A$ and the usual set-theoretic description of accumulation points of sequences gives
    \begin{equation*}
        \CS \CM(A) = \cap_{k \ge 0} \overline{\cup_{N \ge k} \set{Wig(\psi) \, | \, \psi \in E_N(A)}}.
    \end{equation*}
    This is a decreasing intersection of non-empty compact sets in the dual of $\CC^M(\T^{2d})$ by Banach-Alaoglu and \eqref{ineq_bound_Wigner_by_Calderon_Vaillancourt}. Hence, $\CS \CM(A)$ is compact and non-empty in the dual of $\CC^M(\T^{2d})$. As any semiclassical measure is a probability measure, it is straightforward to see that $\CS \CM(A)$ is included in $\CP(\T^{2d})$ and compact for the weak-star topology on $\CP(\T^{2d})$.

    The invariance by $A$ of semiclassical measures associated to $A$ is a classical consequence of the Egorov theorem \eqref{eq_exact_Egorov_observable_torus}.
    The property for different powers of the matrix comes from 
    \begin{equation*}
        E_N(A) = E_N(A^{-1}) \qqtqq{and} E_N(A^s) \subset E_N(A^t),
    \end{equation*}
    which are immediate as $M_N(A^s) = \pm M_N(A)^s$ and an eigenvector for a given matrix is always an eigenvector for any power and multiple of this matrix. 
\end{proof}

\section{Effective Ehrenfest time}\label{sec_effective_ehrenfest_time}

Consider a symplectic matrix $A \in \Sp_{2d}(\R)$ and denote by $\lambda_{max}$ its largest Lyapunov exponent. Fix $\epsilon >0$ and suppose that $\lambda_{max} >0$. For a general observable $f \in S(\R^{2d})$, its evolution $f\circ A^t$ under the symplectic matrix $A$ remains uniformly in $S_{\frac{1-\epsilon}{2}}$ for times $|t| \le \frac{(1-\epsilon)|\log(h)|}{2\lambda_{max}}$. Indeed, suppose that $t \ge 0$, an application of Faà di Bruno's formula leads to 
\begin{align*}
    |\partial^\alpha f\circ A^t (x)| &\le C_\alpha \max_{\beta \le \alpha} |\norm{A^t}^\beta \partial^{\alpha - \beta} f (A^t x)| \le C_\alpha  \norm{A}^{|\alpha|t} \le C_\alpha  e^{\lambda_{max} |\alpha| t} \\
    &\le C_\alpha h^{- \frac{1-\epsilon}{2} |\alpha|}.
\end{align*}
The time $T^E_{min} = \frac{(1-\epsilon)|\log(h)|}{2\lambda_{max}}$, called the Ehrenfest time, puts a hard limit on the timescale at which the semiclassical approximation is valid. For example, a version of Egorov theorem in good symbol classes is only valid for the Anti-Wick quantization for times up to the Ehrenfest one. The case of metaplectic operators and of the Weyl quantization, for which there is an exact Egorov theorem, is the exception rather than the norm. 

Hoping to improve on this Ehrenfest time $T^E_{min}$ which is governed by the largest Lyapunov exponent of $A$, we define the notion of effective Ehrenfest time of an observable.

\begin{definition}\label{def_effective_ehrenfest_time}
    A function $f \in \CC^\infty(\T^{2d})$ is said to have effective Ehrenfest time $(T, \epsilon)$ with $0 < \epsilon < 1/2$ if
    \begin{enumerate}
        \item the family of functions $(f \circ A^t)_{|t| \le T}$ is uniformly in $S_{1/2-\epsilon}(1)$,
        \item for any $g \in \CC^\infty(\T^{2d})$ and any pair of times $t,s \in \Z$ such that $|t-s| \le 2T$ there exists a constant $C_{f,g}$ such that 
        \begin{equation*}
            \norm{[\Op^w_h(f \circ A^s), \Op^w_h(g \circ A^t)]} \le C_{f,g} h^\epsilon.
        \end{equation*}
    \end{enumerate} 
\end{definition}
Both conditions in this definition have been chosen because they are needed in our proof of the almost subadditivity of the quantum dynamical entropy later on.

\subsection{Building periodic functions with an improved effective Ehrenfest time}

The definition \ref{def_effective_ehrenfest_time} would be meaningless to our cause without a procedure to construct periodic functions with an Ehrenfest time larger than $T^E_{min}$, which is an effective Ehrenfest time for any observable. We need periodic functions since those are the classical observables on the torus. Fortunately, the next lemma gives such a procedure as long as $A$ has a non-trivial rational invariant subspace $V$. 

\begin{lemma}\label{lem_good_Ehren_time_obs_Fourier_support}
    Let $A \in \Sp_{2d}(\Z)$ and suppose that it admits a rational invariant subspace $V$. Denote by $\lambda > 0$ the largest Lyapunov exponent of $A$ restricted to $V$ and fix $0 < \epsilon < 1$. Then any function $f \in \CC^\infty(\T^{2d})$ that is invariant under $V^\sigma$ has an effective Ehrenfest time $(T_\lambda,\epsilon/2)$ where 
    \begin{equation*}
        T_\lambda := \frac{(1-\epsilon)|\log(h)|}{2\lambda}.
    \end{equation*}
\end{lemma}

\begin{proof}
    Since $f$ is smooth and periodic, it can be written as a Fourier series
    \begin{equation*}
        f(x) = \sum_{r \in \Z^{2d}} f_r e^{2 \pi i \sigma(r,x)}.
    \end{equation*}
    We use the non-degenerate form $\sigma$ to identify the dual of $\T^{2d}$ with $\Z^{2d}$ instead of the usual scalar product since it is better suited to our setting. 
    As $f$ is invariant under translation of any element $v \in V^\sigma$, we have 
    \begin{equation*}
        \sum_{r \in \Z^{2d}} f_r e^{2 \pi i \sigma(r,x)} = f(x) = f(x+v) = \sum_{r \in \Z^{2d}} f_r e^{2 \pi i \sigma(r,x + v)} = \sum_{r \in \Z^{2d}} f_r e^{2 \pi i \sigma(r,v)} e^{2 \pi i \sigma(r,x)}.
    \end{equation*}
    Therefore, the uniqueness of the Fourier coefficients implies that $f_r(1-e^{2 \pi i \sigma(r,v)}) = 0$. Since this is true for any $v \in V^\sigma$, we get that $f_r = 0$ except if $r \in (V^\sigma)^\sigma = V$. In other words $f$ has a Fourier support included in $V$ and
    \begin{equation*}
        f(x) = \sum_{r \in \tilde{V}} f_r e^{2 \pi i \sigma(r,x)} \qqtqq{where} \tilde{V} = V \cap \Z^{2d}.
    \end{equation*} 

    We need to bound the supremum norm of $\partial^\alpha (f \circ A^t)$. Let $x \in \R^{2d}$, $\alpha \in \N^{2d}$ and $t \in \Z$.
    \begin{align*}
        f(A^t x) & = \sum_{r \in \tilde{V}} f_r e^{2i \pi \sigma(r,A^t x)} = \sum_{r \in \tilde{V}} f_r e^{2i \pi \sigma(A^{-t}r,x)} = \sum_{r \in \tilde{V}} f_r e^{2i \pi \sprod{JA^{-t}r}{x}}  \\
        \partial^\alpha (f\circ A^t) (x) & = \sum_{r \in \tilde{V}} f_r (2i \pi JA^{-t}r)^{\alpha} e^{2i \pi \sprod{A^{-t}r}{x}}.
    \end{align*}

    As $\lambda$ is the largest Lyapunov exponent of $A$ restricted to $V$, the spectrum of $A$ restricted to $V$ is included in a disk centered at $0$ and of radius $e^{\lambda}$. This implies that for any $\delta >0$, there exists a constant $C_\delta >0$ such that
    \begin{equation}\label{ineq_proof_effect_Ehrenfest_time_bound_norm_A_restricted_V}
        \forall t \in \Z, \quad \norm{A_{|V}^t} \le C |t|^{2d} e^{|t| \lambda} \le C_\delta e^{|t| (\lambda + \delta)}.
    \end{equation}
    The polynomial factor in the bound is there to take into account Jordan block of size $2$ or more. The small margin $\delta$ allows us to absorb this factor. We take $\delta >0$ such that 
    \begin{equation*}
        (1-\epsilon) \frac{\lambda + \delta}{\lambda} \le 1- \frac{\epsilon}{2}.
    \end{equation*}

    If $|t| \le T_\lambda$ and $h \le 1$, we have by \eqref{ineq_proof_effect_Ehrenfest_time_bound_norm_A_restricted_V} that
    \begin{align*}
        |\partial^\alpha (f\circ A^t) (x)| & \le C_\alpha \sum_{r \in \tilde{V}} |f_r| |(JA^{-t}r)^{\alpha}| \le C_\alpha \sum_{r \in \tilde{V}} |f_r| \norm{A^{-t}r}^{|\alpha|} \le C_\alpha \sum_{r \in \tilde{V}} |f_r| \norm{r}^{|\alpha|} e^{|t| (\lambda + \delta) |\alpha|} \\
        & \le C_{\alpha,f} e^{|t|(\lambda + \delta) |\alpha|} \le C_{\alpha,f} h^{-(1-\epsilon) \frac{\lambda + \delta }{\lambda}|\alpha|/2} \le C_{\alpha,f} h^{-(1-\epsilon/2)|\alpha|/2}.
    \end{align*}
    We used that the Fourier coefficients of a $\CC^\infty$ function decrease quicker than any polynomial to bound $\sum_{r \in \tilde{V}} |f_r| \norm{r}^{|\alpha|}$.
    
    We turn to the second part of the lemma and fix a function $g \in \CC^\infty(\T^{2d})$ and two times $t,s \in \Z$. Recall from \eqref{eq_commutation_quantum_translation} and \eqref{eq_exact_Egorov_R2d} that for any $x,x' \in \R^{2d}$ and $t \in \Z$, the quantum translations satisfy
    \begin{equation*}
        M_h(A)^{t} U_h(x) M_h(A)^{-t} = U_h(A^t x) \qquad U_h(x)U_h(x') = e^{\frac{i}{2h}\sigma(x,x')} U_h(x + x').
    \end{equation*}
    This implies that 
    \begin{align*}
        \Op^w_h(g)(t) & = M_h(A)^{-t} \Op^w_h(g) M_h(A)^{t} = \sum_{r \in \Z^{2d}} g_r M_h(A)^{-t} U_h(2 \pi h r) M_h(A)^{t} \\
        & = \sum_{r \in \Z^{2d}} g_r U_h(2 \pi h A ^{-t}r).
    \end{align*}
    We turn to the commutator.
    \begin{align*}
        [\Op^w_h(g)(t), \Op^w_h(f)(s)] &= \sum_{r \in \Z^{2d}, l \in \tilde{V}} g_r f_l \left( U_h(2 \pi h A^{-t}r) U_h(2 \pi h A^{-s}l) -  U_h(2 \pi h A^{-s}l) U_h(2 \pi h A^{-t}r)\right) \\
        &= \sum_{r \in \Z^{2d}, l \in \tilde{V}} g_r f_l \left( e^{2 i\pi^2 h\sigma(A^{-t}r,A^{-s}l)} - e^{2 i\pi^2 h\sigma(A^{-s}l,A^{-t}r)} \right) U_h(2 \pi h (A^{-t}r + A^{-s}l)) \\
        &= \sum_{r \in \Z^{2d}, l \in \tilde{V}} g_r f_l 2i \sin(2 \pi^2 h\sigma(A^{-t}r,A^{-s}l)) U_h(2 \pi h (A^{-t}r + A^{-s}l)).
    \end{align*}
    Let $r \in \Z^{2d}$ and $l \in \tilde{V}$. As $|t-s| \le 2 T_\lambda$, the bound \eqref{ineq_proof_effect_Ehrenfest_time_bound_norm_A_restricted_V} implies
    \begin{equation*}
        \norm{A^{t-s}l} \le C e^{|t-s| (\lambda + \delta) } \norm{l} \le C e^{2T_\lambda (\lambda + \delta)} \norm{l} = C e^{2\frac{(1-\epsilon)|\log(h)|}{2 \lambda}(\lambda + \delta)} \norm{l} = C h^{-(1-\epsilon/2)} \norm{l}.
    \end{equation*}
    Therefore, 
    \begin{align*}
        |2i \sin(2 \pi^2 h\sigma(A^{-t}r,A^{-s}l))| &\le 4 \pi^2 h |\sigma(A^{-t}r,A^{-s}l)| \le 4 \pi^2 h |\sigma(r,A^{t-s}l)| \\
        &\le 4 \pi^2 h \norm{r} \norm{A^{t-s}l} \le 4 \pi^2 h \norm{r} C \norm{l} h^{-(1-\epsilon/2)} \\
        &\le C' \norm{r} \norm{l} h^{\epsilon/2}.
    \end{align*}
    Since $\norm{U_h(x)}_{L^2 \to L^2} = 1$, we get 
    \begin{align*}
        \norm{[\Op^w_h(g)(t), \Op^w_h(f)(s)]} &\le 2 \sum_{r \in \Z^{2d}, l \in \tilde{V}} |g_r f_l| |\sin(2 \pi^2 h\sigma(A^{-t}r,A^{-s}l))| \\
        &\le C h^{\epsilon/2} \sum_{r \in \Z^{2d}, l \in \tilde{V}} |g_r f_l| \norm{r} \norm{l}  \\
        &\le C_{f,g} h^{\epsilon/2}.
    \end{align*}
    
\end{proof}

\subsection{Almost subadditivity of quantum dynamical entropy up to an effective Ehrenfest time}

Recall from Section \ref{sec_classical_to_quantum_partitions} that given a smooth partition of unity $\CQ = (q_i^2)_{i \in I}$ of $\T^{2d}$, we denote by $\CQ^q$ the approximate quantum partition obtained by restricting the operators $Q_i = \Op^w_h(q_i)$ to $\CH_N(\theta)$.  This is again a $(1+\CO(h))$ quantum partition thanks to \eqref{ineq_bound_norms_pseudo_on H_N_via_L2_norms}. For simplicity, we may denote the approximate quantum partition associated to $\CQ$ again by $\CQ$ if the context is clear.
Finally, we denote for $t \in \Z$
\begin{equation*}
    \CQ^q(t) := M_h(A)^t \cdot \CQ^q = \set{M_h(A)^{-t} Q_i M_h(A)^t \; | \; i \in I}.
\end{equation*}
We prove in this section that the conditional quantum dynamical entropy of partitions of the form $\CQ^q$ behaves well in a time range of length at most twice the effective Ehrenfest time of the classical partition involved. This stems from the properties of the quantum entropy and from the bound on the commutators of evolved quantum partitions proved in the next lemma. 

\begin{lemma}\label{lem_bound_commutator_quantum_part_up_to_effective_Ehrenfest_time}
    Let $\CQ = (q_i^2)_{i \in I}$ be a smooth partition of $\T^{2d}$ such that each $q_i$ has an effective Ehrenfest time of $(T, \epsilon)$ and $\CP_1, \dots, \CP_m$ be smooth partitions of $\T^{2d}$. Then for any $(s_1, \dots, s_m) \in \Z^m$, $h\le 1$, $\theta \in \T^{2d}$ and $t \in \Z$ satisfying $|t-s_j| \le 2 T$, we have
    \begin{equation}
        \norm{[\CQ^q(t), \bigvee_{j=1}^m \CP^q_j(s_j)]}_{\CH_N(\theta)} = \CO(h^\epsilon).
    \end{equation}
    The implicit constants depend on $|\CQ|$, $|\CP_j|$, $m$ and the $S(1)$ seminorms of the elements of the classical partitions and neither on the value of $t, s_j$ in the specified range, nor on $\theta \in \T^{2d}$.
\end{lemma}

\begin{proof}
    Recall from \eqref{eq_def_error_terms_commutators} that 
    \begin{equation*}
        \norm{[\CQ^q(t), \bigvee_{j=1}^m \CP^q_j(s_j)]} = \sum_{j=1}^m \sum_{P_j \in \CP^q_j} \norm{\sum_{Q \in \CQ^q} [Q(t)^*, P_m(s_m)^* \dots  P_1(s_1)^* P_1(s_1) \dots P_m(s_m)]Q(t)}.
    \end{equation*}
    Thanks to the inequality \eqref{ineq_bound_norms_pseudo_on H_N_via_L2_norms}, it is enough to bound the norms of the operators on $L^2(\R^d)$ and thus to bound 
    \begin{equation*}
        \sum_{j=1}^m \sum_{P_j \in \CP^q_j} \norm{\sum_{Q \in \CQ^q} [Q(t)^*, P_m(s_m)^* \dots  P_1(s_1)^* P_1(s_1) \dots P_m(s_m)]Q(t)}_{L^2(\R^d)}.
    \end{equation*}
    As the approximate quantum partitions $\CQ^q$, $\CP^q_1, \dots, \CP^q_m$ come from classical partitions of $\T^{2d}$, they are $(1+\CO(h))$ quantum partitions, where the term $\CO(h)$ depends on the $S(1)$ semi-norms of the elements of the classical partitions. Therefore, this term is uniform as the number of partitions and their cardinality is independent of $h$. This implies a uniform bound on the norm of the elements of the approximate partitions
    \begin{equation*}
        \max_{Q \in \CQ^q, P_j \in \CP^q_j} (\norm{Q}, \norm{P_j}) \le 1 + \CO(h).
    \end{equation*}
    The Lie bracket satisfies the relation 
    \begin{equation*}
        [A,BC] = [A,B]C + B[A,C] \qqtq{for} A,B,C :L^2(\R^d) \to L^2(\R^d).
    \end{equation*}
    Repeated application of this equality and the uniform bound on the operator norms implies that each term
    \begin{align*}
        \norm{\sum_{Q \in \CQ^q} [Q(t)^*, P_m(s_m)^* \dots  P_1(s_1)^* P_1(s_1) \dots P_m(s_m)]Q(t)} ,
    \end{align*}
    is bounded by 
    \begin{equation*}
        |\CQ| 2m(1+\CO(h))^{2m+1} \max_{Q \in \CQ^q, P_j \in \CP^q_j} (\norm{[Q(t)^*,P_j(s_j)^*]},\norm{[Q(t)^*,P_j(s_j)]}).
    \end{equation*}
    Since any element $q$ of the classical partition $\CQ$ has an effective Ehrenfest time $(T, \epsilon)$, there exists $\epsilon >0$ such that 
    \begin{equation*}
        \forall |t_1-t_2| \le 2T, \; \forall p \in S(1) \quad \norm{[\Op^w_h(q)(t_1), \Op^w_h(p)(t_2)]}_{L^2(\R^d)} = \CO_{p,q}(h^\epsilon),
    \end{equation*}
    where the constants depend on the $S(1)$ seminorms of $p$ and $q$. This relation is also valid for $\bar{q}$, because if $q$ has Ehrenfest time $T$ then $\bar{q}$ has the same Ehrenfest time with the same $\epsilon$, as $\Op^w_h(q)^* = \Op^w_h(\bar{q})$. 
    As $|t-s_j| \le 2T$, we thus have that 
    \begin{align}
        \norm{[Q(t)^*,P_j(s_j)^*]} &= \norm{[\Op^w_h(\bar{q})(t), \Op^w_h(\bar{p_j})(s_j)]} = \CO_{q, p_j}(h^\epsilon).
    \end{align}
    A similar bound holds for $\norm{[Q(t)^*,P_j(s_j)]}$, hence
    \begin{align*}
        \norm{[\CQ^q(t), \bigvee_{j=1}^m \CP^q_j(s_j)]} \le \prod_{j=1}^m |\CP_j| |\CQ| 2m(1+\CO(h))^{2m+1} \CO(h^\epsilon).
    \end{align*}
    We conclude as $2m(1+\CO(h))^{2m+1} \le C_m$ if we only consider $h \le 1$ where $C_m>0$ is a constant depending only on $m$ and the $S(1)$ seminorms of the involved classical observables.
\end{proof}

We deduce from Lemmas \ref{lem_bound_commutator_quantum_part_up_to_effective_Ehrenfest_time} and \ref{lem_almost_monotony_by_deconditioning_quantum_entropy} that the quantum dynamical entropy is subadditive up to the effective Ehrenfest time of the partition considered. 

\begin{lemma}\label{lem_subadditivity_quantum_part}
    Let $\CR_h$ be a $(1+ \CO(h))$ quantum partition, $\CQ = (q_i^2)_{i \in I}$ be a smooth partition of $\T^{2d}$ with an effective Ehrenfest time of $(T, \epsilon)$ and $\rho$ be a state on $\CH_N$ invariant under $M_N(A)$. Suppose that there exists $K>0$ such that $T \le K |\log h|$, then for any $m_0 \in \N^*$ and any $-T \le t_1 \le t_2 \le T$, there exists $h_0$ such that for any $h\le h_0$
    \begin{equation*}
        h^q(\rho, \bigvee_{k = t_1}^{t_2} \CQ^q(k) | \CR) \le \sum_{l=0}^{n-1} h^q(\rho, \bigvee_{k=0}^{m_0-1} \CQ^q(k) | \CR (-t_1 - l m_0)) + C,
    \end{equation*}
    where $n := \floor{(t_2-t_1)/m_0}$ and $C >0$ depend only on $|\CQ|$, $m_0$, $K$ and the $S(1)$ seminorms of the elements of the classical partition $\CQ$. 
\end{lemma}

\begin{proof}
    Thanks to \eqref{eq_action_unitary_conditional_quantum_entropy} and the invariance of $\rho$ under $M_N(A)$, we get
    \begin{align*}
        h^q(\rho, \bigvee_{k = t_1}^{t_2} \CQ(k) | \CR) &= h^q(\rho, \bigvee_{k = 0}^{t_2-t_1} \CQ(k+t_1) | \CR) = h^q(\rho, M_N(A)^{t_1} \cdot \bigvee_{k = 0}^{t_2-t_1} \CQ(k) | \CR) \\
        &= h^q(M_N(A)^{t_1} \cdot \rho, \bigvee_{k = 0}^{t_2-t_1} \CQ(k) | M_N(A)^{-t_1} \cdot \CR) \\
        &= h^q(\rho, \bigvee_{k = 0}^{t_2-t_1} \CQ(k) | \CR(-t_1)).
    \end{align*}
    Remark that for any $l \le T$, $\bigvee_{k=0}^l \CQ(k)$ is a $(1+\CO(\sqrt{h}))$ quantum partition, and uniformly in $l$, as
    \begin{equation*}
        (1+\CO(h))^l \le (1+ \CO(h))^{T} \le \exp(K |\log h| \CO(h)) \le e^{\CO(\sqrt{h})} \le 1+\CO(\sqrt{h}).
    \end{equation*}
    As $\CR(-t_1)$ is a $(1+\CO(h))$ quantum partition, all the partitions involved are uniformly $(1+ \CO(\sqrt{h}))$ quantum partitions. There exists an $h_0$ such that 
    \begin{equation*}
        \CO(\sqrt{h})T = \CO(\sqrt{h}|\log(h)|) < 1/2 \qqtq{for} h\le h_0.
    \end{equation*}
    This means that the conditions of Lemma \ref{lem_almost_monotony_by_deconditioning_quantum_entropy} are satisfied with $\epsilon = \CO(\sqrt{h})$ in the hypothesis of the lemma. We suppose in the following that $h \le h_0$.
    We have according to \ref{lem_key_relation_conditional_quantum_entropy} that
    \begin{align*}
        h^q(\rho, \bigvee_{k = t_1}^{t_2} \CQ(k) | \CR) &= h^q(\rho, \bigvee_{k= n m_0}^{t_2-t_1} \CQ(k) | \bigvee_{k=0}^{n m_0-1} \CQ (k) \vee \CR(-t_1)) \\
        &\qquad + \sum_{l=0}^{n-1} h^q(\rho , \bigvee_{k= lm_0}^{(l+1)m_0 -1} \CQ(k) | \bigvee_{k=0}^{lm_0-1} \CQ(k) \vee \CR(-t_1))  \\
        & \qquad + \CO(n\sqrt{h}(1 + \log(|\bigvee_{k=0}^{t_2-t_1} \CQ(k)|))).
    \end{align*}

    We treat the first term by simply bounding it crudely via the logarithm of the cardinality of the partition $\bigvee_{k=n m_0}^{t_2-t_1} \CQ(k)$ via \eqref{ineq_bound_quantum_entropy_approximate_quant_part}. Denote $\bigvee_{k=0}^{n m_0-1} \CQ (k) \vee \CR(-t_1)$ by $\CP$, which is a $(1+ \CO(\sqrt{h}))$ quantum partition.
    \begin{align*}
        h^q(\rho, \bigvee_{k= n m_0}^{t_2-t_1} \CQ(k) | \bigvee_{k=0}^{n m_0-1} \CQ (k) \vee \CR(-t_1)) &= \sum_{p \in \CP} \rho(p^* p) h^q(\rho_{|p}, \bigvee_{k= n m_0}^{t_2-t_1} \CQ(k)) \\
        &\le \sum_{p \in \CP} \rho(p^* p) (1+ \CO(\sqrt{h})) (\log(|\bigvee_{k= n m_0}^{t_2-t_1} \CQ(k)|) +1) \\
        &\le (1+ \CO(\sqrt{h})) (m_0\log(|\CQ|) +1) (1+\CO(\sqrt{h})) \\
        &\le \CO(1),
    \end{align*}
    as $|t_2-t_1- (n m_0 -1)| \le m_0$ and $(1+\CO(\sqrt{h})) \le 2$ for $h \le h_0$. The constant depends on $|\CQ|$ and $m_0$.

    Let $l \in \eint{0}{n}$. For the second term, we have, by invariance of $\rho$ and the relation \eqref{eq_action_quantum_translation_on_functions}, that
    \begin{equation*}
        h^q(\rho , \bigvee_{k= lm_0}^{(l+1)m_0 -1} \CQ(k) | \bigvee_{k=0}^{lm_0-1} \CQ(k) \vee \CR(-t_1)) = h^q(\rho , \bigvee_{k= 0}^{m_0 -1} \CQ(k) | \bigvee_{k=-lm_0}^{-1} \CQ(k) \vee \CR(-t_1 -l m_0)),
    \end{equation*}
    and by Lemma \ref{lem_almost_monotony_by_deconditioning_quantum_entropy} that 
    \begin{align*}
        h^q(\rho , \bigvee_{k= 0}^{m_0 -1} \CQ(k) | \bigvee_{k=-lm_0}^{-1} \CQ(k) \vee \CR(-t_1 -l m_0)) \le h^q(\rho , \bigvee_{k= 0}^{m_0 -1} \CQ(k) | \CR(-t_1 -l m_0)) \\
        + \CO(\sqrt{lm_0 \sqrt{h} + \sum_{j=-lm_0}^{-1}\norm{[\CQ(j), \bigvee_{k=0}^{m_0-1} \CQ(k)]}}).
    \end{align*}
    For any $l \in \eint{0}{n-1}$, any $j \in \eint{-lm_0}{-1}$ and any $k \in \eint{0}{m_0-1}$, we have
    \begin{equation*}
        |j-k| \le (l+1)m_0 \le nm_0 \le \tilde{t} = t_2-t_1 \le 2 T.
    \end{equation*}
    Hence, we apply Lemma \ref{lem_bound_commutator_quantum_part_up_to_effective_Ehrenfest_time} to obtain 
    \begin{equation*}
        \norm{[\CQ(j), \bigvee_{k=0}^{m_0-1} \CQ(k)]} = \CO(h^\epsilon),
    \end{equation*}
    and this bound is independent of $j \in \eint{-(n-1)m_0}{-1}$. This implies, as $lm_0 \le nm_0 \le 2 T$, that 
    \begin{equation*}
        \CO(\sqrt{lm_0 \sqrt{h} + \sum_{j=-lm_0}^{-1}\norm{[\CQ(j), \bigvee_{k=0}^{m_0-1} \CQ(k)]}}) = \CO(\sqrt{T} h^{\epsilon/2}) = \CO(h^{\epsilon/3}).
    \end{equation*}
    For the third term, we have, as $n \le |t_2 -t_1| \le 2 T $ that 
    \begin{align*}
        \CO(n\sqrt{h}(1 + \log(|\bigvee_{k=0}^{t_2-t_1} \CQ(k)|))) = \CO(T \sqrt{h}(1 + T \log|\CQ|)) = \CO(h^{1/4}).
    \end{align*}
    Combining everything leads to 
    \begin{align*}
        h^q(\rho, \bigvee_{k = t_1}^{t_2} \CQ(k) | \CR) &\le \sum_{l=0}^{n-1} \left( h^q(\rho , \bigvee_{k= 0}^{m_0 -1} \CQ(k) | \CR(-t_1 -l m_0)) + \CO( h^{\epsilon/3}) \right) + \CO(1) + \CO(h^{1/4}) \\
        &\le \sum_{l=0}^{n-1} h^q(\rho , \bigvee_{k= 0}^{m_0 -1} \CQ(k) | \CR(-t_1 -l m_0)) + \CO(n h^{\epsilon/3} + h^{1/4}) \\
        &\le \sum_{l=0}^{n-1} h^q(\rho , \bigvee_{k= 0}^{m_0 -1} \CQ(k) | \CR(-t_1 -l m_0)) + \CO(1),
    \end{align*}
    as $nh^{\epsilon/3} \le K |\log h| h^{\epsilon/3} = \CO(1)$ and $\CO(h^{1/4}) = \CO(1)$ for $h \le h_0$. 
\end{proof}

The next lemma tells us that the quantum entropy increases, up to a small error, when we remove any part of the conditioning partitions, as long as the time differences between the evolution of the partitions remain smaller than twice the effective Ehrenfest time of the removed partition. We will use it to keep only specific parts of the conditioning in the proof of Theorems \ref{thm_lower_bound_full_entropy} and \ref{thm_lower_bound_entropy_projection}.

\begin{lemma}\label{lem_selective_deconditioning_quantum_entropy}
    Let $(\CQ_l)_{0 \le l \le m}$ and $(\CQ_l')_{0 \le l \le m'}$ be smooth partitions of $\T^{2d}$, $\CR$ a $(1+\CO(h))$ quantum partition, $\CP = (p_i)_{i \in I}$ a smooth partition of $\T^{2d}$ with an effective Ehrenfest time of $(T, \epsilon)$ and $\rho$ be a state on $\CH_N$. Suppose that there exists $K>0$ such that $T \le K |\log h|$, then for any $n \in \N^*$, any $t_1 \le t \le t + n \le t_2$, any $(s_l)_{0 \le l \le m}$ and $(s_l')_{0 \le l \le m'}$ with 
    \begin{equation*}
        \max(|t_2-t_1|, |t_2-s_l|, |t_1-s_l|, |t_2-s_l'|, |t_1-s_l'|) \le 2T,
    \end{equation*}
    there exists $h_0$ such that for any $h \le h_0$
    \begin{equation*}
        h^q(\rho,  \bigvee_{l=0}^m \CQ^q_l(s_l) | \bigvee_{l=0}^{m'} {\CQ'}^q_{l}(s_l') \vee \bigvee_{k=t_1}^{t_2} \CP^q(k) \vee \CR) \le h^q(\rho,  \bigvee_{l=0}^m \CQ^q_l(s_l) | \bigvee_{l=0}^{m'} {\CQ'}^q_{l}(s_l') \vee \bigvee_{k=t}^{t+n} \CP^q(k) \vee \CR) + \CO(h^{\min(1/6, \epsilon/4)}),
    \end{equation*}
    where the implicit constants depend only on $m$, $m'$, the cardinals of the partitions $(\CQ_l)_{0 \le l \le m}$, $(\CQ_l')_{0 \le l \le m'}$ and $|\CP|$, $n$ and the $S(1)$ seminorms of the elements of the classical partitions $(\CQ_l)_{l \le m}$, $(\CQ_l')_{l \le m'}$ and $\CP$.
\end{lemma}

\begin{proof}
    Firstly, we remark that all the partitions involved are $(1+ \CO(\sqrt{h}))$ as in the proof of the preceding lemma, with the constants only involving $K$ and the $S(1)$ seminorms of the elements of the classical partitions involved. Therefore, there exists an $h_0>0$ such that 
    \begin{equation*}
        \CO(\sqrt{h})T = \CO(\sqrt{h}|\log(h)|) < 1/2 \qqtq{for} h\le h_0.
    \end{equation*}
    We suppose in the following that $h \le h_0$ which will allow us to apply Lemma \ref{lem_almost_monotony_by_deconditioning_quantum_entropy}. 
    
    We remove first the partition on the left of $\bigvee_{k=t}^{t+n} \CP(k)$ and then on the right.
    \begin{align*}
        h^q(\rho,  \bigvee_{l=0}^m \CQ^q_l(s_l) | \bigvee_{l=0}^{m'} {\CQ'}^q_{l}(s_l') \vee \bigvee_{k=t_1}^{t_2} \CP(k) \vee \CR) &= h^q(\rho,  \bigvee_{l=0}^m \CQ^q_l(s_l) | \bigvee_{l=0}^{m'} {\CQ'}^q_{l}(s_l') \vee \bigvee_{k=t+n+1}^{t_2} \CP(k) \vee \bigvee_{k=t_1}^{t+n} \CP(k) \vee \CR) \\ 
        &\le h^q(\rho,  \bigvee_{l=0}^m \CQ^q_l(s_l) | \bigvee_{l=0}^{m'} {\CQ'}^q_{l}(s_l') \vee \bigvee_{k=t_1}^{t+n} \CP(k) \vee \CR) + E_1 \\
        &\le h^q(\rho,  \bigvee_{l=0}^m \CQ^q_l(s_l) | \bigvee_{l=0}^{m'} {\CQ'}^q_{l}(s_l') \vee \bigvee_{k=t}^{t+n} \CP(k) \vee \bigvee_{k=t_1}^{t-1} \CP(k) \vee \CR) + E_1 \\
        &\le  h^q(\rho,  \bigvee_{l=0}^m \CQ^q_l(s_l) | \bigvee_{l=0}^{m'} {\CQ'}^q_{l}(s_l') \vee \bigvee_{k=t}^{t+n} \CP(k) \vee \CR) + E_1 + E_2.
    \end{align*}
    The errors $E_1$ and $E_2$ come from the application of Lemma \ref{lem_almost_monotony_by_deconditioning_quantum_entropy} and are given by
    \begin{align*}
        E_1 &= \CO(\sqrt{(t_2-t -n)\sqrt{h} + \sum_{j=t+n+1}^{t_2} \norm{[\CP(j), \bigvee_{l=0}^m \CQ^q_l(s_l) \vee \bigvee_{l=0}^{m'} {\CQ'}^q_{l}(s_l')]}}) \\
        E_2 &= \CO(\sqrt{(t-t_1)\sqrt{h} + \sum_{j=t_1}^{t-1} \norm{[\CP(j), \bigvee_{l=0}^m \CQ^q_l(s_l) \vee \bigvee_{l=0}^{m'} {\CQ'}^q_{l}(s_l') \vee \bigvee_{k=t}^{t+n} \CP(k)]}}).
    \end{align*}
    We treat first $E_1$. As $|j-s_l| \le 2T$ and $|j-s_l'| \le 2T$ if $j \in \eint{t+n+1}{T}$, Lemma \ref{lem_bound_commutator_quantum_part_up_to_effective_Ehrenfest_time} implies that 
    \begin{equation*}
        \norm{[\CP(j), \bigvee_{l=0}^m \CQ^q_l(s_l) \vee \bigvee_{l=0}^{m'} {\CQ'}^q_{l}(s_l')]} = \CO(h^\epsilon).
    \end{equation*}
    Thus, since $t_2 -t-n \le 2T = 2 K|\log h| = o(h^{-1/2+1/3})$ and $t_2 - t -n = o(h^{-\epsilon/2})$, we get
    \begin{equation*}
        E_1 = \CO(\sqrt{h^{1/3} + h^{-\epsilon/2} h^{\epsilon}}) = \CO(h^{\min(1/6, \epsilon/4)}) .
    \end{equation*}
    For $E_2$, we see that if $j \in \eint{t_1}{t-1}$ and $k \in \eint{t}{t+n}$, we have $|j-k| \le 2T$ and $|j-s_l| \le 2T$ and $|j-s_l'| \le 2T$, thus Lemma \ref{lem_bound_commutator_quantum_part_up_to_effective_Ehrenfest_time} implies that 
    \begin{equation*}
        \norm{[\CP(j), \bigvee_{l=0}^m \CQ^q_l(s_l) \vee \bigvee_{l=0}^{m'} {\CQ'}^q_{l}(s_l') \vee \bigvee_{k=t}^{t+n} \CP(k)]} = \CO(h^\epsilon).
    \end{equation*}
    Therefore, similarly as $E_1$, we get
    \begin{equation*}
        E_2 = \CO(\sqrt{h^{1/3} + h^{\epsilon/2}}) = \CO(h^{\min(1/6, \epsilon/4)}) .
    \end{equation*}
\end{proof}

\section{Hyperbolic dispersive estimate}\label{sec_hyperbolic_disp_estimate}

The goal of this section is to prove an estimate of the form

\begin{equation*}
    \norm{M_h(A) \Op^w(p_{n-1}) \cdots M_h(A) \Op^w(p_0)}_{L^2} \le C e^{-\lambda n},
\end{equation*}
for $n$ of order $|\log(h)$, $A \in \Sp_{2d}(\R)$ and $p_k \in S(1)$ where each $p_k$ can be expressed as a sum $p_k = \sum_{i \in I} p_{k,i}$ with the $(p_{k,i})$ being disjoint compactly supported smooth functions.

Before giving a precise statement, we introduce some terminology to make precise our constraints on the splitting $p_k = \sum_{i \in I} p_{k,i}$. A family of functions $(f_i)_{i \in I}$ is said \textit{$\epsilon$-separated} if for any $i \neq j$, we have $d(\supp f_i, \supp f_j) > \epsilon$. A function $g$ is said \textit{$(\epsilon, \delta)$-splittable} if there exists a countable set $I$ and functions $(f_i)_{i \in I}$ such that the family $(f_i)_{i \in I}$ is $\epsilon$-separated, the diameter of the support of each $f_i$ is at most $\delta$. If $V$ is a symplectic vector subspace of $\R^{2d}$ and $g$ is invariant under translation by vectors in $V$, we require that the functions $(f_i)_{i \in I}$ satisfy the same invariance and the condition on the diameter and separation of the support of the $(f_i)_{i \in I}$ is only required for the restriction of the $(f_i)_{i \in I}$ to $V^\sigma$. 

\subsection{The main estimate and its corollary}

With the notion of splittable functions, the hyperbolic dispersive estimate takes the following form. 

\begin{theorem}\label{thm_hyp_disp_est}
    Let $\kappa \in \R^*_+$ and $h >0$. Denote $n = \floor{\kappa | \log h|}$ and let $(p_k)_{k \in \eint{0}{n-1}}$ be a family of smooth functions and $(A_k)_{k \in \eint{0}{n-1}}$ a family of symplectic matrices such that
    \begin{enumerate}
        \item there exists a constant $M$ such that $\norm{A_k} \le M$ and $A_k$ is of the form 
        \begin{equation*}
            A_k = \begin{pmatrix}
            D_k & B_k \\
            0 & D_k^{-T}
        \end{pmatrix} \qqtqq{where} D_k = \begin{pmatrix}
            D_{k,1} & \tilde{B}_k \\
            0 & D_{k,2}
        \end{pmatrix} \in \GL(d_1+ d_2, \R),
        \end{equation*} 
        \item uniformly in $m \ge 0$, 
        \begin{equation*}
            \limsup_{k \to + \infty} \frac{1}{k} \log \norm{D_{m+k}^{-T} D_{m+k-1}^{-T}  \cdots D_{m}^{-T}} \le 0,
        \end{equation*}
        \item there exists $\epsilon >0$ such that each function $p_k$ is $(\epsilon(3+M), \epsilon)$-splittable,
        \item the functions $(p_k)_{k \in \eint{0}{n-1}}$ are uniformly in $S(1)$ and satisfy $\norm{p_k}_{\infty} \le 1$.
    \end{enumerate}
    Then there exists a constant $h_\kappa >0$ such that for any $\delta > 0$ and any $h < h_\kappa$, we have 
    \begin{equation}
        \norm{\prod_{k=0}^{n-1} M_h(A_k) \Op^w_h(p_{k})}_{L^2(\R^d)} \le \frac{C_\delta}{h^{(d_1+\delta)/2}} \prod_{k=0}^{n-1} |\det(D_{k,1})|^{-1/2}(1+ \CO(|\log h| h^{1- \kappa \delta})).
    \end{equation}
    The constants depend on $\epsilon, \delta$, $M$ and the $S(1)$ seminorms of the $p_k$.
\end{theorem}

\begin{remark}\label{rem_thm_hyp_disp_est}
    The condition on the supremum norm can be relaxed. The theorem also holds if each function $p_k$ can be decomposed $p_k = p_{k,0} + h p_{k,1}$ where $\norm{p_{k,0}}_{\infty} \le 1$ and this decomposition is preserved when splitting $p_k$.
\end{remark}

To get an intuition on Theorem \ref{thm_hyp_disp_est}, consider the case where all the $(A_k)_{k \in \eint{0}{n-1}}$ are equal to the same block diagonal symplectic matrix $A = \diag(D, D^{-T})$ and all the eigenvalues of $D$ have a modulus greater than $1$. This implies that $M_h(A) u = |\det D|^{-1/2} u \circ D^{-1}$ for any $u \in L^2(\R^{d})$, thus $M_h(A)$ is a `stretching' operator. Suppose for the moment that each $p_k$ is a non-negative function that depends only on the position, is smaller than $1$, supported in $B_0(2)$ and equal to $1$ in $B_0(1)$, the ball of radius $1$ in position around $0$. The Weyl quantization $\Op^w_h(p_k)$ is just the multiplication operator by $p_k$, in other words it is a `cutting' operator. Therefore, the operator to bound in Theorem \ref{thm_hyp_disp_est} acts on a function $u \in L^2(\R^d)$ by 
\begin{equation*}
    T u := \prod_{k=0}^{n-1} M_h(A_k) \Op^w_h(p_{k}) u = \left(\prod_{k=0}^{n-1} p_k \circ D^{n-k} \right) |\det D|^{-n/2} u \circ D^{-n}.
\end{equation*}  
It can be seen as a succession of cutting and stretching operation on the function $u$: each multiplication operator $\Op^w_h(p_k)$ cuts the function $u$ to keep only a region around $0$, while each metaplectic operator $M_h(A)$ stretches the function. 

For a function $u \in L^2(\R^d)$ to satisfy $\norm{Tu} \approx \norm{u}$, it is enough and necessary that $u$ is supported in the $D^{-n} B_0(1)$, thus $u$ must be extremely concentrated in position. Suppose that we add an initial cutoff $C_{freq}$ in frequency at the beginning, the uncertainty principle tells us that $C_{freq}u$ cannot be too concentrated in position, which prevents $\norm{TC_f u} \approx \norm{u}$. Hence, the cutting by the operators $\Op^w_h(p_k)$ must also be done in frequency to obtain a decrease of the norm of the operator $T$ as in the theorem. 

Considering splittable functions $(p_k)_k$ for the cutoff operators $(\Op^w_h(p_k))_k$ does not change much of the underlying logic of the succession of cutting and stretching, except that there are multiples pieces after each cutting step. The hypotheses on the splittable functions $(p_k)_k$ are made to ensure that no two pieces overlap after the stretching and that they are cut at most by one part of the operator $\Op^w_h(p_k)$. This implies that each piece behaves independently, which reduces the analysis to the case of compactly supported functions $(p_k)_k$. This setting was already treated by Anantharaman \cite{anantharamanExponentialDecayProducts2011}.

We obtain as a corollary of Theorem \ref{thm_hyp_disp_est} a bound of the desired form. Its proof is given in Section \ref{sec_proof_cor_hyp_disp_est}.
\begin{corollary}\label{cor_applied_dispersive_estimate}
    Let $A \in \Sp_{2d}(\R)$, $V$ be a symplectic vector space invariant under $A$ of dimension $d_1$ and $\lambda_{max}$ the highest Lyapunov exponent of $A_{|V}$, supposed positive. Denote $n = \floor{ 2(1-\epsilon)| \log h|/\lambda_{max}}$ for $\epsilon>0$ and let $(p_k)_{k \in \eint{0}{n-1}}$ be a family of smooth functions invariant under $V^\sigma$ that are uniformly in $S(1)$ and satisfy $\norm{p_k}_{\infty} \le 1$. Suppose that there exists $\gamma >0$ such that each $p_k$ is $(l\gamma, \gamma)$-splittable where $l$ is a thresholds that depends only on $A$. Then there exists $\delta_0> 0$ and $h_0 >0$ such that 
    \begin{equation*}
        \forall \delta < \delta_0, \; h<h_0 \quad \norm{\prod_{k=0}^{n-1} M_h(A) \Op^w_h(p_{k})}_{L^2(\R^d)} \le C_\delta e^{- \frac{n}{2} (\Lambda_0(A_{|V}) - o(\epsilon + \delta))} (1+ \CO(\sqrt{h})).
    \end{equation*}
    The constants depend on $\epsilon, \delta$, $A$ and the $S(1)$ seminorms of the $p_k$, except for the term $o(\epsilon + \delta)$ that depends only on $\epsilon$, $\delta$ and $A$.
\end{corollary}

\begin{remark}\label{rem_cor_applied_dispersive_estimate}
    Similarly to Theorem \ref{thm_hyp_disp_est}, the condition on the supremum norm can be relaxed. The corollary also holds if each function $p_k$ can be decomposed $p_k = p_{k,0} + h p_{k,1}$ where $\norm{p_{k,0}}_{\infty} \le 1$ and this decomposition is preserved when splitting $p_k$.
\end{remark}

\subsection{Lemmas on splittable functions}

We prove a few simple lemmas that will be needed when manipulating splittable functions. The first one tells us that if $f = \sum_{i \in I} f_i$ is a splittable function, it is possible to number nicely the functions $(f_i)_{i \in I}$. 
\begin{lemma}\label{lem_renumbering_splittable_functions}
    Let $p$ be a $(\epsilon, \delta)$ splittable function on $\R^d$. Then there exists a subset $I \subset \Z^d$ and smooth functions $(p_i)_{i \in I}$ such that $p = \sum_{i \in I}p_i$, for any $i \in I$, $\norm{p_i} \le 1$, the diameter of the support of $p_i$ is at most $\delta$ and  
    \begin{equation*}
        \forall i \neq j, \quad d(\supp (p_i), \supp (p_j)) \ge \max(\epsilon, C|i-j|).
    \end{equation*}
    The constant $C$ depends only on $\epsilon$ and $\delta$.
\end{lemma}

\begin{proof}
    By hypothesis, there exist a countable set $J$ and functions $(f_j)_{j \in J}$ that satisfy all the conditions of the lemma except the one on the distance between the supports and $J \subset \Z^d$. We define an injective function $f : J \to \Z^d$ such that 
    \begin{equation*}
        \forall i \neq j \in J, \quad d(\supp (p_{f(i)}), \supp (p_{f(j)})) \ge \max(\epsilon, C|f(i)-f(j)|).
    \end{equation*}
    It is then immediate to check that the set $f(J)$ and the functions $(p_i)_{i \in f(J)}$ satisfy the condition of the lemma. 

    The function $f$ is defined by associating to $j \in J$ the element $n \in \Z^d$ such that $n\frac{\epsilon}{2\sqrt{d}} $ is the element of $\frac{\epsilon}{2\sqrt{d}} \Z$ closest to $\supp(p_j)$. If the minimum distance is attained by multiple elements, we take the smallest one for the lexicographic order. Remark that $f(j)\frac{\epsilon}{2\sqrt{d}}$ is at a distance, for the $l^2$ norm, at most $\frac{\epsilon}{2}$ from the support of $p_j$. This implies that $f$ is injective as the distance between the supports of $p_i$ and $p_j$ is strictly larger than $\epsilon$ if $i \neq j$. 
    
    \noindent Let $i,j \in J$. Since the diameter of the support of $p_i$ and $p_j$ is at most $\delta$, we have
    \begin{align}\label{ineq_renumbering_splittable_functions}
        d(\supp(p_i), \supp(p_j)) \ge | \frac{f(i)\epsilon}{2\sqrt{d}} -  \frac{f(j)\epsilon}{2\sqrt{d}}| - \epsilon - 2 \delta \ge \frac{\epsilon}{2\sqrt{d}} |f(i)  - f(j)| - \epsilon - 2 \delta .
    \end{align}
    There exist $0 < C < \frac{\epsilon}{2\sqrt{d}}$ such that
    \begin{equation*}
        \frac{\epsilon}{C} (\frac{\epsilon}{2 \sqrt{d}} - C) \ge \epsilon + 2 \delta,
    \end{equation*}
    as the left hand-side goes to $+ \infty$ as $C$ tends to $0$. It is straightforward to check that for such a choice of $C$ and for any $i,j \in J$, we have 
    \begin{equation*}
        C |f(i)-f(j)| \ge \epsilon \implies \frac{\epsilon}{2\sqrt{d}} |f(i)  - f(j)| - \epsilon - 2 \delta \ge C |f(i)-f(j)|.
    \end{equation*}
    Therefore, as the distance between the supports of $p_i$ and $p_j$ is at least $\epsilon$ if $i \neq j$, we get 
    \begin{equation*}
        d(\supp(p_i), \supp(p_j)) \ge \max(\epsilon, C|f(i)-f(j)|).
    \end{equation*}
\end{proof}

The precomposition of a splittable function by a bi-Lipschitz homeomorphism is again a splittable function.

\begin{lemma}\label{lem_precomposition_bilip_splittable_functions}
    Let $p$ be a $(\epsilon, \delta)$ splittable function on $\R^d$ and $g$ a homeomorphism of $\R^d$ such that 
    \begin{equation*}
        L\norm{x-y} \le \norm{g(x) - g(y)} \le K\norm{x-y} \qqtq{for} x, y \in \R^d.
    \end{equation*}
    Then $p\circ g$ is $(L\epsilon,K\delta)$ splittable.
\end{lemma}

\begin{proof}
    Write $p = \sum_{i \in I} p_i$ with $(p_i)_{i \in I}$ $\epsilon$-separated and the support of $p_i$ having diameter at most $\delta$. We have $p \circ g = \sum_{i \in I} p_i \circ g$. The diameter of the support of $p_i \circ g$ is thus at most $K\delta$. Indeed, any $x,y \in \supp(p_i)$ verifying
    \begin{equation*}
        \norm{x-y} \le K \norm{g^{-1}(x) - g^{-1}(y)} \le K\delta,
    \end{equation*}
    as $g^{-1}(x)$ and $g^{-1}(y)$ are in the support of $p_i$. 
    Moreover, the family $(p_i \circ g)_{i \in I}$ is $\epsilon$ separated. Let $i \neq j \in I$ and consider $x$, resp. $y$, in the support of $p_i \circ g$, resp. $p_j \circ g$. They satisfy 
    \begin{equation*}
        \norm{x-y} \ge L \norm{g^{-1}(x) - g^{-1}(y)} \ge L \epsilon,
    \end{equation*}
    as $g^{-1}(x)$, resp. $g^{-1}(y)$, is in the support of $p_i$, resp. $p_j$. 
\end{proof}

Lifting a function defined on a torus to $\R^{2d}$ gives a splittable function whose parameters depend on the diameter of the function on the torus and the radius of injectivity of the torus. 
\begin{lemma}\label{lem_splittable_function_by_lifting}
    Let $V$ be a rational symplectic vector subspace of $\R^{d}$ of dimension $2d'$. Denote $\tilde{\T} = \T^{2d} / V^\sigma$, $R$ the radius of injectivity of $\tilde{\T}$ and $\pi : \R^{2d} \to \tilde{\T}$ the natural projection. Let $\tilde{f} : \tilde{\T} \to \C$ have a support of diameter at most $\delta$. Then the lift $f = \tilde{f} \circ \pi$ is a $(R- \delta, \delta)$ splittable function if $R-\delta >0$.
\end{lemma}

\begin{proof}
    As a lift of a function defined on the torus $\tilde{\T}$ of dimension $2d'$, the function $f$ on $\R^{2d}$ can be written
    \begin{equation*}
        f = \sum_{k \in \Z^{2d'}} f_k,
    \end{equation*}
    where each function $f_k$ is just a translation of $f_0$ and is invariant under $V^\sigma$. By construction, the diameter of each $f_k$ restricted to $V$ is at most $\delta$ and they are $(R - \delta)$-separated in this direction. Indeed, consider two points $x_k$ and $y_{k'}$ that minimize the distance between $\supp(f_k) \cap V$ and $\supp(f_{k'}) \cap V$. The point $x_k$ is at a distance at least $R$ of the point $x_{k'}$ in $\supp(f_{k'})\cap V$ that has the same image under the projection in $\tilde{\T}$. Since the diameter in the direction $V$ is at most $\delta$, we find that 
    \begin{equation}\label{eq_proof_entropy_projection_R_delta_separation_classical_partition}
        d(\supp(f_k) \cap V, \supp(p_{i,k'}) \cap V) = \norm{x_k - y_{k'}} \ge |\norm{x_k - x_{k'}} - \norm{y_{k'} - x_{k'}}| \ge R - \delta.
    \end{equation}
\end{proof}

The product of splittable functions $(p_j)_{j \in J}$ that are invariant under translation by vectors of respectively $(V_j)_{j \in J}$ is again splittable if the vector subspaces $(V_j)_{j \in J}$ are in direct sum. 
\begin{lemma}\label{lem_product_splitting_function}
    Let $V_1, \dots, V_n$ be symplectic rational subspaces of $\R^{2d}$ such that they are in direct symplectic sum and $(p_i)_{i \in \eint{1}{n}}$ functions that are $(\epsilon_i, \delta_i)$-splittable and invariant under $V_i^\sigma$. Denote $\hat{V} := V_1 \oplus^\sigma \dots \oplus^\sigma V_n$ and consider $q = p_1 \dots p_n$. Then $q$ is invariant under $\hat{V}^\sigma$ and $(C^{-1} \sum_{i=1}^n \epsilon_i, C \sum_{i=1}^n \delta_i)$-splittable where $C$ is a constant that depends only on $\hat{V}$. 
\end{lemma}

\begin{proof}
    The function $q = p_1 \dots p_j$ can be written 
    \begin{equation*}
        q = \sum_{(k_1, \dots, k_n) \in \Z^{2d}} p_{1,k_1} \dots p_{n, k_n},
    \end{equation*}
    where each function is invariant under $V^\sigma_i$. By construction, each $p_{i,k_i}$ has diameter at most $\delta$ when restricted to $V_i$. We define the norm
    \begin{equation*}
        \norm{x}_{\hat{V}_j} = \norm{x_1} + \cdots + \norm{x_j} \qqtqq{where} x = x_1 + \cdots + x_j \qtq{with} x_i \in V_i.
    \end{equation*}
    This norm is equivalent to the restriction to $\hat{V}$ of the usual norm on $\R^{2d}$. Each product $p_{1,k_1} \dots p_{n, k_n}$ has diameter at most $\sum_{i=1}^n \delta_i$ for the norm $\norm{\cdot}_{\hat{V}_j}$, and by equivalence of norms diameter at most $C\sum_{i=1}^n \delta_i$ for the usual norm. A reasoning similar to the one leading to \eqref{eq_proof_entropy_projection_R_delta_separation_classical_partition}, gives that 
    \begin{equation*}
        d_{\hat{V}}(\supp(p_{1,k_1} \dots p_{n, k_n}) \cap \hat{V}, \supp(p_{1,k_1'} \dots p_{n, k_n'}) \cap \hat{V}) \ge \sum_{i=1}^n \epsilon_i,
    \end{equation*}
    where the distance is induced by $\norm{\cdot}_{\hat{V}}$. We obtain the desired result by equivalence of norms. 
\end{proof}

\subsection{Proof of the corollary \ref{cor_applied_dispersive_estimate}}\label{sec_proof_cor_hyp_disp_est}

Corollary \ref{cor_applied_dispersive_estimate} follows straightforwardly from Theorem \ref{thm_hyp_disp_est}. The proof consists in conjugating the matrix $A$ thanks to Theorem \ref{thm_reduction_symplectic_matrices_real} to obtain the form of symplectic matrices required by Theorem \ref{thm_hyp_disp_est} and checking the hypotheses of Theorem \ref{thm_hyp_disp_est}. 
\begin{proof}
    We start by dealing with the $V^\sigma$ invariance of the $(p_k)$ and putting $A_{|V}$ in a nice form. As $V$ is $A$ invariant, $V^\sigma$ is also $A$ invariant. Thus, we put $A$ in symplectic block diagonal form via a symplectic change of basis $Q \in \Sp_{2d}(\R)$ sending the decomposition $\R^{2d_1} \oplus^\sigma \R^{2(d-d_1)}$ of $\R^{2d}$ to $V \oplus^\sigma V^\sigma$. Up to a modification of $Q$, we also symplectically reduce the restriction of $A$ to $V$ via Theorem \ref{thm_reduction_symplectic_matrices_real}. Denote by $\lambda_{max} = \lambda_1 > \cdots > \lambda_r > 0 $ the positive Lyapunov exponents of $A_{|V}$. This leads to the form 
    \begin{equation*}
        A = Q (\tilde{A} \diamond B) Q^{-1}  \qqtqq{with} \tilde{A} = \begin{pmatrix}
            D_1 & 0 \\ 0 & D^{-T}_1
        \end{pmatrix} \diamond \dots \diamond \begin{pmatrix}
            D_r & 0 \\ 0 & D^{-T}_r
        \end{pmatrix} \diamond C_{r+1},
    \end{equation*}
    where $Q$ is a symplectic matrix on $\R^{2d}$, $C_{r+1} \in \Sp(2d_0, \R)$ has its spectrum included in the unit circle and that of each $D_i$ is included in the circle of radius $e^{\lambda_i}$. 
    Recall that $M_h(\tilde{A} \diamond B) = M_h(\tilde{A}) \otimes M_h(B)$ acting on $L^2(\R^{d_1}) \otimes L^2(\R^{d-d_1})$ according to \eqref{eq_compatibility_metaplectic_op_concatenation_symplectic}. Therefore, using the exact Egorov relation \eqref{eq_exact_Egorov_R2d} leads to 
    \begin{align*}
        \prod_{k=0}^{n-1} M_h(A) \Op^w_h(p_{k}) &= \prod_{k=0}^{n-1} M_h(Q) M_h(\tilde{A}) \otimes M_h(B) M_h(Q^{-1}) \Op^w_h(p_{k}) \\
        &= M_h(Q) M_h(\tilde{A}) \otimes M_h(B) M_h(Q)^{-1} \Op^w_h(p_{n-1}) M_h(Q) \dots  \\
        & \qquad \dots M_h(\tilde{A}) \otimes M_h(B) M_h(Q)^{-1} \Op^w_h(p_{0}) M_h(Q) M_h(Q)^{-1} \\
        & = M_h(Q) \left(\prod_{k=0}^{n-1} M_h(\tilde{A}) \otimes M_h(B) \Op^w_h(p_{k} \circ Q) \right) M_h(Q)^{-1} \\
        & = M_h(Q) \left(\prod_{k=0}^{n-1} M_h(\tilde{A}) \Op^w_h(p_{k} \circ Q) \otimes M_h(B) Id_{L^2(\R^{d-d_1})} \right)  M_h(Q)^{-1} \\
        &= M_h(Q) \left(\prod_{k=0}^{n-1} M_h(\tilde{A}) \Op^w_h(p_{k} \circ Q) \right) \otimes M_h(B)^n  M_h(Q)^{-1}.
    \end{align*}
    We used that $\Op^w_h(p_k \circ Q) = \Op^w_h(p_k \circ Q) \otimes Id_{L^2(\R^{d-d_1})}$ where the first operator acts on $L^2(\R^{d_1})$. Indeed, as $Q$ sends $\R^{2(d-d_1)}$ to $V^\sigma$ and $p$ is invariant under $V^\sigma$, the observable $p \circ Q$ is invariant under $\R^{2(d-d_1)}$ hence its quantization $\Op^w_h(p \circ Q)$ acts only on $L^2(\R^{d_1})$. 
    Since $\norm{T \otimes S} = \norm{T} \norm{S}$ if $T$ and $S$ are bounded operators acting on two Hilbert spaces and metaplectic operators are unitary, we get 
    \begin{align*}
        \norm{\prod_{k=0}^{n-1} M_h(A) \Op^w_h(p_{k})}_{L^2(\R^{d})}
        = \norm{\prod_{k=0}^{n-1} M_h(\tilde{A}) \Op^w_h(p_{k} \circ Q) }_{L^2(_R^{d_1})}.
    \end{align*}
    
    We need to further reduce $C_{r+1}$. Denote by $L_0$ the Lagrangian subspace of $\R^{2d_0}$ generated by $(f_1, \dots, f_{d_0})$ and define $L_k := C_{r+1}^k L_0$. As $C_{r+1}$ is symplectic, the subspaces $L_k$ are again Lagrangian. Since the action of orthogonal symplectic matrices on the set of Lagrangian subspaces \cite[lem. 2.2.1 and 2.3.2]{mcduffIntroductionSymplecticTopology2017} is transitive, there exists an orthogonal and symplectic matrix $U_k$ sending $L_k$ to $L_0$. Hence, the symplectic matrices $C_{r+1,k} := U_k C_{r+1} U_{k-1}^{-1}$ stabilize $L_0$ and are thus of the form 
    \begin{equation*}
        C_{r+1,k} = \begin{pmatrix}
            D_{r+1,k} & * \\ 0 & D_{r+1,k}^{-T}
        \end{pmatrix}.
    \end{equation*}
    We naturally extend $U_k$ to an orthogonal and symplectic matrix on $\R^{2d}$. Define $j$ to be the largest index such that $\lambda_j > \lambda_{max}/2$. We get that $A_k := U_k \tilde{A} U_{k-1}^{-1}$ has the form
    \begin{equation*}
        A_k = \begin{pmatrix}
            \tilde{D}_{k,1} & 0 & * & * \\
            0 & \tilde{D}_{k,2} & * & * \\
            0 & 0 & \tilde{D}_{k,1}^{-T} & 0 \\
            0 & 0 & 0 & \tilde{D}_{k,2}^{-T}
        \end{pmatrix} \qqtqq{where} \begin{aligned}
        \tilde{D}_{k,1} &= \diag(D_1, \dots, D_j) \\
        \tilde{D}_{k,2} &= \diag(D_{j+1}, \dots, D_r, D_{r+1,k})
        \end{aligned}.
    \end{equation*}
    We introduce the metaplectic operators associated to the $U_k$ to replace $\tilde{A}$ by $A_k$. 
    \begin{align*}
        \prod_{k=0}^{n-1} M_h(\tilde{A}) \Op^w_h(p_{k} \circ Q) &= M_h(U_n)^{-1} M_h(A_n) M_h(U_{n-1}) \Op^w_h(p_{n-1} \circ Q) M_h(U_{n-1})^{-1} \\
        &\dots M_h(U_1)^{-1} M_h(A_1) M_h(U_0)^{-1} M_h(U_0) \Op^w_h(p_{0} \circ Q) M_h(U_0)^{-1} M_h(U_0) \\
        &= M_h(U_n)^{-1} \prod_{k=0}^{n-1} M_h(A_{k+1}) \Op^w_h(p_{k} \circ Q \circ U_k^{-1}) M_h(U_0).
    \end{align*}
    Hence, if we denote $\tilde{p}_k = p_{k} \circ Q \circ U_k^{-1}$, we have
    \begin{equation*}
        \norm{\prod_{k=0}^{n-1} M_h(\tilde{A}) \Op^w_h(p_{k} \circ Q)} = \norm{\prod_{k=0}^{n-1} M_h(A_{k+1}) \Op^w_h(\tilde{p}_{k})}.
    \end{equation*}
    It remains to check the conditions of Theorem \ref{thm_hyp_disp_est}. First, the $A_k$ have the good forms and since the $U_k$ are unitary, we find that $\norm{A_k} \le \norm{\tilde{A}}$. Second, in view of the shape of $\tilde{A}$, it suffices to prove that uniformly in $m \ge 0$ and for any $j \in \eint{1}{r}$
    \begin{equation*}
        \limsup_{k \to + \infty} \frac{1}{k} \log \norm{(D_{j}^T)^{-k}} \le 0 \qqtqq{and} \limsup_{k \to + \infty} \frac{1}{k} \log \norm{D_{r+1,m+k}^{-T} \cdots D_{r+1,m}^{-T}} \le 0,
    \end{equation*}
    for the second condition of the theorem to be satisfied. Since the spectrum of $D_j$ is included in the circle of radius $e^{\lambda_j}$, we get 
    \begin{equation*}
        \limsup_{k \to + \infty} \frac{1}{k} \log \norm{(D_{j}^T)^{-k-1}} \le - \tilde{d}_j \lambda_j,
    \end{equation*}
    where $\tilde{d}_j$ is the size of $D_j$. Remark that $D_{r+1,m+k}^{-T} \cdots D_{r+1,m}^{-T}$ is a submatrix of 
    \begin{equation*}
        C_{r+1,m+k} \cdots C_{r+1,m} = U_{m+k} C_{r+1}^{k+1} U_{m-1}.
    \end{equation*}
    Moreover, as the spectrum of $C_{r+1}$ is included in the unit circle, the sequence $(\norm{C_{r+1}^{k+1}})_{k \in \N}$ can grow at most polynomially in $k$. Thus, since the $U_k$ are isometries, we get
    \begin{align*}
        \limsup_{k \to + \infty} \frac{1}{k} \log \norm{D_{r+1,m+k}^{-T} \cdots D_{r+1,m}^{-T}} &\le \limsup_{k \to + \infty} \frac{1}{k} \log \norm{U_{m+k} C_{r+1}^{k+1} U_{m-1}} \\ 
        &\le \limsup_{k \to + \infty} \frac{1}{k} \log \norm{C_{r+1}^{k+1}} \le 0.
    \end{align*}

    Third, Lemma \ref{lem_precomposition_bilip_splittable_functions} implies that $\tilde{p}_k = p_k \circ Q \circ U_{k}^{-1}$ is $(l \gamma/\norm{Q}, \norm{Q^{-1}}\gamma)$-splittable as $U_k$ is unitary and $p_k$ is $(l \gamma, \gamma)$-splittable, therefore the third hypothesis of the theorem is satisfied if $l \ge \norm{Q^{-1}} \norm{Q} (3 + \norm{\tilde{A}})$. 

    Finally, as the $U_k$ are unitary and the $S(1)$ seminorms of the $p_k$ are uniformly bounded, the seminorms of $\tilde{p}_k$ are also uniformly bounded. It is also clear that their supremum norm satisfies $\norm{\tilde{p}_k}_\infty \le 1 $. Denote $\tilde{d_i}$ the dimension of $D_i$ and $d' = \sum_{i=1}^j \tilde{d}_i$. Therefore, we can apply the theorem to obtain for $h$ and $\delta$ smaller than some $h_0$ and any $\delta_0$ that
    \begin{align*}
        \norm{\prod_{k=0}^{n-1} M_h(A_{k+1}) \Op^w_h(\tilde{p}_{k})} &\le \frac{C_\delta}{h^{(d' + \delta)/2}} \prod_{k=0}^{n-1} |\det(\tilde{D}_{k,1})|^{-1/2} (1+ \CO(|\log h| h^{1- 2 (1-\epsilon)\delta/\lambda_{max}})) \\
        &\le \frac{C_\delta}{h^{(d' + \delta)/2}} \prod_{i=1}^j |\det(D_i)|^{-n/2} (1+ \CO(\sqrt{h})) \\
        &\le \frac{C_\delta}{h^{(d' + \delta)/2}} \prod_{i=1}^j e^{-n \tilde{d_i} \lambda_i /2} (1+ \CO(\sqrt{h})) ,
    \end{align*}
    by choosing $\delta_0$ small enough to ensure that $\CO(|\log h| h^{1- 2 \delta/\lambda_{max}}) = \CO(\sqrt{h})$. Our definition of $j$ and the choice $n = \floor{2(1-\epsilon)|\log h|/\lambda_{max}}$ implies 
    \begin{align*}
        \frac{C_\delta}{h^{(d' + \delta)/2}} \prod_{i=1}^j |\det(D_i)|^{-n/2} &= C_\delta e^{- \frac{n}{2} \sum_{i=1}^j \tilde{d_i} (\lambda_i - \frac{|\log h|(1 + \delta/\tilde{d})}{n}) } \le C_\delta e^{- \frac{n}{2} (\sum_{i=1}^j \tilde{d_i} (\lambda_i - \frac{\lambda_{max}}{2}) - o(\epsilon + \delta) )} \\
        &\le C_\delta e^{- \frac{n}{2} (\Lambda_0(\tilde{A}) - o(\epsilon + \delta))}.
    \end{align*}

    To conclude, we remark that $\Lambda_0(\tilde{A}) = \Lambda_0(A_{|V})$ as the two matrices are conjugate to one another.
\end{proof}

The rest of the section is dedicated to the proof of Theorem \ref{thm_hyp_disp_est}, which consists of two main steps: replacing each symbol $p_k$ by one of its summand $p_{k,i}$ and then applying the main result of \cite{anantharamanExponentialDecayProducts2011}.

\subsection{Reduction to a single trajectory}

As promised, we start by replacing each symbol $p_k$ by only one of its constituents $p_{k,i}$. The main technical tool used is the Cotlar-Stein lemma \cite[thm. C.5]{zworskiSemiclassicalAnalysis2012}.

\begin{lemma}[Cotlar-Stein]
    Let  $(A_i)_{i \in \Z^m}$ be a countable family of bounded operators between two Hilbert spaces $E$ and $F$ such that there exists $C >0$ that satisfies 
\begin{equation*}
    \sup_{i \in \Z^m} \sum_{j \in \Z^m} \norm{A_i^* A_j}^{1/2} \le C \qqtqq{and} \sup_{i \in \Z^m} \sum_{j \in \Z^m} \norm{A_i A_j^*}^{1/2} \le C
\end{equation*}
then the series $\sum_{i \in \Z^m} A_i$ converges in the strong operator topology and 
\begin{equation*}
    \norm{\sum_{i\in \Z^m} A_i} \le C.
\end{equation*}
\end{lemma}

We encapsulate our use of the Cotlar-Stein lemma in the next lemma that deals with operators of the form $A_{(i,j)}(h) = P_i(h) Q_j(h)$ depending on a small parameter $h>0$.
\begin{lemma}\label{lem_Cotlar-Stein_on_steroids}
    Let $(P_i(h))_{i \in \Z^m}$ and $(Q_j(h))_{j \in \Z^m}$ be two families of $h$-dependent operators on $L^2(\R^d)$. Suppose that there exist uniform constants $(C_n)_{n \ge 1}$, $C >0$, a subset $J \subset \Z^m$, a function $f : \Z^m \to \Z^m$ injective on $J$ such that 
    \begin{itemize}
        \item almost orthogonality: 
        \begin{equation}\label{lem_Cotlar-Stein_on_steroids_hyp_almost_ortho}
            \forall i \neq j \in \Z^m, \; \forall n \in \N^*, \quad \norm{P_i^* P_j} \le C_n \frac{h^n}{\jap{i-j}^n} \qtq{and} \norm{Q_i Q_j^*} \le C_n \frac{h^n}{\jap{i-j}^n},
        \end{equation}
        \item one intersection for elements of $J$: 
        \begin{equation}\label{lem_Cotlar-Stein_on_steroids_hyp_one_intersect}
            \forall j \in J, \; \forall i \neq f(j), \; \forall n \in \N^* \quad \norm{P_i Q_j} \le C_n \frac{h^n}{\jap{i-f(j)}^n},
        \end{equation}
        \item zero intersection for elements of $J^c$: 
        \begin{equation}\label{lem_Cotlar-Stein_on_steroids_hyp_zero_intersect}
            \forall j \in J^c, \; \forall i \in \Z^m, \; \forall n \in \N^* \quad \norm{P_i Q_j} \le C_n \frac{h^n}{\jap{i-f(j)}^n},
        \end{equation}
        \item uniform boundedness: $\forall i,j \in \Z^m \quad \norm{P_i} \le C \qtq{and} \norm{Q_j} \le C $,
    \end{itemize}
    then the operator $\sum_{i,j \in \Z^m} P_i Q_j$ converges in the strong operator topology and satisfies 
    \begin{enumerate}
        \item $\sum_{i,j \in \Z^m} P_i Q_j = \sum_{j \in J} P_{f(j)} Q_j + R$ with $\norm{R} = \CO(h^\infty)$,
        \item $\norm{\sum_{j \in J} P_{f(j)} Q_j} \le \sup_{j \in J} \norm{P_{f(j)} Q_j} + \CO(h^\infty)$ .
    \end{enumerate}
\end{lemma}

\begin{proof}
    As claimed above, the proof is a direct application of the Cotlar-Stein lemma. Define
    \begin{equation*}
        R_1 = \sum_{\substack{(i,j) \in \Z^m \times J \\ i \neq f(j)}} P_i Q_j \qqtqq{and} R_2 = \sum_{(i,j) \in \Z^m \times J^c} P_i Q_j.
    \end{equation*}
    and note that
    \begin{equation*}
        \sum_{i,j \in \Z^m} P_i Q_j = \sum_{j \in J} P_{f(j)} Q_j + R_1 + R_2
    \end{equation*}
    In order to apply the Cotlar-Stein lemma to $R_1$, we need to bound 
    \begin{equation*}
        A_1 := \sup_{\substack{i \in \Z^m \\ j \in J \\ i \neq f(j)}} \sum_{\substack{k \in \Z^m \\l \in  J \\ k \neq f(l)}} \norm{P_i Q_j Q_l^* P_k^*}^{1/2} \qqtqq{and} B_1 := \sup_{\substack{i \in \Z^m \\ j \in J \\ i \neq f(j)}} \sum_{\substack{k \in \Z^m \\l \in  J \\ k \neq f(l)}} \norm{Q_j^* P_i^* P_k Q_l}^{1/2}.
    \end{equation*}
    Let $n \in \N^*$, $i,k \in Z^m$ and $j,l \in J$ such that $i \neq f(j)$ and $k \neq f(l)$. The hypothesis \eqref{lem_Cotlar-Stein_on_steroids_hyp_one_intersect} implies on the one hand 
    \begin{equation}\label{ineq_proof_abstract_Cotlar-Stein_first_ineq}
        \norm{P_i Q_j Q_l^* P_k^*} \le \norm{P_i Q_j} \norm{P_k Q_l} \le C_n^2 \frac{h^{2n}}{\jap{i-f(j)}^n\jap{k-f(l)}^n} \le C_n^2 \frac{h^{2n}}{\jap{k-f(l)}^n},
    \end{equation}
    and on the other hand, if $j \neq l$ the hypothesis \eqref{lem_Cotlar-Stein_on_steroids_hyp_almost_ortho} leads to 
    \begin{equation*}
        \norm{P_i Q_j Q_l^* P_k^*} \le \norm{P_i} \norm{P_k} \norm{Q_j Q_l^*} \le C^2 C_{2n} \frac{h^{2n}}{\jap{j-l}^{2n}}.
    \end{equation*}
    Combining the two bounds, we get for $j \neq l$ that 
    \begin{equation*}
        \norm{P_i Q_j Q_l^* P_k^*} \le C \sqrt{C_{2n}} \frac{h^{n}}{\jap{j-l}^{n}} C_n \frac{h^{n}}{\jap{k-f(l)}^{n/2}} \le C_{n}' \frac{h^{2n}}{\jap{j-l}^{n}\jap{k-f(l)}^{n/2}}.
    \end{equation*}
    Such a bound is also valid for $j=l$, thanks to the inequality \eqref{ineq_proof_abstract_Cotlar-Stein_first_ineq}. Thus, $A_1$ is bounded by
    \begin{align*}
        A_1 &\le \sup_{\substack{(i,j) \in \Z^m \times J \\ i \neq f(j)}} \sum_{\substack{(k,l) \in \Z^m \times J \\ k \neq f(l)}} C_{n}' \frac{h^{n}}{\jap{j-l}^{n/2}\jap{k-f(l)}^{n/4}} \\ 
        &\le C_n' h^n \sup_{j \in \Z^m} \sum_{k,l \in \Z^m} \frac{1}{\jap{j-l}^{n/2}\jap{k-f(l)}^{n/4}} \\
        &\le C_n' h^n \sup_{j \in \Z^m} \sum_{l \in \Z^m} \frac{1}{\jap{j-l}^{n/2}} \sum_{k \in \Z^m} \frac{1}{\jap{k-f(l)}^{n/4}} \\
        &\le C_n' h^n \sup_{j \in \Z^m} \sum_{l \in \Z^m} \frac{1}{\jap{l}^{n/2}} \sum_{k \in \Z^m} \frac{1}{\jap{k}^{n/4}} \\
        &\le C_n' h^n \qtq{as long as} n \ge 4m +1.
    \end{align*}
    Similarly, we obtain $B_1 \le C_n' h^n$ as long as $n \ge 4m+1$. The Cotlar-Stein lemma thus implies that $R_1$ converges in the strong operator topology and satisfies $\norm{R_1} = \CO(h^\infty)$. 
    The second error term $R_2$ is bounded similarly via the Cotlar-Stein lemma and a control on 
    \begin{equation*}
        A_2 := \sup_{\substack{i \in \Z^m \\ j \in J^c }} \sum_{\substack{k \in \Z^m \\l \in  J^c }} \norm{P_i Q_j Q_l^* P_k^*}^{1/2} \qqtqq{and} B_2 := \sup_{\substack{i \in \Z^m \\ j \in J^c }} \sum_{\substack{k \in \Z^m \\l \in  J^c }} \norm{Q_j^* P_i^* P_k Q_l}^{1/2}.
    \end{equation*}
    The $\CO(h^\infty)$ bounds on $A_2$ and $B_2$ are obtained similarly to the bounds for $A_1$ and $B_1$ except that \eqref{ineq_proof_abstract_Cotlar-Stein_first_ineq} follows from the hypothesis \eqref{lem_Cotlar-Stein_on_steroids_hyp_zero_intersect} instead of \eqref{lem_Cotlar-Stein_on_steroids_hyp_one_intersect}.

    We treat the remaining part $\sum_{j \in J} P_{f(j)}Q_j$ again with the Cotlar-Stein lemma. We need to bound 
    \begin{equation*}
        A_3 := \sup_{j \in J} \sum_{l \in J} \norm{P_{f(j)} Q_j Q_l^* P_{f(l)}^*}^{1/2} \qqtqq{and} B_3 := \sup_{j \in J} \sum_{l \in J} \norm{Q^*_j P^*_{f(j)} P_{f(l)} Q_l}^{1/2}.
    \end{equation*}
    We have for $j \in J$ and $n \in \N^*$, using the injectivity of $f$, 
    \begin{align*}
        \sum_{l \in J} \norm{Q_j^* P_{f(j)}^* P_{f(l)} Q_l}^{1/2} &= \norm{P_{f(j)} Q_j} + \sum_{\substack{l \in J \\ l \neq j}} \norm{Q_j^* P_{f(j)}^* P_{f(l)} Q_l}^{1/2} \\
        &\le \norm{P_{f(j)} Q_j} + \sum_{\substack{l \in J \\ l \neq j}} \sqrt{C_{2n}} C \frac{h^n}{\jap{f(j) - f(l)}^n} \\
        &\le \norm{P_{f(j)} Q_j} + C_n' h^n \sum_{l \in J} \frac{1}{\jap{f(j) - f(l)}^n} \\
        &\le \norm{P_{f(j)} Q_j} + C_n' h^n \sum_{l \in \Z^m} \frac{1}{\jap{l}^n} \\
        &\le \norm{P_{f(j)} Q_j} + C_n' h^n \qtq{as long as} n \ge m+1.
    \end{align*}
    Thus $B_3 \le \sup_{j \in J} \norm{P_{f(j)} Q_j} + \CO(h^\infty)$ and we obtain similarly $A_3 \le \sup_{j \in J} \norm{P_{f(j)} Q_j} + \CO(h^\infty)$. Therefore, the Cotlar-Stein lemma implies that $\sum_{j \in J} P_{f(j)}Q_j$ converges in the strong operator topology and 
    \begin{equation*}
        \norm{\sum_{j \in J} P_{f(j)}Q_j} \le \sup_{j \in J} \norm{P_{f(j)} Q_j} + \CO(h^\infty).
    \end{equation*}
\end{proof}

Equipped with this lemma, we prove the next one, which replaces each $p_k$ by one of its summands $p_{k,i}$. This allows us to replace the symbols $p_k$ whose support is unbounded by the symbols $p_{k,i}$ whose supports have a uniformly bounded diameter.

\begin{lemma}\label{lem_Cotlar-Stein_trajectory}
    Let $\epsilon >0$, $0< h \le 1$ and $\kappa \in \N^*$. Define $n = \floor{\kappa | \log h|}$ and take symplectic matrices $(A_k)_{0 \le k \le n} \in \Sp_{2d}(\R)$ and smooth functions $(p_k)_{0 \le k \le n}$. Suppose that 
    \begin{enumerate}
        \item there exists $M>0$ such that all symplectic matrices $A_k$ satisfy $\norm{A_k} \le M$,
        \item the functions $(p_k)_{k \in \eint{0}{n-1}}$ are uniformly in $S(1)$ and their supremum norm satisfies $\norm{p_k}_\infty \le 1$,
        \item there exists $\epsilon >0$ such that each function $p_k$ is $(\epsilon(3+M), \epsilon)$-splittable.
    \end{enumerate}
    Then there exist a countable set $I$ and smooth functions $(p_{k,i})_{k \in \eint{0}{n-1}, i \in I}$ such that 
    \begin{equation*}
        \prod_{k = 0}^{n-1} M_h(A_k) \Op^w_h(p_k) = \sum_{i \in I} \prod_{k=0}^{n-1} M_h(A_k) \Op^w_h(p_{k,i}) + R,
    \end{equation*}
    where $\norm{R}_{L^2 \to L^2} = \CO(h^\infty)$ and 
    \begin{equation*}
        \norm{\sum_{i \in I} \prod_{k=0}^{n-1} M_h(A_k) \Op^w_h(p_{k,i})} = \sup_{i \in I} \norm{\prod_{k=0}^{n-1} M_h(A_k) \Op^w_h(p_{k,i})} + \CO(h^\infty).
    \end{equation*}
    Moreover, the diameter of the support of each $p_{k,i}$ is at most $\epsilon$, the functions $(p_{k,i})_{k \in \eint{0}{n-1}, i \in I}$ are uniformly in $S(1)$, their supremum norm is bounded by one $\norm{p_{k,i}} \le 1$ and for each $i \in I$
    \begin{equation*}
        d(\supp p_{k+1,i}, p_{k,i} \circ A_k^{-1}) \le \epsilon.
    \end{equation*}
\end{lemma}

\begin{proof}
    By hypothesis and Lemma \ref{lem_renumbering_splittable_functions}, each function $p_k$ can be decomposed as a sum $p_k = \sum_{i \in I_k} p_{k,i}$ such that 
    \begin{itemize}
        \item the functions $(p_{k,i})_{0 \le k \le n, i \in I_k}$ are uniformly in $S(1)$:
        \begin{equation}\label{eq_lem_Cotlar-Stein_trajectory_hyp_unif_S(1)}
            \forall k \in \eint{0}{n-1}, \; \forall i \in I_k, \; \forall \alpha \in \N^d, \; \exists C_\alpha \in \R^+ \qquad \norm{\partial^\alpha p_{k,i}}_\infty \le C_\alpha ,
        \end{equation}
        \item their supremum norm is bounded by $1$:
        \begin{equation}\label{eq_lem_Cotlar-Stein_trajectory_hyp_sup_norm}
            \forall k \in \eint{0}{n-1}, \; \forall i \in I_k, \quad \norm{p_{k,i}}_{\infty} \le 1,
        \end{equation}
        \item the supports of the $p_{k,i}$ have a diameter less than $\epsilon$ and are well separated:
        \begin{equation}\label{eq_lem_Cotlar-Stein_trajectory_hyp_supports}
            \exists C >0, \; \forall k \in \eint{0}{n-1} , \; \forall i \neq j \quad d(\supp p_{k,i} , \supp p_{k,j}) > \max\left(\epsilon (3+M), C|i-j|\right),
        \end{equation}
    \end{itemize}
    Denote $P_{k,i} = M_h(A_k) \Op^w_h(p_{k,i})$. Since $M_h(A_k)$ is unitary, we have that for any $N \in \N^*$ and any $j \neq i \in I_k$
    \begin{nalign}\label{eq_lem_Cotlar-Stein_trajectory_ineq_Pki*_Pkj}
        \norm{P_{k, i}^*P_{k, j}} &= \norm{\Op^w_h (p_{k, i})^* M_h (A_k)^* M_h (A_k) \Op^w_h (p_{k, j})} = \norm{\Op^w_h (p_{k, i})^*\Op^w_h (p_{k, j})} \\
        &\le{} C_N \frac{h^N}{\jap{i-j}^n},
    \end{nalign}
    as the distance between the supports of $p_{k,i} \in S(1)$ and $p_{k,j} \in S(1)$ is at least $C/|i-j|$ and by applying Lemma \ref{lem_disjoint_support_weyl_quantization_R2d}. The constants $C_n$ are uniform in $k, i,j$  as they depend only on the dimension, $C$ and the supremum norm of the derivatives of $p_{k,i}$ and $p_{k,j}$, which are uniformly bounded according to \eqref{eq_lem_Cotlar-Stein_trajectory_hyp_unif_S(1)}. Similarly,
    \begin{nalign}\label{ineq_lem_Cotlar-Stein_trajectory_Pki_Pkj*}
        \norm{P_{k, i} P_{k, j}^*} &= \norm{M_h (A_k) \Op^w_h (p_{k, i}) \Op^w_h (p_{k, j})^* M_h (A_k)^*} = \norm{\Op^w_h (p_{k, i}) \Op^w_h (p_{k, j})^*} \\
        &\le{} C_n \frac{h^N}{\jap{i-j}^N}.
    \end{nalign}
    Denote the support of $p_{k,i}$ by $U_{k,i}$ and that of $p_{k,i} \circ A^{-1}_k$ by $W_{k,i}$ and define a set of indices $J_k \subset I_k$ by
    \begin{equation*}
        J_k = \set{j \in I_k \; | \; \exists i \in I_{k+1}, \; d(U_{k+1,j},W_{k,i}) \le \epsilon}.
    \end{equation*}
    We claim that $W_{k,j}$ is at distance less than $\epsilon$ of at most one element of $(U_{k+1,i})_{i \in I_{k+1}}$. Suppose by contradiction that there exist $i \neq i'$ such that the distance between $W_{k,j}$ and both $U_{k+1,i}$ and $U_{k+1,i'}$ is less than $\epsilon$. By compactness of the supports, there exist four points $x_i \in U_{k+1,i}$, $y_0, y_1 \in W_{k,j}$ and $x_{i'} \in U_{k+1,i'}$ such that 
    \begin{equation*}
        |x_i - y_0| \le \epsilon \quad |x_i' - y_1| \le \epsilon.
    \end{equation*}
    Remark that $W_{k,j} = A_k U_{k,j}$ and since the diameter of $U_{k,j}$ is at most $\epsilon$, that of $W_{k,j}$ is at most $\norm{A_k} \epsilon \le M\epsilon$. This implies that $|y_0 - y_1| \le M\epsilon$. Therefore,
    \begin{equation*}
        |x_i - x_{i'}| \le |x_i - y_0| + |y_0 - y_1| + |y_1 - x_{i'}| \le 2 \epsilon + \epsilon M = \epsilon (2+M),
    \end{equation*}
    which is in contradiction with the hypothesis \eqref{eq_lem_Cotlar-Stein_trajectory_hyp_supports}. Therefore, the function $\tilde{f}_k : J_k \to I_{k+1}$ that sends $j \in J_k$ to the only index $i \in I_{k+1}$ such that $W_{k,j} \cap U_{k+1,i}$ is non-empty is well-defined. The function $\tilde{f}_k$ is injective as for any $i \in I_{k+1}$ the support $U_{k+1,i}$ intersects at most one element of $(W_{k,j})_{j \in I_{k}}$. The definition of $\tilde{f}_k$ and the hypothesis \eqref{eq_lem_Cotlar-Stein_trajectory_hyp_supports} imply that 
    \begin{equation}\label{ineq_lem_Cotlar-Stein_trajectory_support_close}
        \forall j \in J_k, \; \forall i \neq \tilde{f}_k(j) \quad \dist(\supp p_{k+1,i}, \supp p_{k,j} \circ A^{-1}_k) \ge C' |i-\tilde{f}_k(j)|.
    \end{equation}
    We extend the function $\tilde{f}_k$ to $I_k$ by sending an index $j \in I_k \setminus J_k$ to the index $i \in I_{k+1}$ such that the distance between $W_{k,j}$ and $U_{k+1,i}$ is the smallest and $i$ is as small as possible for the lexicographic order on $\Z^d$. This last criterion is there to handle the case of the smallest distance being attained by multiple indices. Our definition of $\tilde{f}_k$ on $J^c_k$, the fact that for any element $j \in J^c_k$ the support $\supp p_{k,j} \circ A^{-1}_k$ is at distance at least $\epsilon$ of the support of any $p_{k+1,i}$ and the hypothesis \eqref{eq_lem_Cotlar-Stein_trajectory_hyp_supports} imply that 
    \begin{equation}\label{ineq_lem_Cotlar-Stein_trajectory_support_far}
        \forall j \in J^c_k, \; \forall i \in I_{k+1} \quad \dist(\supp p_{k+1,i}, \supp p_{k,j} \circ A^{-1}_k) \ge C' \epsilon \jap{i-\tilde{f}_k(j)}.
    \end{equation}
    Remark that the functions $(p_{k,i} \circ A^{-1}_k)_{k \le n, i \in I_k}$ are uniformly in $S(1)$ because of \eqref{eq_lem_Cotlar-Stein_trajectory_hyp_unif_S(1)} and the uniform bound on the norm of the matrices $(A_k)_{k \le n}$. We have 
    \begin{equation*}
        \norm{P_{k+1,i} P_{k,j}} = \norm{M_h(A_{k+1}) \Op^w_h(p_{k+1,i}) M_h(A_k) \Op^w_h(p_{k,j})} = \norm{\Op^w_h(p_{k+1,i}) \Op^w_h(p_{k,j} \circ A^{-1}_k)}.
    \end{equation*}
    Therefore the lower bounds on the distance between the support \eqref{ineq_lem_Cotlar-Stein_trajectory_support_close}, \eqref{ineq_lem_Cotlar-Stein_trajectory_support_far} and Lemma \ref{lem_disjoint_support_weyl_quantization_R2d} imply that 
    \begin{equation}\label{ineq_lem_Cotlar-Stein_trajectory_Pk+1i_P_kj_close}
        \forall j \in J_k, \; \forall i \neq \tilde{f}_k(j) \quad \norm{P_{k+1,i} P_{k,j}} \le C_N \frac{h^N}{\jap{i-\tilde{f}_k(j)}^N},
    \end{equation}
    and 
    \begin{equation}\label{ineq_lem_Cotlar-Stein_trajectory_Pk+1i_P_kj_far}
        \forall j \in J^c_k, \; \forall i \in I_{k+1} \quad \norm{P_{k+1,i} P_{k,j}} \le C_N \frac{h^N}{\jap{i-\tilde{f}_k(j)}^N}.
    \end{equation}
    Since the functions $(p_{k,i})_{k \le n, i \in I_k}$ are uniformly in $S(1)$ by \eqref{eq_lem_Cotlar-Stein_trajectory_hyp_unif_S(1)} and satisfy $\norm{p_{k,i}}_\infty \le 1 +C h$, the Calderon-Vaillancourt theorem \ref{Calderon-Vaillancourt_R2d} implies that 
    \begin{equation*}
        \forall k \in \eint{0}{n}, \; \forall i \in I_k \quad \norm{P_{i,k}} \le 1 + \CO(h) \qqtqq{and} \norm{\Op^w_h(p_k)} \le 1 + \CO(h).
    \end{equation*} 
    with the $\CO(h)$ term independent of $k$ and $i$.
    We apply Lemma \ref{lem_Cotlar-Stein_on_steroids} to $(P_{1,i})_{i \in I_1}$ and $(P_{0,j})_{j \in I_0}$ to obtain 
    \begin{equation}\label{eq_lem_Cotlar-Stein_trajectory_first_iteration}
        \sum_{(i,j) \in I_1 \times I_0} P_{1,i} P_{0,j} = \sum_{j \in J_0} P_{1,\tilde{f}_0(j)} P_{0,j} + R_1 \qqtqq{with} \norm{R_1} = \CO(h^\infty).
    \end{equation}
    We iterate this procedure. Denote $Q_j = P_{1,\tilde{f}_0(j)} P_{0,j}$. We have for $i \neq j$, according to \eqref{ineq_lem_Cotlar-Stein_trajectory_Pki_Pkj*} 
    \begin{equation}\label{ineq_lem_Cotlar-Stein_trajectory_Qi_Qj*}
        \norm{Q_i Q_j^*} = \norm{P_{1,\tilde{f}_0(i)} P_{0,i} P^*_{0,j} P_{1,\tilde{f}_0(j)}^*} \le (1+\CO(h))^2 \norm{P_{0,i} P^*_{0,j}} \le C_N' \frac{h^N}{\jap{i-j}^N}.
    \end{equation}
    Denote $L_1 = f_0^{-1}(J_1) \cap J_0$, $f_1 = \tilde{f_0}$ and $f_2 = \tilde{f}_1 \circ \tilde{f}_0$. We get for any $j \in L_1$ and any $i \neq f_0(j)$ thanks to \eqref{ineq_lem_Cotlar-Stein_trajectory_Pk+1i_P_kj_close} that 
    \begin{align}\label{ineq_lem_Cotlar-Stein_trajectory_P2i_Qj}
        \norm{P_{2,i} Q_{j}} &= \norm{P_{2,i} P_{1,\tilde{f}_0(j)} P_{0,j}} \le (1+ \CO(h)) \norm{P_{2,i} P_{1,\tilde{f}_0(j)}} \le C_N' \frac{h^N}{\jap{i - f_2(j)}^N},
    \end{align}
    and for any $j \in J_0 \setminus L_1$ and any $i \in I_2$, thanks to \eqref{ineq_lem_Cotlar-Stein_trajectory_Pk+1i_P_kj_far} that 
    \begin{equation*}
        \norm{P_{2,i} Q_{j}} = \norm{P_{2,i} P_{1,\tilde{f}_0(j)} P_{0,j}} \le (1+ \CO(h)) \norm{P_{2,i} P_{1,\tilde{f}_0(j)}} \le C_N' \frac{h^N}{\jap{i - g_1(f_0(j))}^N}.
    \end{equation*}
    We apply again Lemma \ref{lem_Cotlar-Stein_trajectory} to $(P_{2,i})_{i \in I_2}$ and $(Q_j)_{j \in J_0}$.
    \begin{equation*}
        \sum_{\substack{i \in I_2 \\j \in J_0}} P_{2,i} P_{1,f_0(j)} P_{0,j} = \sum_{j \in L_1} P_{2,f_2(j)} P_{1,f_1(j)} P_{0,j} + R_2 \qqtqq{with} \norm{R_2} = \CO(h^\infty).
    \end{equation*}
    Combining this equation and \eqref{eq_lem_Cotlar-Stein_trajectory_first_iteration}, we get
    \begin{align*}
        \sum_{(i,k,j) \in I_2 \times I_1 \times I_0} P_{2,i} P_{1,k} P_{0,j} &= \sum_{\substack{i \in I_2 \\j \in J_0}} P_{2,i} P_{1,f_0(j)} P_{0,j} + \sum_{i \in I_2} P_{2,i} R_1 \\
        &= \sum_{j \in L_1} P_{2,f_2(j)} P_{1,f_1(j)} P_{0,j} + R_2 + M_h(A_2) \Op^w_h(p_2) R_1.
    \end{align*}
    Define by recurrence $f_k = \tilde{f}_{k-1} \circ \dots \circ \tilde{f}_0$, $f_0 = Id$ and $L_{k} = f^{-1}_{k}(J_k) \cap L_{k-1}$ with $L_0 = J_0$. We obtain by iterating the procedure above
    \begin{equation*}
        \prod_{k=0}^{n-1} \sum_{i_k \in I_k} P_{k,i_k} = \sum_{j \in L_{n-1}} \prod_{k=0}^{n-1} P_{k,f_k(j)} + \sum_{k=1}^{n-1} \left( \prod_{l=k+1}^n M_h(A_l) \Op^w_h(p_l) \right) R_k,
    \end{equation*}
    with $\norm{R_k} = \CO(h^\infty)$ uniformly in $k$. It is uniform since the constants used in Lemma \ref{lem_Cotlar-Stein_trajectory} are only multiplied by a factor $(1+ \CO(h))^2$, as in \eqref{ineq_lem_Cotlar-Stein_trajectory_Qi_Qj*}, or $(1+\CO(h))$, as in \eqref{ineq_lem_Cotlar-Stein_trajectory_P2i_Qj}, at each step to account for the norms of $P_{k,i}$, resulting in a factor at most 
    \begin{equation*}
        (1+\CO(h))^{2\kappa |\log h|} = (1+ \CO(h |\log h|)),
    \end{equation*}
    and such terms can be absorbed in the $h^\infty$ term. Therefore, we have 
    \begin{align*}
        \norm{\sum_{k=1}^{n-1} \left( \prod_{l=k+1}^{n-1} M_h(A_l) \Op^w_h(p_l) \right) R_k} &\le \sum_{k=1}^{n-1} \left( \prod_{l=k+1}^{n-1} \norm{M_h(A_l) \Op^w_h(p_l)} \right) \norm{R_k} \\
        &\le \CO(h^\infty) \sum_{k=1}^{n-1} (1+ \CO(h))^{n-k} \\
        &\le \CO(h^\infty) \kappa |\log h| (1 + \CO(h))^{\kappa |\log h|} \\
        &\le \CO(h^\infty).
    \end{align*}
    This gives the first part of the conclusion. The second item of the conclusion of Lemma \ref{lem_Cotlar-Stein_trajectory} gives 
    \begin{equation*}
        \norm{\sum_{j \in L_{n-1}} \prod_{k=0}^{n-1} P_{k,f_k(j)}} = \sup_{j \in L_{n-1}} \norm{\prod_{k=0}^{n-1} P_{k,f_k(j)}} + \CO(h^\infty),
    \end{equation*}
    which is the second part of the conclusion. The last part of the conclusion comes from the fact that we build $L_k$ by keeping only the indices $i$ such that $f_k(i)$ ends up in $J_k$ at each step $k$, which means that $p_{k+1, f_{k+1}(i)}$ and $p_{k,f_k(i) \circ A^{-1}_k}$ are at a distance at most $\epsilon$.
\end{proof}

\subsection{Applying the results of \texorpdfstring{\cite{anantharamanExponentialDecayProducts2011}}{cite anantharamanExponentialDecayProducts2011}}

Thanks to the reduction done in the previous section, we need to prove a bound of the form 
\begin{equation*}
    \norm{\prod_{k=0}^{n-1} M_h(A_k) \Op^w_h(p_{k})} \le \frac{C}{h^{d'/2}} \exp(-\frac{n}{2} \lambda),
\end{equation*}
where the functions $p_k$ are uniformly in $S(1)$, have a support of diameter at most $\epsilon$ and a supremum norm at most one, i.e. $\norm{p_k}_{\infty} \le 1$. 

To apply the results of \cite{anantharamanExponentialDecayProducts2011}, we first need to center all the symbols $p_k$ so that their support is in a ball centered at $0$. We achieve this via the quantum translations and the relation \eqref{eq_exact_Egorov_quantum_transla}. Let $k \in \eint{0}{n-1}$. Since the diameter of the support of $p_k$ is at most $\epsilon$, there exists $z_k \in \supp(p_k)$ such that the support of $p_k(\cdot - z_k)$ is included in the ball of radius $\epsilon$. The last part of Lemma \ref{lem_Cotlar-Stein_trajectory} implies that the distance between the support of $p_{k+1}$ and $p_k \circ A_k^{-1}$ is at most $\epsilon$, therefore we have
\begin{equation*}
    \norm{z_{k+1} - A_k z_k} \le \norm{z_{k+1} - x_{k+1}} + \norm{x_{k+1} - A_k y_k} + \norm{A_k y_k - A_k z_k } \le 2 \epsilon + M \epsilon,
\end{equation*}
where $x_{k+1}$ and $A_k y_k$ realize the distance between the support of $p_{k+1}$ and the support of $p_k \circ A^{-1}_k$. Introducing quantum translations to center the symbols $p_k$, we obtain 
\begin{align*}
    \prod_{k=0}^{n-1} M_h(A_k) \Op^w_h(p_{k}) = \prod_{k=0}^{n-1} U_h(z_{k+1})M_h(A_k) \Op^w_h(p_{k}) U_h(z_k)^{-1} U_h(z_0),
\end{align*}
with $z_n = 0$ so that $U_h(z_n) = Id$. Using the exact Egorov relations \eqref{eq_exact_Egorov_quantum_transla} and \eqref{eq_exact_Egorov_R2d}, we get 
\begin{align*}
    U_h(z_{k+1})M_h(A_k) \Op^w_h(p_{k}) U_h(z_k)^{-1} &= U_h(z_{k+1})M_h(A_k) U_h(z_k)^{-1} \Op^w_h(p_{k}(\cdot - z_k)) \\
    &= U_h(z_{k+1}) U_h(A_k z_k)^{-1} M_h(A_k) \Op^w_h(p_{k}(\cdot - z_k)) \\
    &= e^{-\frac{i}{2h}\sigma(z_{k+1},A_k z_k)} U_h(z_{k+1} - A_k z_k) M_h(A_k) \Op^w_h(p_{k}(\cdot - z_k)),
\end{align*}
thus, as the phases $e^{i/(2h) \sigma(z_{k+1},A_k z_k)}$ and the unitary operators $U_h(z_0)$ and $U_h(z_n)$ do not change the operator norm, we find 
\begin{align*}
    \norm{\prod_{k=0}^{n-1} M_h(A_k) \Op^w_h(p_{k})} &= \norm{\prod_{k=0}^{n-1} U_h(z_{k+1} - A_k z_k) M_h(A_k) \Op^w_h(p_{k}(\cdot - z_k))} \\
    &= \norm{\prod_{k=0}^{n-1} U_h(z_k') M_h(A_k) \Op^w_h(p_k')},
\end{align*}
with $\norm{z_k'} \le (2+M) \epsilon$ and the support of $p_k'$ included in $B(0,\epsilon)$, the ball centered at $0$ with radius $\epsilon$.

It remains to express the operators $U_h(z_k') M_h(A_k) \Op^w_h(p_k')$ as semiclassical Fourier integral operators and compute the corresponding phases and amplitudes. The Fourier integral operator associated to a phase $S$ and an amplitude $a$, denoted $FIO(S, a)$, is given by 
\begin{equation*}
    \forall u \in \mathscr{S}(\R^d), \quad FIO(S,a) u(x') = \frac{1}{(2\pi h)^d} \int_{\R^{2d}} e^{\frac{i}{h}S(x,x', \theta)} a(x,x',\theta, h) u(x) \diff x \diff \theta.
\end{equation*}

\begin{lemma}\label{lem_expression_FIO_est_disp}
    Let $D \in \GL(d, \R)$, $C \in M_d(\R)$, $z_0 =(x_0, \xi_0) \in \R^{2d}$ and $p \in S(1)$ whose support is included in $B(0,r)$ and such that $\norm{p}_{\infty} \le 1$. Define the block diagonal matrix $A = \begin{pmatrix}
        D & B \\ 0 & D^{-T}
    \end{pmatrix}$. Then 
    \begin{equation*}
        U_h(z_0) M_h(A) \Op^w_h(p) = FIO(S, a) + R,
    \end{equation*}
    where
    \begin{enumerate}
        \item the amplitude $a$ is a smooth function supported in $B(0,3r)$ for $(x, \theta)$ and in $B(-x_0, 3r\norm{A})$ for $x'$ and admits an asymptotic expansion
        \begin{equation*}
            a(x,x',\theta, h) \sim_{h \to 0} |\det D|^{-1/2} \sum_{k =0}^{\infty} h^k a_k(x,x', \theta),
        \end{equation*}
        valid up to any order and in all the $\CC^l$ norms on compact sets. The derivatives of $a_k$ and the constants in the remainder terms of the expansion are controlled by the $S(1)$ semi-norms of $p$. Moreover, the function $a_0$ satisfies $\norm{a_0}_{\infty} \le 1$.
        \item The phase $S(x,x', \theta)$ is given by
        \begin{equation*}
            S(x,x', \theta) = \sprod{x'}{D^{-T} \theta + \xi_0} - \sprod{x}{\theta} - \sprod{x_0}{D^{-T} \theta + \xi_0/2} -\frac{1}{2} \sprod{D^{-1}B \theta}{\theta},
        \end{equation*}
        \item the operator $R$ is bounded on $L^{2}(\R^d)$ and satisfies $\norm{R} = \CO(h^\infty)$ where the constants depend on the $S(1)$ semi-norms of $p$, $r$ and the norm of $A^{-1}$. 
    \end{enumerate}
\end{lemma}

\begin{proof}
    This lemma could be deduced directly from general results about Fourier Integral Operators as the quantum translations, the metaplectic operators and the pseudo-differential operators are Fourier integral operators. The proof presented here is a pedestrian version of the general theory, in which we leverage the simplicity of the phase functions at hand.  
    According to \cite[thm. 11.10]{zworskiSemiclassicalAnalysis2012}, the metaplectic operators are Fourier integral operators and they are given in our case by 
    \begin{equation*}
        \forall u \in \mathscr{S}(\R^{d})_, \quad M_h(A) u(x') = \frac{|\det D|^{-1/2}}{(2 \pi h)^d} \int_{\R^{2d}} e^{\frac{i}{h}(\sprod{D^{-1} x'}{\eta} - \sprod{x}{\eta} - \frac{1}{2} \sprod{D^{-1}B \eta}{\eta})} u(x) \diff x \diff \eta.
    \end{equation*}
    Therefore the product $M_h(A) \Op^w_h(p)$ acts as 
    \begin{align*}
        M_h(A) \Op^w_h(p)u(x') &= \frac{|\det D|^{-1/2}}{(2 \pi h)^{2d}} \int_{\R^{4d}}  e^{\frac{i}{h}(\sprod{D^{-1} x'-x}{\eta} + \sprod{x-y}{\xi} - \frac{1}{2} \sprod{D^{-1}B \eta}{\eta})} p(\frac{x+y}{2}, \xi) u(y) \diff x \diff y \diff \eta \diff \xi \\
        &= \frac{|\det D|^{-1/2}}{(2 \pi h)^{d}} \int_{\R^{2d}}  e^{\frac{i}{h}(\sprod{D^{-1} x'-y}{\eta} - \frac{1}{2} \sprod{D^{-1}B \eta}{\eta})} K(y,\eta) u(y) \diff y \diff \eta ,
    \end{align*}
    where
    \begin{align*}
        K(y, \eta) &= \frac{1}{(2 \pi h)^{d}} \int_{\R^{2d}} e^{\frac{i}{h}(\sprod{x-y}{\xi - \eta})} p(\frac{x+y}{2}, \xi) \diff x \diff \xi \\
        &= \frac{1}{(2 \pi h)^{d}} \int_{\R^{2d}} e^{\frac{i}{h}(\sprod{x}{\xi})} p(\frac{x}{2} + y, \xi + \eta) \diff x \diff \xi.
    \end{align*}

    Fix a smooth function $\chi \in \CC^\infty_c(\R^{2d})$ equal to $1$ on $B(0,r)$ and $0$ on $B(0, 2r)^c$. We decompose the kernel $K$ into two parts 
    \begin{align*}
        K_1(y, \eta) &= \frac{1}{(2 \pi h)^{d}} \int_{\R^{2d}} (1- \chi)(\frac{x}{2}, \xi) e^{\frac{i}{h}(\sprod{x}{\xi})} p(\frac{x}{2} + y, \xi + \eta) \diff x \diff \xi \\
        K_2(y, \eta) &= \frac{1}{(2 \pi h)^{d}} \int_{\R^{2d}} \chi(\frac{x}{2}, \xi) e^{\frac{i}{h}(\sprod{x}{\xi})} p(\frac{x}{2} + y, \xi + \eta) \diff x \diff \xi .
    \end{align*}
    On the one hand, since $1- \chi$ is supported away from $0$, we obtain after successive integrations by parts both in $x$ and $\xi$ that 
    \begin{equation*}
        \norm{K_1}_\infty = \CO(h^\infty),
    \end{equation*}
    where the constants depend on the $S(1)$ norm of $p$. This implies that the operator $R_1$ defined by 
    \begin{equation*}
        R_1 u(x') = \frac{|\det D|^{-1/2}}{(2 \pi h)^{d}} \int_{\R^{2d}}  e^{\frac{i}{h}(\sprod{D^{-1} x'-y}{\eta}) - \frac{1}{2} \sprod{D^{-1}B \eta}{\eta}} K_1(y,\eta) u(y) \diff y \diff \eta ,
    \end{equation*}
    is bounded on $L^2$ and satisfies $\norm{R_1} = \CO(h^\infty)$.
    
    On the other hand, the stationary phase argument \cite[thm. 3.17]{zworskiSemiclassicalAnalysis2012} applied to $K_2$ gives the following asymptotic expansion
    \begin{equation}\label{asympt_exp_proof_expression_FIO_est_disp}
        K_2(y, \eta) \sim_{h \to 0} \sum_{k \ge 0} \frac{h^k}{k!} \left( i\sprod{\partial_x}{\partial_\xi}\right)^k p(\frac{\cdot}{2} + y, \cdot + \eta) (0,0),
    \end{equation}
    valid up to any order and in all the $\CC^l$ norms on compact sets. The constants in the remainder terms of the expansion are controlled by the $S(1)$ semi-norms of $p$, independently of $y$ and $\eta$ as they only translate the function $p$ and thus do not change the $S(1)$ semi-norms.
    As $\chi$ is supported in the ball $B(0,2r)$ and $p$ is supported in the ball $B(0,r)$, the function $K_2$ is supported in the ball $B(0,3r)$. We denote by $T_2$ the operator defined by 
    \begin{equation*}
        T_2 u(x') = \frac{|\det D|^{-1/2}}{(2 \pi h)^{d}} \int_{\R^{2d}}  e^{\frac{i}{h}(\sprod{D^{-1} x'-y}{\eta} - \frac{1}{2} \sprod{D^{-1}B \eta}{\eta})} K_2(y,\eta) u(y) \diff y \diff \eta .
    \end{equation*}

    We add $U_h(z_0)$ to the computation using \eqref{eq_action_quantum_translation_on_functions}.
    \begin{align*}
        U_h(z_0) T_2 u(x') &= e^{\frac{i}{h}\sprod{x'-x_0/2}{\xi_0}} T_2 u (x'-x_0) \\
        &= \frac{e^{\frac{i}{h}\sprod{x'-x_0/2}{\xi_0}}|\det D|^{-1/2}}{(2 \pi h)^{d}} \int_{\R^{2d}}  e^{\frac{i}{h}(\sprod{D^{-1} (x'-x_0)-x}{\theta}- \frac{1}{2} \sprod{D^{-1}B \theta}{\theta})} K_2(x,\theta) u(x) \diff x \diff \theta \\
        &= \frac{|\det D|^{-1/2}}{(2 \pi h)^{d}} \int_{\R^{2d}}  e^{\frac{i}{h}S(x,x', \theta)} K_2(x,\theta) u(x) \diff x \diff \theta,
    \end{align*}
    where $S$ is as claimed in the lemma. It remains to control the support in $x'$ of the amplitude. We fix $\chi' \in \CC^\infty_c(\R^{d})$ a non-negative smooth compactly supported function bounded by $1$, equal to $1$ on $B(-x_0, 2r\norm{A})$ and $0$ on $B(-x_0, 3r\norm{A})^c$. As $M_h(A) \Op^w_h(p) = T_2 + R_1$, we get
    \begin{align*}
        U_h(z_0) M_h(A) \Op^w_h(p) &= \Op^w_h(\chi') U_h(z_0) T_2 + \underset{=R}{\underbrace{\Op^w_h(1- \chi') U_h(z_0) M_h(A) \Op^w_h(p) + U_h(z_0) R_1}}.
    \end{align*}
    Using \eqref{eq_action_quantum_translation_on_functions}, \eqref{eq_exact_Egorov_R2d}, we find 
    \begin{align*}
        R &= U_h(z_0) \Op^w_h(1- \chi'(\cdot - x_0))  M_h(A) \Op^w_h(p) + U_h(z_0) R_1 \\
        &= U_h(z_0) M_h(A) \Op^w_h(1- \chi'(\pi_x(A\cdot) - x_0))  \Op^w_h(p) + U_h(z_0) R_1 ,
    \end{align*}
    where $\pi_x : \R^{2d} \to \R^d$ is the projection on the $x$ coordinate. 
    Our choice for $\chi'$ implies that the support of $1- \chi'(\pi_x(A\cdot) - x_0)$ is included in $B(0,2r)^c$. Indeed, if $z \in B(0,2r)$, we get
    \begin{equation*}
        \norm{(\pi_x(A z) - x_0) + x_0} \le 2r\norm{A} \qqtqq{thus} \chi'((\pi_x(A z) - x_0)) =1.
    \end{equation*}
    As $p$ is supported in $B(0,r)$, we obtain that 
    \begin{equation*}
        \Op^w_h(1- \chi'(\pi_x(A\cdot) - x_0))  \Op^w_h(p) = \CO(h^\infty),
    \end{equation*}
    according to Lemma \ref{lem_disjoint_support_weyl_quantization_R2d} with the implicit constant depending only on $\chi'$, $r$, $\norm{A}$ and the $S(1)$ seminorms of $p$. Therefore, as $U_h(z_0)$ and $M_h(A)$ are unitary, we get
    \begin{align*}
        \norm{R} \le \norm{U_h(z_0) M_h(A) \Op^w_h(1- \chi'(\pi_x(A\cdot) - x_0))  \Op^w_h(p)} + \norm{R_1} = \CO(h^\infty),
    \end{align*}
    where the constants depend on the $S(1)$ semi-norms of $p$, $r$ and the norm of $A$. 
    Recalling that the Weyl quantization of a symbol depending only on the position is just the multiplication by the symbol, we have 
    \begin{align*}
        \Op^w_h(\chi' ) U_h(z_0) T_2 u (x') = \frac{1}{(2 \pi h)^{d}} \int_{\R^{2d}}  e^{\frac{i}{h}S(x,x', \theta)} a(x,x', \theta) u(x) \diff x \diff \theta ,
    \end{align*}
    with 
    \begin{equation*}
        a(x,x',\theta) = \chi'(x') |\det D|^{-1/2} K_2(x,\theta).
    \end{equation*}
    Since $K_2$ is supported in $B(0,3r)$ and $\chi'$ in $B(-x_0, 3r\norm{A})$, the amplitude $a$ is supported in $B(0,3r)$ in $(x, \theta)$ and in $B(-x_0, 3r\norm{A})$. Drawing on the asymptotic expansion of $K_2$ \eqref{asympt_exp_proof_expression_FIO_est_disp}, we have one for $a$
    \begin{equation*}
        a(x,x',\theta) \sim_{h \to 0} |\det D|^{-1/2} \chi'(x') \sum_{k \ge 0} \frac{h^k}{k!} \left( i\sprod{\partial_x}{\partial_\xi}\right)^k p(\frac{\cdot}{2} + x, \cdot + \theta) (0,0),
    \end{equation*}
    valid up to any order and in all the $\CC^l$ norms and the constants in the remainder terms depend only on the $S(1)$ semi-norms of $p$. Furthermore, the first term in the expansion is 
    \begin{equation*}
        a_0(x,x', \theta) = \chi'(x') p(x, \theta),
    \end{equation*}
    and is bounded by one as $\chi'$ is bounded by one and so is $p$.
\end{proof}

We apply this lemma to $U_h(z_k')M_h(A_k) \Op^w_h(p_k')$ and obtain for each $k \in \eint{0}{n-1}$ the following expression
\begin{equation*}
    U_h(z_k')M_h(A_k) \Op^w_h(p_k') = FIO(S_k, a_k) + R_k,
\end{equation*}
where $S_k$ and $a_k$ are smooth functions that satisfy
\begin{itemize}
    \item $S_k(x,x', \theta) = \sprod{x'}{s_k(\theta)} - \sprod{x}{\theta} + \alpha_k(\theta)$ with $s_k(\theta) = D^{-T}_k \theta$. The derivatives of $s_k$, $s^{-1}_k$ and $\alpha_k$ are controlled independently of $k$. 
    \item $a_k$ is supported in $B(0, (5+4M)\epsilon)$ (see details below) and admits an asymptotic expansion 
    \begin{equation*}
        a(x,x',\theta, h) \sim_{h \to 0} |\grad s_k|^{1/2} \sum_{j =0}^{\infty} h^j a_{k,j}(x,x', \theta),
    \end{equation*}
    valid up to any order and in all the $\CC^l$ norms on compact sets. The derivatives of $a_{k,j}$ and the constants in the remainder terms of the expansion are controlled independently of $k$. Moreover, each function $a_{k,0}$ satisfies $\norm{a_{k,0}}_{\infty} \le 1$.
    \item $R_k$ is a bounded operator on $L^2$ and $\norm{R_k} = \CO(h^\infty)$, where the constants are independent of $k$.
\end{itemize}
Every constant in the remainder of the asymptotic expansions or control of the derivatives is independent of $k$ since they depend on the norm of $A_k$ or $A^{-1}_k$ and on the $S(1)$ semi-norms of $p_k'$, which are all controlled independently of $k$. 
We detail the part on the support of $a_k$. The lemma tells us directly that $a_k$ is supported in $B(0, 3\epsilon)$ in $(x, \theta)$ and in $B(-x_{z_k'}, 3 M\epsilon)$ where $x_{z_k'}$ is the position coordinate of $z_k'$. As $\norm{z_k'} \le (2+M) \epsilon$, we have that $B(-x_{z_k'}, 3 M\epsilon) \subset B(0, (2+4M)\epsilon)$ and thus that $a_k$ is supported in $B(0, 3 \epsilon + (2+4M)\epsilon)$. 

The error terms $R_k$ are of order $h^\infty$ and we get rid of them. We have
\begin{align*}
    \norm{\prod_{k=0}^{n-1} U_h(z_k')M_h(A_k) \Op^w_h(p_k')} &= \norm{\prod_{k=0}^{n-1} FIO(S_k, a_k) + R_k} \\
    &\le \norm{\prod_{k=0}^{n-1} FIO(S_k, a_k)} + \sum_{k=0}^{n-1} \norm{\prod_{l=k+1}^{n-1} FIO(S_l,a_l) R_k} \\
    &\le \norm{\prod_{k=0}^{n-1} FIO(S_k, a_k)} + \sum_{k=0}^{n-1}  (1 + \CO(h))^{n-k-1} \CO(h^\infty) \\
    &\le \norm{\prod_{k=0}^{n-1} FIO(S_k, a_k)} + \CO(h^\infty),
\end{align*}
because $n$ is of order $|\log h|$. We used that $\norm{FIO(S_k,a_k)} \le 1 + \CO(h)$ because $\norm{a_{k,0}} \le 1$, and uniformly in $k$ as we have a uniform control of the derivatives.

The last thing to understand before applying \cite[thm. 2.3]{anantharamanExponentialDecayProducts2011}, is the behavior of $\nabla s_{k+m-1} \circ \dots \circ \nabla s_{m}$ for $m \in \eint{0}{n-1}$ and $k \in \eint{0}{n-m}$ as $k$ tends to infinity, which implies that $n$ tends to infinity. This composition is equal to 
\begin{equation*}
    \nabla s_{k+m-1} \circ \dots \circ \nabla s_{m} = D_{k+m-1}^{-T}  \cdots D_{m}^{-T}.
\end{equation*}
Therefore, we get
\begin{equation*}
    \limsup_{k \to \infty} \frac{1}{k} \log \norm{\nabla s_{m+k-1} \circ \dots \circ \nabla s_{m} } = \limsup_{k \to \infty} \frac{1}{k} \log \norm{D_{m+k-1}^{-T}  \cdots D_{m}^{-T}} \le 0,
\end{equation*}
uniformly in $m$ by hypothesis. 

Finally, remark that each $s_k $ can be written as
\begin{equation*}
    s_k(\xi_1, \xi_2) = (D_{k,1}^{-T} \xi_1 + \tilde{B}_k \xi_2, D_{k,2}^{-T} \xi_2) = (g_k(\xi_1, \xi_2), \tilde{s}_k(\xi_2)),
\end{equation*}  
where $\xi_j \in \R^{d_j}$ and the equality defines $g_k$ and $\tilde{s}_k$. Since the phases are linear, we obtain that the gradients
\begin{align*}
    \nabla s_{n-1} \circ \dots \nabla s_{0} (\xi) = D^{-T}_{n-1} \dots D^{-T}_0 \qtq{and} \nabla \tilde{s}_{n-1} \circ \dots \nabla \tilde{s}_{0} (\xi_2) = D^{-T}_{n-1,2} \dots D^{-T}_{0,2},
\end{align*} 
do not depend on the points $\xi \in \R^d$ and $\xi_2 \in \R^{d_2}$.

We apply \cite[thm. 2.3]{anantharamanExponentialDecayProducts2011} and obtain that there exist $h_\kappa >0$ such that for any $h < h_\kappa$ and any $\delta >0$, we have 
\begin{align*}
    \norm{\prod_{k=0}^{n-1} FIO(S_k, a_k)} &\le \frac{C_\delta}{h^{(d_1+\delta)/2}} \frac{|\det D^{-T}_{n-1} \dots D^{-T}_0|^{1/2}}{|\det D^{-T}_{n-1,2} \dots D^{-T}_{0,2}|^{1/2}} (1 + \CO(|\log h|^3 h e^{\delta n})) \\
    &\le \frac{C_\delta}{h^{(d_1+\delta)/2}} \prod_{k=0}^{n-1} |\det(D_{k,1})|^{-1/2} (1 + \CO(|\log h|^3 h^{1- \kappa \delta})).
\end{align*}
All the errors of order $h^\infty$ made along the way do not appear in the final result as they are negligible in comparison to the bound obtained above.

\section{The lower bound for the entropy}\label{sec_full_proofs}

The goal of this section is to prove the two theorems \ref{thm_lower_bound_entropy_projection} and \ref{thm_lower_bound_full_entropy}. Therefore, we fix an integer symplectic matrix $A \in \Sp_{2d}(\Z)$, a sequence of $\theta_N(A)$ such that $M_h(A)$ can be restricted to $\mathcal{H}_N(\theta_N(A))$ and a sequence $\psi_N$ of normed eigenvectors of $M_N(A)$ such that the associated semiclassical measure $\mu$ exists:
\begin{equation*}
    \forall a \in \CC^\infty(\T^{2d}) \quad \mu_N(a) = \sprod{\Op^w_N(a) \psi_N}{\psi_N} \xrightarrow[N \to \infty]{} \mu(a).
\end{equation*}
In the case it is necessary to consider a subsequence, we do it and instantly drop the extraction to simplify the notations. We will not mention the dependence in $\theta$ for the same reason. The $C^*$ algebra of bounded operators on $\mathcal{H}_N(\theta_N(A))$ is denoted $\CA_N$ and the state associated to $\psi_N$ is denoted by $\rho_N$ and given by 
\begin{equation*}
    \forall x \in \CA_N, \quad \rho_N(x) = \sprod{x \psi_N}{\psi_N}.
\end{equation*}
We highlight that the state $\rho_N$ is invariant under $M_N(A)$ as $\psi_N$ is an eigenvector of the unitary operator $M_N(A)$.

We fix two small parameters $\delta>0$ and $\epsilon>0$. The first one will control the width of the partitions considered, while the second one will be involved in the Ehrenfest time, since we require a small margin from the extreme case $S_{1/2}$ to make the semiclassical calculus work. Some conditions will be imposed on them during the proof.

\subsection{A lower bound for the entropy of the projection}\label{subsection_proof_lower_bound_entropy_projection}

We first prove Theorem \ref{thm_lower_bound_entropy_projection} since the proof is a simplified version of the one of Theorem \ref{thm_lower_bound_full_entropy}. Let $V$ be a rational symplectic subspace that is invariant under $A$, $\lambda$ the highest Lyapunov exponent of $A$ restricted to $V$ and $R$ the radius of injectivity of the torus $\tilde{\T} := \T^{2d} / V^\sigma$. Suppose that $\lambda >0$ and impose that $\delta <R$.

We consider $\tilde{\CP} = (\tilde{p}_{i}^2)_{i \in I}$ a smooth partition of $\tilde{\T}$ whose elements $\tilde{p}_i$ have a support of diameter at most $\delta$. 
We obtain from it a partition $\CP = (p_{i}^2)_{i \in I}$ of $\T^{2d} $ by setting $p_{i} = \tilde{p}_{i} \circ \pi$, where $\pi : \T^{2d} \to \tilde{\T}$ is the natural projection. Since each element of $\CP$ is invariant under $V^\sigma$, Lemma \ref{lem_good_Ehren_time_obs_Fourier_support} implies that they have an effective Ehrenfest time of $(T, \epsilon/2)$, where 
\begin{equation*}
    T = \floor{\frac{(1- \epsilon)| \log h|}{2 \lambda}}.
\end{equation*}

As the support of each function $\tilde{p}_{i}$ has diameter at most $\delta$,  the function $p_{i}$ on $\R^{2d}$ is $(R - \delta, \delta)$-splittable by Lemma \ref{lem_splittable_function_by_lifting}.

It remains to ensure that the partition $\CP^q$ is a $(1+\CO(h^\infty))$ quantum partition. This is always possible by slight modifications of the symbols as a consequence of Lemma \ref{lem_modifyin_partition_to_have_1_plus_h_infinity}. Note that the new symbol $p_i'$ given by Lemma \ref{lem_modifyin_partition_to_have_1_plus_h_infinity} has the same support as $p_i$ and converges to $p_i$ when $h$ tends to $0$. In particular, the family of functions $\CP = (p_i')_{i \in I}$ depends on $h$ and is a partition of unity up to an error of order $\CO(h)$. We drop the primes and will continue to denote by $p_i$ the symbols of the classical partition. All in all, we have built a classical partition $\CP$ such that 
\begin{itemize}
    \item each element $p_i$ of $\CP$ is $(R-\delta, \delta)$ splittable and has an effective Ehrenfest time of $(T, \epsilon/2)$,
    \item $\CP^q$ is a $(1+\CO(h^\infty))$ quantum partition.
\end{itemize}

In particular, $\bigvee_{k = -T}^T \CP^q(k)$ is a $(1 + \CO(h^\infty))$ quantum partition. This is a direct consequence of Lemma \ref{lem_joint_approx_quantum_part_is_approx_quantum_part} as 
\begin{equation*}
    (1+ \CO(h^\infty))^{2T+1} \le (1+ \CO(h^\infty))^{|\log h|/ \lambda} = e^{\frac{|\log h|}{\lambda}\log(1 + \CO(h^\infty))} = e^{\CO(h^\infty)} = 1 + \CO(h^\infty).
\end{equation*}

The next step consists in applying the entropic uncertainty principle for approximate quantum partitions \ref{thm_conditional_entropic_uncertainty_principle_approximate quantum_partitions} on $\CA_N$ to the state $\rho_N$ to obtain that
\begin{equation*}
    h^q(\rho_N, M_N(A)^{T+1} \cdot \bigvee_{k = -T}^T \CP^q(k)) + h^q(\rho_N, M_N(A)^{-T} \cdot \bigvee_{k = -T}^T \CP^q(k)),
\end{equation*}
is bounded from below by 
\begin{equation}\label{eq_lower_bound_quantum_entrop_via_entropic_uncertainty_principle}
    - 2 \log \left( \max_{P, Q \in \bigvee_{k = -T}^T \CP^q(k)} \norm{P M_N(A)^{T+1} (M_N(A)^{-T})^* Q^*} + \CO(N^{-\infty}) \right) + \CO(N^{-\infty}).
\end{equation}
We can apply the theorem because $\bigvee_{k = -T}^T \CP^q(k)$ is a $(1 + \CO(N^{-\infty}))$ quantum partition and its cardinality is less than $N^{|\CP|/\lambda}$. 
We simplify the lower bound via the hyperbolic dispersive estimate. Let $P,Q \in \bigvee_{k = -T}^T \CP^q(k)$. There exist two sequences of indices $(i_k)_{k \in \eint{-T}{T}}$ and $(j_k)_{k \in \eint{-T}{T}}$ such that 
\begin{equation*}
    P = \prod_{k=-T}^T \Op^w_N(p_{i_k})(k) \qqtqq{and} Q = \prod_{k=-T}^T \Op^w_N(p_{j_k})(k).
\end{equation*}
Hence,
\begin{align*}
    P M_N(A)^T (M_N(A)^{-T})^* Q^* &= \prod_{k=-T}^T \Op^w_N(p_{i_k})(k) M_N(A)^{2T+1} \prod_{k=-T}^T \Op^w_N(\bar{p}_{j_{-k}})(k) \\
    &= M_N(A)^{-T-1} M_N(A) \Op^w_N(p_{i_T}) \dots M_N(A) \Op^w_N(p_{i_{-T}}) \\
    & \qquad M_N(A) \Op^w_N(\bar{p}_{j_{-T}}) \dots M_N(A) \Op^w_N(\bar{p}_{j_{T}}) M_N(A)^{-T}.
\end{align*}
Thus, as all the $p \in \CP$ are real valued,
\begin{align*}
    \max_{P, Q \in \bigvee_{k = -T}^T \CP^q(k)} \norm{P M_N(A)^{T+1} (M_N(A)^{-T})^* Q^*} &= \max_{(i_k)_{k \in \eint{0}{4T+1}}} \norm{\prod_{k=0}^{4T+1} M_N(A) \Op^w_N(p_{i_k})}_{\CA_N} \\
    &\le \max_{(i_k)_{k \in \eint{0}{4T+1}}} \norm{\prod_{k=0}^{4T+1} M_h(A) \Op^w_h(p_{i_k})}_{\CL(L^2(\R^{d}))}.
\end{align*}
The last inequality is a consequence of \eqref{ineq_bound_norms_pseudo_on H_N_via_L2_norms}. Fix $(i_k)_{k\in \eint{0}{4T+1}} \in I^{4T+2}$. Remark that 
\begin{equation*}
    4T + 1 \simeq \floor{2\frac{(1-\epsilon)|\log h|}{\lambda}} ,
\end{equation*}
and $\lambda$ is the highest Lyapunov exponent of $A_{|V}$. Moreover, each element of the partition $\CP$ is $(R- \delta, \delta)$-splittable and $(R-\delta)/\delta$ tends to $+\infty$ when $\delta$ tends to $0$. Corollary \ref{cor_applied_dispersive_estimate} implies that, for $\delta$ and $h$ small enough,
\begin{equation*}
    \norm{\prod_{k=0}^{4T+1} M_h(A) \Op^w_h(p_{i_k})}_{\CL(L^2(\R^{d}))} \le C_\delta e^{- 2T (\Lambda_0(A_{|V}) - o(\epsilon + \delta))} .
\end{equation*}
To simplify, we have taken $\gamma = \delta$ in the corollary. We also used Remark \ref{rem_cor_applied_dispersive_estimate} since the functions $(p_i)_{i \in I}$ satisfy $p_i = p_{i,0} + h^2 p_{i,1}$ with $\norm{p_{i,0}}_\infty \le 1$. The constant $C_\delta$ depends only on $\epsilon, \delta, A$ and the $S(1)$ seminorms of the $(p_{i_k})_{k \in \eint{0}{4T+1}}$. But each of those symbols is selected among the elements of the partition $\CP$, therefore, the constant $C_\delta$ does not depend on the choice of the sequence $(i_k)_{k\in \eint{0}{4T+1}}$. Hence, 
\begin{equation*}
    - 2 \log \left( \max_{P, Q \in \bigvee_{k = -T}^T \CP^q(k)} \norm{P M_N(A)^{2T+1} Q^*} + \CO(N^{-\infty}) \right) \ge 4T (\Lambda_0(A_{|V}) - o(\epsilon + \delta)) + C_{\delta,\epsilon,A}.
\end{equation*}
Combining this with \eqref{eq_lower_bound_quantum_entrop_via_entropic_uncertainty_principle}, we get 
\begin{equation*}
    h^q(\rho_N, \bigvee_{k = -T}^T \CP^q(k)) \ge 2T (\Lambda_0(A_{|V}) - o(\epsilon + \delta)) - C_{\delta,\epsilon,A}.
\end{equation*}
We used that 
\begin{equation*}
    h^q(\rho_N, M_N(A)^{T+1} \cdot \bigvee_{k = -T}^T \CP^q(k)) = h^q(\rho_N, M_N(A)^{-T} \cdot \bigvee_{k = -T}^T \CP^q(k)) = h^q(\rho_N, \bigvee_{k = -T}^T \CP^q(k)),
\end{equation*}
thanks to \eqref{eq_action_unitary_conditional_quantum_entropy} and the invariance of $\rho_N$ under $M_N(A)$. 
Fix $m_0 \in \N^*$. We then apply Lemma \ref{lem_subadditivity_quantum_part} to obtain 
\begin{equation*}
    h^q(\rho_N, \bigvee_{k = -T}^T \CP^q(k)) \le \sum_{k=0}^{\floor{2T/m_0}} h^q(\rho_N, \bigvee_{k =0}^{m_0-1} \CP^q(k)) + C \le \frac{2T}{m_0} h^q(\rho_N, \bigvee_{k =0}^{m_0-1} \CP^q(k)) + C.
\end{equation*}
This leads to 
\begin{equation*}
    \frac{1}{m_0} h^q(\rho_N, \bigvee_{k =0}^{m_0-1} \CP^q(k)) \ge \Lambda_0(A_{|V}) - o(\epsilon + \delta) -\frac{C}{2T}.
\end{equation*}
We now pass to the limit $N \to +\infty$, in particular $T \to +\infty$, to obtain 
\begin{equation*}
    \frac{1}{m_0} h^c(\mu, \bigvee_{k =0}^{m_0-1} \CP(k)) \ge \Lambda_0(A_{|V}) - o(\epsilon + \delta) .
\end{equation*}
Indeed, we have thanks to the semiclassical calculus \eqref{Moyal_product_R2d}, \eqref{Calderon-Vaillancourt_R2d} and the Egorov theorem \eqref{eq_exact_Egorov_observable_torus} that, for any $p_{i_0},\dots, p_{i_{m_0-1}} \in \CP$,
\begin{equation*}
    \Op^w_N(p_{i_{m_0-1}})(m_0-1)^* \dots \Op^w_N(p_{i_{0}})^* \Op^w_N(p_{i_{0}}) \dots \Op^w_N(p_{i_{m_0-1}})(m_0-1) ,
\end{equation*}
is equal to 
\begin{equation*}
    \Op^w_N(|p_{i_0}|^2 |p_{i_1}\circ A|^2 \dots |p_{i_{m_0-1}}\circ A^{m_0-1}|^2 ) + \CO(N^{-1}).
\end{equation*}
Therefore, 
\begin{align*}
    h^q(\rho_N, \bigvee_{k =0}^{m_0-1} \CP^q(k)) &\xrightarrow[N \to + \infty]{} \sum_{p_{i_0},\dots, p_{i_{m_0-1}} \in \CP} \eta(\mu(|p_{i_0}|^2 |p_{i_1}\circ A|^2 \dots |p_{i_{m_0-1}}\circ A^{m_0-1}|^2 )) \\
    &= h^c(\mu, \bigvee_{k =0}^{m_0-1} \CP(k)).
\end{align*}
Recall that $p_j = \tilde{p}_j \circ \pi$ and $\pi$ and $A$ commute as $V^\sigma$ is invariant under $A$. Therefore,
\begin{equation}\label{eq_proof_lower_bound_proj_entrop_after_h_to_0}
    h^c(\pi_\# \mu, \bigvee_{k =0}^{m_0-1} \tilde{\CP}(k)) = h^c(\mu, \bigvee_{k =0}^{m_0-1} \CP(k)) \ge \Lambda_0(A_{|V}) - o(\epsilon + \delta) .
\end{equation}

This is valid for any smooth partition $\tilde{\CP}$ of $\tilde{\T}$ such that the support of each element of the partition has a diameter less than $\delta$. However we need a partition composed of disjoint indicator functions to compute the Kolmogorov-Sinai entropy. Fix such a partition $\CQ$ of diameter less than $\delta/2$ and such that the boundary of each element of $\CQ$ has null measure for $\pi_\# \mu$ (see \cite[sec. 5]{anantharamanEntropyLocalizationEigenfunctions2008} for a way to construct such a partition). We smooth it by convolution to obtain a partition with smooth functions $\tilde{\CP}$ of diameter less than $\delta$, to which we apply the above reasoning. As the right-hand side of \eqref{eq_proof_lower_bound_proj_entrop_after_h_to_0} does not depend on the derivatives of the partition $\tilde{\CP}$, we can replace it by $\CQ$. 
Hence, if we let $m_0$ tend to $+ \infty$ in \eqref{eq_proof_lower_bound_proj_entrop_after_h_to_0}, we recover the Kolmogorov-Sinai entropy.
\begin{equation*}
    h_{KS}(\pi_\# \mu, A, \CQ) \ge \Lambda_0(A_{|V}) - o(\epsilon + \delta) .
\end{equation*}
We let $\epsilon$ and $\delta$ go to $0$ and use that $h_{KS}(\pi_\# \mu, A)$ is defined as a supremum over all partitions to obtain 
\begin{equation*}
    h_{KS}(\pi_\# \mu, A) \ge \Lambda_0(A_{|V}).
\end{equation*}

\subsection{A lower bound for the full entropy}

We prove in this section Theorem \ref{thm_lower_bound_full_entropy}. We suppose that $\Q^{2d}$ splits into rational symplectic $A$-invariant subspaces $(V_j)_{j \in \eint{1}{r}}$ of dimension $2d_j$
\begin{equation*}
    \Q^{2d} = V_1 \oplus^{\sigma} \dots \oplus^{\sigma} V_r \qqtqq{and} A V_j = V_j .
\end{equation*}
We denote by $\lambda_i$ the maximal Lyapunov exponent of $A$ restricted to $V_i$ and suppose that 
\begin{equation*}
    0 = \lambda_0 < \lambda_1 < \lambda_2 < \cdots < \lambda_{r}.
\end{equation*}

As in the preceding Section \ref{subsection_proof_lower_bound_entropy_projection}, we denote $R_j$ the radius of injectivity of the torus $\tilde{\T_j} := \T^{2d} / V_j^\sigma$ for $j \in \eint{1}{r}$ and suppose that $\delta < \min(R_j)$. For $i\in \eint{1}{r}$, we consider $\tilde{\CP_j}$ a classical partition of $\tilde{\T_j}$ of diameter at most $\delta$ and obtain from it a family of functions $\CP_j$ on $\T^{2d} $ such that 
\begin{enumerate}
    \item each element of $\CP_j$ is invariant under $V^\sigma_j$ and has an effective Ehrenfest time $(T_j, \epsilon/2)$,
    \item $\CP^q_j$ is a $(1+ \CO(h^\infty))$ quantum partition of $L^2(\R^d)$,
\end{enumerate}
where 
\begin{equation*}
    T_j = \floor{\frac{(1- \epsilon)| \log h|}{2 \lambda_j}}.
\end{equation*}
As in the preceding section, the family of function $\CP_j$ are only a partition of unity up to an error of order $\CO(h)$. 
Define $\CQ_j = \bigvee_{i=1}^j \CP_i$. Up to modifying the symbols at order $h$ as in Lemma \ref{lem_modifyin_partition_to_have_1_plus_h_infinity}, we can ensure that $\CQ_j^q$ is a $(1+ \CO(h^\infty))$ quantum partition. Moreover, each element of $\CQ_j$ has an effective Ehrenfest time of $(T_j,\epsilon)$, as a finite product of symbols having at least such an Ehrenfest time. Denote 
\begin{equation*}
    \hat{V}_j = V_1 \oplus^\sigma \dots \oplus^\sigma V_j.
\end{equation*}
It is clear that the elements of $\CQ_j$ are invariant under $\hat{V}_j^\sigma$. Finally, the elements of $\CQ_j$ are $(C_j^{-1}\sum_{i=1}^j R_i - \delta, C_j j \delta)$-splittable, where $C_j$ depends only on $\hat{V}_j$ according to Lemma \ref{lem_product_splitting_function}. 

Fix $j \in \eint{1}{r}$ and let $\CR$ be a $(1+\CO(h^\infty))$ quantum partition. We apply the conditional entropic uncertainty principle \ref{thm_conditional_entropic_uncertainty_principle_approximate quantum_partitions} to 
\begin{equation*}
    h^q(\rho_N, M_N(A)^{T_j+1} \cdot \bigvee_{k=-T_j}^{T_j} \CQ^q_j(k) | \CR) + h^q(\rho_N, M_N(A)^{-T_j} \cdot \bigvee_{k=-T_j}^{T_j} \CQ^q_j(k) | \CR),
\end{equation*}
and find that it is bounded from below by 
\begin{equation*}
    (- 2 \log \left( \max_{P, Q \in \bigvee_{k=-T_j}^{T_j} \CQ^q_j(k)} \norm{P M_N(A)^{T_j+1} (M_N(A)^{-T_j})^* Q^*} + \CO(N^{-\infty}) \right))(1 - \CO(N^{-\infty})).
\end{equation*}
As in the preceding proof, if $\delta$ is small enough and $N$ large, we apply Corollary \ref{cor_applied_dispersive_estimate} and find that 
\begin{equation*}
    - 2 \log \left( \max_{P, Q \in \bigvee_{k=-T_j}^{T_j} \CQ^q_j(k)} \norm{P M_N(A)^{2T_j+1} Q^*} + \CO(N^{-\infty}) \right) \ge 4T_j (\Lambda_0(A_{|\hat{V}_j}) - o(\epsilon + \delta)) + C_{\delta,\epsilon,A}.
\end{equation*}
Therefore, using \eqref{eq_action_unitary_conditional_quantum_entropy} leads to
\begin{equation*}
    h^q(\rho_N, \bigvee_{k=-T_j}^{T_j} \CQ^q_j(k) | \CR(-T_j -1)) + h^q(\rho_N, \bigvee_{k=-T_j}^{T_j} \CQ^q_j(k) | \CR(T_j)) \ge 4T_j (\Lambda_0(A_{|\hat{V}_j}) - o(\epsilon + \delta)) + C_{\delta,\epsilon,A}.
\end{equation*}

We now concentrate on the left-hand side and fix $m_0 \in \N^*$. Applying the subadditivity of conditional entropy lemma \ref{lem_subadditivity_quantum_part} gives for $N$ large enough
\begin{equation*}
    h^q(\rho_N, \bigvee_{k=-T_j}^{T_j} \CQ^q_j(k) | \CR(-T_j -1)) \le \sum_{l=0}^{\floor{2T_j/m_0}} h^q(\rho_N, \bigvee_{k=0}^{m_0-1} \CQ^q_j(k) | \CR(-1 -lm_0)) + C,
\end{equation*}
and
\begin{equation*}
    h^q(\rho_N, \bigvee_{k=-T_j}^{T_j} \CQ^q_j(k) | \CR(T_j)) \le \sum_{l=0}^{\floor{2T_j/m_0}} h^q(\rho_N, \bigvee_{k=0}^{m_0-1} \CQ^q_j(k) | \CR(2T_j-lm_0)) + C.
\end{equation*}
We now choose the partition $\CR$. Recall the number $\nu_{i,j}$ defined for $i \le j$ by 
\begin{equation*}
    \nu_{i,j} = \min(\frac{\lambda_j}{\lambda_i}-1, 1).
\end{equation*}
They satisfy two important inequalities
\begin{equation}\label{ineq_coefficient_nu_i_j}
    (1+\nu_{i,j}) T_j \le T_i \qqtqq{and} 2 \nu_{i,j} T_j \le T_i.
\end{equation}
The first one is a direct consequence of $T_i/T_j = \lambda_j/\lambda_i \ge 1$. For the second one, we have two cases. If $\lambda_j/\lambda_i \le 2$, then 
\begin{equation*}
    2\nu_{i,j} = 2 \frac{\lambda_j}{\lambda_i} - 2 \le \frac{\lambda_j}{\lambda_i}.
\end{equation*}
If $\lambda_j/\lambda_i \ge 2$, then 
\begin{equation*}
    2\nu_{i,j} = 2 \le \frac{\lambda_j}{\lambda_i}.
\end{equation*}
We define $\CR$ by
\begin{equation*}
    \CR = \bigvee_{k=-2\nu_{j-1,j}T_j}^{2\nu_{j-1,j} T_j} \CP^q_{j-1}(k) \vee \dots \vee \bigvee_{k=-2\nu_{1,j}T_j}^{2\nu_{1,j} T_j} \CP^q_1(k).
\end{equation*}
This is indeed a $(1+\CO(h^\infty))$ quantum partition, as it is the refinement of $(1+\CO(h^\infty))$ quantum partitions a number of times of order $|\log h|$. We drop for a moment the superscript $q$ to denote the quantum analog of a classical partition and introduce two notations for the refinement of a quantum partition $\CQ$ and two times $t_1,t_2 \in \Z$,
\begin{equation*}
    \CQ^{(t_2)}_{(t_1)} = \bigvee_{k=t_1}^{t_2} \CQ(k) \qqtqq{and} \CQ^{(t_2)} = \bigvee_{k=0}^{t_2} \CQ(k).
\end{equation*}

We will apply multiple times Lemma \ref{lem_selective_deconditioning_quantum_entropy} to 
\begin{equation*}
    h^q(\rho_N, \CQ^{(m_0-1)}_j | (\CP_{j-1})^{(-1-lm_0+2 \nu_{j-1,j}T_j)}_{(-1-lm_0-2 \nu_{j-1,j}T_j)} \vee \dots \vee (\CP_{1})^{(-1-lm_0+2 \nu_{1,j}T_j)}_{(-1-lm_0-2 \nu_{1,j}T_j)}).
\end{equation*}
We can apply the lemma since the effective Ehrenfest time of the partition $\CP_j$ is $(T_j, \epsilon)$ and thanks to \eqref{ineq_coefficient_nu_i_j}. Indeed, those inequalities imply that for any $i \in \eint{1}{j-1}$  
\begin{equation*}
    |-1-lm_0+2 \nu_{i,j}T_j - (-1-lm_0-2 \nu_{i,j}T_j)| \le 4 \nu_{i,j} T_j + m_0 \le 2 T_i + m_0,
\end{equation*}
and for $k \in \eint{0}{m_0-1}$ and $l \in \eint{0}{\floor{2T_j/m_0}}$,
\begin{equation*}
    |-1-lm_0 \pm 2 \nu_{i,j}T_j - k| \le 2T_j |\frac{lm_0}{2T_j} \pm \nu_{i,j}| + m_0 \le 2T_j |1 + \nu_{i,j}| + m_0 \le 2T_i + m_0.
\end{equation*}
The addition of $m_0$ is not a problem as if a symbol has an effective Ehrenfest time $(T, \epsilon)$ then it also has an effective Ehrenfest time $(T+K, \epsilon)$ for any $K$ independent of $h$. 
Our goal is to keep $\CP_{j-1}^{(m_0-1)}, \dots, \CP_{1}^{(m_0-1)}$ when they appear in the conditioning. For $i \in \eint{1}{j-1}$, we find that $\eint{0}{m_0-1} \subset \eint{-1-lm_0-2 \nu_{i,j}T_j}{-1-lm_0+2 \nu_{i,j}T_j}$ if 
\begin{equation*}
    l \le \frac{2 \nu_{i,j}T_j}{m_0 + 1}.
\end{equation*}
If the piece $\CP_{i}^{(m_0-1)}$ is not present in the conditioning, we remove the whole $(\CP_{i})^{(-1-lm_0+2 \nu_{i,j}T_j)}_{(-1-lm_0-2 \nu_{i,j}T_j)}$. 
We treat first the case of $l \le \frac{2 \nu_{j-1,j}T_j}{m_0 + 1}$. All the pieces $\CP_{j-1}^{(m_0-1)}, \dots, \CP_{1}^{(m_0-1)}$ are present as $\nu_{i,j}$ is decreasing in $i$. One application of Lemma \ref{lem_selective_deconditioning_quantum_entropy} implies that 
\begin{align*}
    h^q(\rho_N, \CQ^{(m_0-1)}_j | (\CP_{j-1})^{(-1-lm_0+2 \nu_{j-1,j}T_j)}_{(-1-lm_0-2 \nu_{j-1,j}T_j)} \vee \dots \vee (\CP_{1})^{(-1-lm_0+2 \nu_{1,j}T_j)}_{(-1-lm_0-2 \nu_{1,j}T_j)}),
\end{align*}
is bounded by 
\begin{equation*}
    h^q(\rho_N, \CQ^{(m_0-1)}_j | \CP^{(m_0-1)}_{j-1} \vee (\CP_{j-2})^{(-1-lm_0+2 \nu_{j-2,j}T_j)}_{(-1-lm_0-2 \nu_{j-2,j}T_j)} \vee \dots \vee (\CP_{1})^{(-1-lm_0+2 \nu_{1,j}T_j)}_{(-1-lm_0-2 \nu_{1,j}T_j)}) + \CO(h^{\epsilon/4}).
\end{equation*}
Multiple applications of the same lemma lead to the following bound 
\begin{equation*}
    h^q(\rho_N, \CQ^{(m_0-1)}_j | \CP^{(m_0-1)}_{j-1} \vee \dots \vee \CP_{1}^{(m_0-1)}) + \CO(h^{\epsilon/4}).
\end{equation*}
If $\frac{2 \nu_{i+1,j}T_j}{m_0 + 1} < l \le \frac{2 \nu_{i,j}T_j}{m_0 + 1}$, we can only keep the pieces $\CP_{i}^{(m_0-1)}, \dots, \CP_{1}^{(m_0-1)}$. This leads to the bound
\begin{equation*}
    h^q(\rho_N, \CQ^{(m_0-1)}_j | \CP^{(m_0-1)}_i \vee \dots \vee \CP_{1}^{(m_0-1)}) + \CO(h^{\epsilon/4}).
\end{equation*}
Finally, grouping everything for the different values of $l \in \eint{0}{\floor{2T_j/m_0}}$ gives that both
\begin{equation}\label{eq_proof_full_entropy_what_to_bound_after_subadditivity}
    \sum_{l=0}^{\floor{2T_j/m_0}} h^q(\rho_N, \CQ^{(m_0-1)}_j | \CR(-1 -lm_0)) \qtq{and} \sum_{l=0}^{\floor{2T_j/m_0}} h^q(\rho_N, \CQ^{(m_0-1)}_j  | \CR(2T_j-lm_0)),
\end{equation}
are bounded by 
\begin{equation*}
    \frac{2T_j}{m_0+1} \sum_{i=0}^{j-1} (\nu_{i,j}-\nu_{i+1,j}) h^q(\rho_N, \CQ^{(m_0-1)}_j | \CP^{(m_0-1)}_i \vee \dots \vee \CP_{1}^{(m_0-1)}) + \CO(h^{\epsilon/5}).
\end{equation*}
We used that $|\log h| = \CO(h^{-\epsilon/4 + \epsilon/5})$ to bound the sum of the errors and remarked that $\nu_{0,j} = 1$ and $\nu_{j,j} = 0$. Note that there is no conditioning if $i=0$ in the sum. The reasoning also works for the second member of \eqref{eq_proof_full_entropy_what_to_bound_after_subadditivity} as $0 \le 2T_j - l m_0 \le 2T_j$ for $l \in \eint{0}{\floor{2T_j/m_0}}$, which suffices to make the argument work.

This leads to the bound 
\begin{equation*}
    \frac{4T_j}{m_0+1} \sum_{i=0}^{j-1} (\nu_{i,j}-\nu_{i+1,j}) h^q(\rho_N, \CQ^{(m_0-1)}_j | \CP^{(m_0-1)}_i \vee \dots \vee \CP_{1}^{(m_0-1)}) \ge 4T_j (\Lambda_0(A_{|\hat{V}_j}) - o(\epsilon + \delta)) + C_{\delta,\epsilon,A} + \CO(h^{\epsilon/4}).
\end{equation*}
Dividing by $4T_j$ and taking the limit $N$ tends to $+ \infty$ gives
\begin{equation*}
    \frac{1}{m_0+1} \sum_{i=0}^{j-1} (\nu_{i,j}-\nu_{i+1,j}) h^c(\mu, \CQ^{(m_0-1)}_j | \CP^{(m_0-1)}_i \vee \dots \vee \CP_{1}^{(m_0-1)}) \ge  \Lambda_0(A_{|\hat{V}_j}) - o(\epsilon + \delta).
\end{equation*}
The partitions are now the classical partitions on $\T^{2d}$. The limit of the left-hand side is obtained in the same manner as in the preceding Section \ref{subsection_proof_lower_bound_entropy_projection}. As in this section, we replace the smooth partition by a partition consisting of disjoint indicator functions. Since we are now in the classical world and everything commutes, we have 
\begin{equation*}
    \CP^{(m_0-1)}_i \vee \dots \vee \CP_{1}^{(m_0-1)} = \CQ^{(m_0-1)}_i \qqtqq{and} \CQ^{(m_0-1)}_j \vee \CQ^{(m_0-1)}_i = \CQ^{(m_0-1)}_{\max(i,j)}.
\end{equation*}
The second point is valid only because the partitions are now composed of indicator functions, which are idempotent for the multiplication of functions. We define 
\begin{equation*}
    B_j = \Lambda_0(A_{|\hat{V}_j}) + \sum_{i=1}^{j-1} (\nu_{i,j}-\nu_{i+1,j}) B_{i},
\end{equation*}
and we claim that 
\begin{equation}\label{ineq_proof_full_entropy_antepenultiem_ineq}
    \frac{1}{m_0+1} h^c(\mu, \CQ^{(m_0-1)}_j ) \ge B_j - o(\epsilon + \delta).
\end{equation}
We prove it via a strong recurrence. It is true for $j=1$. Let $j \in \eint{2}{r}$ and suppose that \eqref{ineq_proof_full_entropy_antepenultiem_ineq} is true for $i <j$. We have 
\begin{equation*}
    \frac{1}{m_0+1} \sum_{i=0}^{j-1} (\nu_{i,j}-\nu_{i+1,j}) h^c(\mu, \CQ^{(m_0-1)}_j | \CQ^{(m_0-1)}_i) \ge  \Lambda_0(A_{|\hat{V}_j}) - o(\epsilon + \delta),
\end{equation*}
and if we add to this inequality $(\nu_{i,j}-\nu_{i+1,j})$ times \eqref{ineq_proof_full_entropy_antepenultiem_ineq} for $i$, we get 
\begin{align*}
    B_j - o(\epsilon + \delta) &\le \frac{1}{m_0+1} \sum_{i=0}^{j-1} (\nu_{i,j}-\nu_{i+1,j}) \left( h^c(\mu, \CQ^{(m_0-1)}_j | \CQ^{(m_0-1)}_i) +  h^c(\mu, \CQ^{(m_0-1)}_i )\right) \\
    &\le \frac{1}{m_0+1} \sum_{i=0}^{j-1} (\nu_{i,j}-\nu_{i+1,j}) h^c(\mu, \CQ^{(m_0-1)}_j \vee \CQ^{(m_0-1)}_i) \\
    &\le \frac{h^c(\mu, \CQ^{(m_0-1)}_j)}{m_0+1} \sum_{i=0}^{j-1} (\nu_{i,j}-\nu_{i+1,j})  \\
    &\le \frac{1}{m_0+1} h^c(\mu, \CQ^{(m_0-1)}_j ),
\end{align*}
as $\nu_{0,j} = 1$ and $\nu_{j,j} = 0$. We let $m_0$ go to infinity, use that $\CQ_j$ is invariant under $\hat{V}_j^{\sigma}$ to obtain 
\begin{equation*}
    h_{KS}(\hat{\pi}_{j\#}\mu, A,\tilde{\CQ_j}) \ge B_j - o(\epsilon + \delta),
\end{equation*}
where $\hat{\pi}_j : \T^{2d} \to \T^{2d}/\hat{V}_j^\sigma$ is the natural projection. Recalling that the Kolmogorov-Sinai entropy is defined as a supremum over all partitions and letting $\epsilon$ and $\delta$ go to $0$ leads to 
\begin{equation*}
    h_{KS}(\hat{\pi}_{j\#} \mu, A) \ge B_j ,
\end{equation*}
for any $j \in \eint{1}{r}$. Remark finally that $B_j>0$ as $\Lambda_0(A_{|\hat{V}_j})>0$ and $\nu_{i,j} > \nu_{i+1,j}$.

\subsection{Proof of Theorem \ref{thm_entropy_bound_split_torus_one_lyap_exp}}

We present below the proof of Theorem \ref{thm_entropy_bound_split_torus_one_lyap_exp} from Theorem \ref{thm_lower_bound_full_entropy}.

\begin{proof}
    To simplify the proof, we suppose that $A$ has no zero Lyapunov exponent. This does not impact the argument below as the zero Lyapunov exponents do not contribute to the entropy and any of the lower bounds we derived. We denote by $\lambda_i$ the Lyapunov exponent of $A$ restricted to $V_i$ and suppose that 
    \begin{equation*}
        \lambda_1 < \lambda_2 < \cdots < \lambda_r.
    \end{equation*} 
    We apply Theorem \ref{thm_lower_bound_full_entropy}, and use its notation, to obtain 
    \begin{equation*}
        h_{KS}(\mu, A) \ge B_r.
    \end{equation*}
    Hence it is sufficient to prove that $B_r \ge \sum_{i=1}^r \frac{d_i\lambda_i}{2}$. Remark that 
    \begin{equation*}
        \Lambda_0(A_{|\hat{V}_j}) = \sum_{i=1}^j d_i \max(\lambda_i - \frac{\lambda_j}{2},0) = \sum_{i=1}^j \frac{d_i \lambda_i}{2} (1- \nu_{i,j}).
    \end{equation*}
    We prove by a strong recurrence that 
    \begin{equation*}
        B_j \ge \sum_{i=1}^j \frac{d_i \lambda_i}{2} .
    \end{equation*}
    This is clearly the case for $j=1$. Fix $j\ge 2$ and suppose that it is true for all $i < j$. We have, as $\nu_{j,j} = 0$, that 
    \begin{align*}
        B_j &= \Lambda_0(A_{|\hat{V}_j}) + \sum_{i=1}^{j-1} (\nu_{i,j}-\nu_{i+1,j}) B_{i} \ge \sum_{i=1}^j \frac{d_i \lambda_i}{2} (1- \nu_{i,j}) + \sum_{i=1}^{j-1} (\nu_{i,j}-\nu_{i+1,j}) \sum_{k=1}^i \frac{d_k \lambda_k}{2} \\
        &\ge \frac{d_j \lambda_j}{2} + \sum_{k=1}^{j-1} \frac{d_k \lambda_k}{2} (1- \nu_{k,j})+ \sum_{k=1}^{j-1} \frac{d_k \lambda_k}{2} \sum_{i=k}^{j-1} (\nu_{i,j}-\nu_{i+1,j}) \\
        &\ge \frac{d_j \lambda_j}{2} + \sum_{k=1}^{j-1} \frac{d_k \lambda_k}{2} (1- \nu_{k,j} + \nu_{k,j} - \nu_{j,j}) \ge \sum_{k=1}^j \frac{d_k \lambda_k}{2} .
    \end{align*}
\end{proof}

\printbibliography[]

\end{document}